\documentclass[letterpaper, DIV=17]{scrartcl}

\usepackage{amssymb, amsthm, mathtools,bbm}
\usepackage[square,sort,comma,numbers]{natbib}
\usepackage{comment}

\usepackage[none]{hyphenat}
\usepackage[english]{babel}
\usepackage{enumerate}
\usepackage{graphicx}%
\usepackage{subcaption}
\usepackage{dsfont}%
\usepackage{float,ulem}%
\usepackage{hyperref}
\usepackage{color}
\usepackage[utf8]{inputenc}
\usepackage{graphicx}
\usepackage{verbatim} 
	\usepackage{hyperref}
	\usepackage{comment}
	\usepackage{xfrac} 
	\usepackage{tensor} 
	\usepackage{listings}  

	\usepackage{pgfplots}              
	\pgfplotsset{compat=newest}
	\usepgfplotslibrary{fillbetween}

	\usepackage[title]{appendix}

\usepackage[shortcuts]{extdash}  
	
	\newtheorem{theorem}{Theorem}[section]  
	\newtheorem{cor}[theorem]{Corollary}
	\newtheorem{prop}[theorem]{Proposition}
	\newtheorem{lemma}[theorem]{Lemma}

	\numberwithin{equation}{section}
	\theoremstyle{definition}
	
	\newtheorem{rem}[theorem]{Remark}

\newcommand{\bxi}{\boldsymbol{\xi}}
\newcommand{\bta}{\boldsymbol{\eta}}

\newcommand{\rr}{{\mathbb{R}}}

\newcommand{\nn}{{\mathbb{N}}}

\newcommand{\tr}{{\operatorname{Tr}\,}}

\newcommand{\indfct}{\operatorname{1}}

\newcommand{\beq}[1]{\begin{equation} \label{#1}}
	\newcommand{\eeq}{\end{equation}}

\renewcommand{\epsilon}{\varepsilon}

\def\be{\begin{equation}}
	\def\ee{\end{equation}} 

\DeclareMathOperator\arctanh{arctanh}

\begin{document}
		\title{The quantum Almeida-Thouless line\\ in the self-overlap-corrected quantum\\ Sherrington-Kirkpatrick model}
		
		\author{Chokri Manai$^{1}$ and Simone Warzel$^{2,3,4}$ \\
			\\
            \small $^1$ Courant Institute, New York University, USA \\[-.5ex]
			\small $^2$ Department of Mathematics, TU Munich, Garching, Germany \\[-.5ex]
			\small $^3$ Munich Center for Quantum Science and Technology, Munich, Germany \\[-.5ex]
			\small $^4$ Department of Physics, TU Munich, Garching, Germany}

        \date{\small \today \\[-.5cm]}
		\maketitle		
		
		\begin{abstract}
			We present a complete analysis of the glass transition in the self-overlap-corrected Sherrington\--Kirkpatrick (SK) model in a transverse magnetic field, also referred to as the quantum SK (QSK) model. In particular, we determine the  phase boundary separating the glassy and paramagnetic phases explicitly. Such an analytic characterization of the glass transition is not expected for the true QSK and even unknown for most vector glass models. Despite being of independent interest, the analysis of the self-overlap corrected QSK model serves as important ingredient in the characterization of paramagnetic behavior in the real QSK.
            The proof is based on a simplified Parisi variational principle for the quantum pressure, which only involves classical Parisi order parameters. As part of the proof, we also analyze the pressure of the self-overlap-constrained quantum SK model and its Parisi description, as well as the pressure of generalized quantum Hopfield models. \\

            \noindent
          {\small \textbf{Keywords:} Quantum spin glass,  Parisi variational formula, phase transition }
		\end{abstract}
	
\bigskip		
\tableofcontents
\bigskip
\section{Motivation}
The Sherrington-Kirkpatrick (SK) model is the most prominent spin-glass model with Ising spins $\pmb{\sigma} = (\sigma_1, \dots , \sigma_N) \in \{-1,1\}^N $. 
The energy of these spins is built from an all-to-all random pair interaction
\begin{equation}\label{eq:SK}
U(\pmb{\sigma}) \coloneqq \frac{1}{ \sqrt{2 N} }   \sum_{j , k =1 }^N  g_{jk }  \, {\sigma}_{j} {\sigma}_{k}
\end{equation}
with couplings $ \big( g_{jk } \big) $ given by independent and identically distributed, standard normal random variables. 
One of the celebrated thermodynamic features~\cite{Parisi:1980aa,Mezard:1986aa} is the emergence of a spin-glass phase associated with replica symmetry breaking below a critical temperature $ T_c = \beta_c^{-1} $, which turns out to be $ \beta_c = 1 $ in the chosen units.  The mathematics developed  for a rigorous description of this phase transition has led to proofs of many of Parisi's predictions starting with Guerra's interpolation bound and Talagrand's proof of Parisi's variational formula for the free energy \cite{Guerra:2003aa,Tal06}, see also \cite{Tal11a,Tal11b,Pan13}. Panchenko established the ultrametricity conjecture and provided a streamlined derivation of Parisi's formula for mixed $p$-spin models \cite{P13,Pan13}. In recent years, a series of works settled properties of the Parisi measure in the glass phase \cite{AC15a, AC15b, ACZ20, Z24, Z25} and the correctness of the so-called TAP approach \cite{C19, CPS20, CPS21}. 

Adding to \eqref{eq:SK} the term $ \sum_{j=1}^N h  {\sigma}_{j}  $, which models a longitudinal magnetic field,  the spin-glass phase is known to disappear at high enough $ h $. Denoting by $ T = \beta^{-1} $ the temperature, the graph in the $ h-T $-plane, which marks this phase transition, has been predicted to coincide with the Almeida-Thouless (AT) line~\cite{AT78}.  The latter is defined through the implicit equation
\begin{equation}\label{eq:AT}
 \beta^2\mathbb{E}_z \frac{1}{\cosh^4(\beta \sqrt{q} z + \beta h)} = 1 , 
\end{equation}
 where $ z$ is a standard normal variable with $ \mathbb{E}_z $ its expectation, and $ q $ is the unique positive solution of $q = \mathbb{E}_z \tanh^2(\beta \sqrt{q} z +\beta  h)$. Toninelli proved the instability of the replica-symmetric solution below the AT line already 20 years ago \cite{Toninelli02}, but the proof of replica-symmetry up to the AT line \eqref{eq:AT} turned out to be a challenge despite being in principle a consequence of the Parisi formula. Several recent works addressed this or related questions \cite{Bolthausen2014,BrYa22, JT16, WKChen21}, and a full proof of the AT conjecture is presented in a recent preprint  \cite{L26} based on a characterization of the phase boundary by Jagannath-Tobasco, who already established the AT line for large enough $\beta$ \cite{JT17}. For the following discussion we recall that $\mathbb{E}_z \frac{1}{\cosh^4(\beta \sqrt{q} z + \beta h)} \sim C \, \beta^{-1} e^{-h^2/2} $ as $\beta \to \infty$ and hence for any fixed $h > 0$ the SK model undergoes a glass transition.\\ 

Taking the quantumness of spins seriously, one may inquire about the behaviour of the glass phase in the presence of a transversal field $ b $. 
In this case, the energy is modeled by the quantum SK (QSK) Hamiltonian
\begin{equation}\label{eq:H}
H_N = \lambda \ U(S_1^z,\dots, S_N^z) \   -  \  b \sum_{j=1}^N S_j^x  \qquad \mbox{on}\quad  \bigotimes_{j=1}^N \mathbbm{C}^2 .
\end{equation}
The first term is the SK energy~\eqref{eq:SK} scaled by $\lambda \in \mathbb{R}$ and evaluated at the jointly diagonalizable $ z $-components of the Pauli matrices. They  act nontrivially only  in the $ j$th tensor component as $S_j^z \coloneqq \mathbbm{1} \otimes \dots \otimes S^z  \otimes  \dots \otimes  \mathbbm{1}  $ with $ \mathbbm{1} $ the identity matrix on $  \mathbbm{C}^2  $ and 
$$ S^x = \left(\begin{matrix} 0 & 1 \\ 1 & 0 \end{matrix} \right)\, , \quad S^y =  \left(\begin{matrix} 0&- i \\ i & 0 \end{matrix}\right)\,, \quad
	S^z = \left(\begin{matrix} 1 & 0 \\ 0 & -1 \end{matrix}\right) 
$$
the three Pauli matrices. The linear operators $ S_j^x $ entering the second term in~\eqref{eq:H} are defined analogously. 
This second term represents the transverse (in the $ x $-direction) magnetic field of constant strength~$ b $. It is one of the striking predictions in the physics literature \cite{FS86,Usadel:1987aa,Yamamoto:1987aa,RCC89,GL90,Young:2017aa,TTC17} that the line, which marks the disappearance of the glass phase in the presence of a transversal field, ends at $T= 0$ in a critical field $b_c > 0$ in the $ b-T$-plane - quite in contrast to its classical analogue, which does not end in a critical point on the $ T = 0 $-axis. Based on the present work, the existence of such a critical $ b_c $ has recently been established \cite{ManaiWarzel2026}. However, there is no conjectural form of an implicit equation of the quantum AT line, that is, the phase boundary which seperates the glass from the paramagnetic phase. Results in the physics litarature are mostly based on numerics~\cite{Young:2017aa,TTC17}. The current rigorous understanding of the phase diagram of the QSK model, including the conjectured QAT line, is also based on upper and lower bounds in~\cite{Leschke:2021aa,Leschke:2021ab} and summarized in Figure~\ref{Figure}. 

If $U$ is replaced by toy models such as the Random Energy Model (REM), which corresponds to the $p \to\infty$ limit of $p$-spin glasses in a transverse field \cite{MW21b, KMW25}, or its hierarchical generalizations, the free energy takes a comparatively simple form, and phase boundaries are explicitly computable \cite{MW21,MW20, MW22}. Even more advanced questions, such as localization-delocalization transitions and trajectory dynamics, can be tackled~\cite{Manai:2023ys, GMW22,MW25b}. In REM-type glasses, a longitudinal field shrinks the glass region in its most correlated sector first, whereas the transversal field begins by changing the frozen region in its least correlated sector.  \\

In the quantum case, the thermodynamic properties of the QSK are encoded in the partition function given by the normalized trace, $$ Z_N \coloneqq \left( 2 \cosh(\beta b) \right)^{-N}\tr e^{-\beta H_N} . $$ 
The normalization is chosen such that in the absence of $ U $, the partition function is one.  A useful mathematical tool for this quantity is the functional integral representation, which is based on the fact that $ S^x  $ is the generator of spin-flips in the $ z $-basis. 

More precisely, let $ \omega_1 $ be a Poisson process on the unit interval with intensity $ \beta b \geq 0 $ and conditioned to have an even number of points in $ [0,1) $, and set
 $ \eta_1(t,\omega_1) := \#\{ \mbox{Poisson points of $ \omega_1 $ in $ [0,t) $} \} $. This induces the paths 
 \begin{equation}\label{paths}
\xi_1 : [0,1) \to \{-1,1\} , \quad
\xi_1(t) := \sigma_1  \ (-1)^{\eta_1(t)} . 
\end{equation} 
with initial condition $ \sigma_1 \in \{-1,1\} $. Choosing the initial condition $ \sigma_1 \in \{-1,1\} $ uniformly distributed and the Poisson process to be independent from $\sigma_1$, we denote the induced probability measure on  $ \xi_1 $ by  $ \nu_1 $. Setting $ \nu_N \coloneqq  \nu_1^{\otimes N} $, the product measure of $ N $ identical copies of the path measure, we hence define a joint probability measure on $ N $  paths  $ \bxi = (\xi_1, \dots , \xi_N ) $. 
Abbreviating the expectation value with respect to this probability measure by $ \int (\cdot)  \nu_N(d\mathbf{\bxi})  $, the functional integral representation of the normalized partition function then takes the Feynman-Kac form
\begin{equation}\label{eq:ZN}
   Z_N =   \left( 2 \cosh(\beta b) \right)^{-N}\tr e^{-\beta H_N} =  \int \exp\left( -\int_0^1 \beta \lambda U\left(\bxi(t)\right) dt  \right) \nu_N(d\mathbf{\bxi}) ,
\end{equation}
see e.g.~\cite{Leschke:2021aa} for a derivation.  From the point of view of the theory of spin glasses, the QSK thus takes the form of a continuum limit of vector spin glasses with the paths $ \bxi $ taking the role of the vectors  \cite{Pan18,Chen23,AdBr20}.

Theoretically, the information about the quantum AT line is buried in the quantum Parisi formula for its free energy, which was derived in~\cite{MW25}. For its formulation, we introduce the 
quantum Parisi functional 
\begin{equation}\label{eq:Parisifctl}
	\mathcal{P}_\lambda\left(\pi, x \right) \coloneqq X_0(\pi,x) + \frac{\beta^2\lambda^2}{4} \int_0^1 \left\| \pi(m) \right\|^2_2 \ dm , 
\end{equation}
on the set of paths $$ \Pi \coloneqq \{\pi: [0,1] \to \mathcal{S}_2^+ \ | \ \pi \, \mbox{monotone non-decreasing and right-continuous} \} $$ taking values in the space $ \mathcal{S}_2^+ $ of non-negative Hilbert-Schmidt operators on the real Hilbert space $ L^2(0,1)$. The latter is equipped with the real scalar product
\begin{equation}\label{eq:scalarprod}
\langle f , g \rangle \coloneqq \int_0^1 f(t) g(t) dt , \qquad f,g \in L^2(0,1) . 
\end{equation}
The second entry $ x \in \mathcal{S}_2 $ in~\eqref{eq:Parisifctl} is again a Hilbert-Schmidt operator. By the relation $ (xf)(t) = \int_0^1 x(t,s) f(s) ds $ for all $ f \in L^2(0,1) $, it may hence be identified with a square-integrable kernel $ x: [0,1)^2 \to \mathbbm{R} $. 
Its Hilbert-Schmidt norm is abbreviated by 
$$ 
\| x \|_2 \coloneqq \left( \int_0^1 \int_0^1 x(t,s)^2 dt ds\right)^{1/2} ,
$$
and $ \langle x , y \rangle \coloneqq \int_0^1 \int_0^1  x(t,s) y(t,s) ds dt $ is the corresponding scalar product. 
The functional 
$ X_0(\pi,x) $ is most easily defined recursively in case $ \pi :[0,1] \to \mathcal{S}_2^+ $  is a monotone step function 
with exactly $ r $ steps. In this case, one picks $r$  independent Gaussian random functions $w_1^{\pi}, \ldots, w_r^{\pi}$ in $ L^2(0,1) $, which are uniquely characterized by their mean and covariance,
\begin{equation}
\mathbb{E}\left[w_j^{\pi}(s)\right] = 0 , \qquad \mathbb{E}\left[w_j^{\pi}(s) w_j^{\pi}(t)\right] =  \pi(m_{j})(s,t) - \pi(m_{j-1})(s,t)  ,
\end{equation}
and recursively defines the real random functionals:
\begin{align*}
X_r(\pi,x) \  \coloneqq & \ln \int \exp\left(\sum_{j=1}^{r}\beta \lambda \langle \xi, w_j^{\pi} \rangle  + \left\langle \xi  , \left( x- \tfrac{\beta^2\lambda^2}{2} \pi(1) \right)  \xi  \right\rangle  \right) \nu_1(d\xi) \\
& \qquad X_{j-1}(\pi,x) \  \coloneqq\frac{1}{m_j} \ln \mathbb{E}_{j}\!\left[ e^{m_{j} X_{j}(\pi,x)}\right]  .
\end{align*}
The Parisi functional is Lipschitz continuous \cite[Prop.~1.5]{MW25} in the sense that there is some $ L_\lambda \in (0,\infty) $ such that for all $ \pi , \pi' \in \Pi $ and $ x , x' \in \mathcal{S}_2 $: 
    \begin{equation}\label{eq:Lipschitz}
        \left| \mathcal{P}_\lambda(\pi, x) - \mathcal{P}_\lambda(\pi', x')\right| \leq L_\lambda \left[ \int_0^1 \left\| \pi(m) - \pi'(m) \right\|_2 dm + \| x - x' \|_2 \right] .
    \end{equation}
Hence $ \mathcal{P}_\lambda $  may be uniquely extended from step functions to arbitrary $\pi\in\Pi$.  
The quantum Parisi formula then reads as follows \cite[Thm.~2.1]{MW25}
\begin{equation}\label{eq:QParisi}
\lim_{N\to \infty} N^{-1} \mathbb{E}\left[\ln Z_N\right] =  \sup_{x \in \mathcal{S}_2^+} \left[ \inf_{\pi \in \Pi}  \mathcal{P}_\lambda\left(\pi, x \right) - \frac{1}{\beta^2\lambda^2} \| x \|_2^2  \right] . 
\end{equation}
The variational problem involves an outer supremum. This is a common feature of the Parisi formula in all vector-spin glasses. In these problems, the self-overlap, which models correlations of the same spin but different times in the unit interval, is generally non-trivial. A characterization of the symmetry properties of the self-overlap is an interesting, but challenging problem. In the Potts spin glasses, the color symmetry is broken for models with a high enough number of species and the problem remains unsettled for a low number of colors \cite{Mourrat2025Potts, Chen25}.  In the QSK, it is believed that this self-overlap is invariant under time shift. If correct, then the outer variational problem may be restricted to $ x \in \mathcal{S}_2^+$ that share this shift invariance. In this case, it was shown in~\cite[Thm.~2.6]{MW25} that the inner variational problem may be restricted to 'classical' paths. This reduction can be seen as analog of the main result in \cite{BS24} on Potts glasses - the proofs, however, are completely different.
Despite this simplification, determining the quantum AT line is still beyond reach. In fact, even the annealed solution, which at $ \lambda = 1 $ is given by the variational problem $  \sup_{x \in \mathcal{S}_2^+} \left[  \mathcal{P}_1\left(0, x \right) - \frac{1}{\beta^2} \| x \|_2^2  \right] $ does not possess an explicit solution. It was, however, proved in~\cite{Leschke:2021aa} that this annealed solution agrees with~\eqref{eq:QParisi} for all $ \beta < 1 $ - hence extending the result in~\cite{ALR87} to $ b > 0 $ and proving a bound on the replica symmetric region, see~Figure~\ref{Figure}.

\section{Self-overlap-corrected quantum Sherrington-Kirkpatrick model}
\subsection{Parisi formula with classical paths}

We look at a simplification to study the quantum AT line. Instead of $ Z_N $ from~\eqref{eq:ZN}, we set  $ \lambda = 1 $, which may be done without loss of generality, and investigate the so-called self-overlap-corrected path integral
\begin{equation}\label{eq:ZNSC}
\widehat Z_N \coloneqq \int \exp\left( -\int_0^1 \beta U\left(\bxi(t)\right) dt  - \frac{\beta^2}{2} \int_0^1 \int_0^1  \mathbb{E}\left[U\left(\bxi(s)\right)  U\left(\bxi(t)\right) \right] ds dt \right) \nu_N(d\mathbf{\bxi}) ,
\end{equation}
and its corresponding random probability  measure 
\begin{equation}\label{eq:probmeas}
 \langle (\cdot) \rangle_N \coloneqq \widehat Z_N^{-1}   \int (\cdot) \ \exp\left( -\int_0^1 \beta U\left(\bxi(t)\right) dt  - \frac{\beta^2}{2} \int_0^1 \int_0^1  \mathbb{E}\left[U\left(\bxi(s)\right)  U\left(\bxi(t)\right) \right] ds dt \right) \nu_N(d\mathbf{\bxi}) . 
 \end{equation}
  The difference between the partition function~\eqref{eq:ZN}  of the QSK and~\eqref{eq:ZNSC} lies in the term
 $$
 \int_0^1 \int_0^1  \mathbb{E}\left[U\left(\bxi(s)\right)  U\left(\bxi(t)\right) \right] ds dt = \frac{N}{2} \left\| Q_N[\bxi] \right\|_2^2 
 $$
 which is proportional to the square of the Hilbert-Schmidt norm of the self-overlap
 \begin{equation}\label{def:SO}
 Q_N[\bxi](t,s) \coloneqq \frac{1}{N} \sum_{j=1}^N \xi_j(t) \xi_j(s) . 
 \end{equation}
 The latter is known to concentrate \cite{MW25} under the random probability measure~\eqref{eq:probmeas}, when averaged over the randomness. The study of self-overlap corrected vector spin glasses goes back to \cite{Chen23}, where the simplification caused by the counter-term were exploited to give novel proof of the Parisi formula via a Hamilton-Jacobi approach. The Hamilton-Jacobi perspective lead to further novel insights - most importantly to a characterization of the free energy of vector glasses with non-convex interactions for which the Parisi variational formula in general fails \cite{Mou21, Mourrat2021, ChenMourrat2025}.  For us, the study of the self-overlap corrected QSK model has two main motivations. First, the self-overlap corrected QSK is the closest cousin of the QSK for which complete  analytic characterization of the phase diagram is feasible. The phase boundary shares all main characteristics of the predicted quantum AT line of the real QSK (see Figure~\ref{Figure}). This is in stark contrast to previously studied caricatures, such as the QREM, which showcase a partially different physics \cite{MW20,MW21, MW22}. Secondly, our subsequent results on self-overlap constraint QSK naturally play a vital role in understanding the quantum phase transition of the QSK. This has already been demonstrated in a subsequent work \cite{ManaiWarzel2026}, which addresses the QSK's paramagnetic phase. 
 
   The following proposition is a strengthening of \cite[Cor.~2.4]{MW25}, which in turn generalizes the result~\cite{Chen24} for self-overlap-corrected vector spin glasses. 
 \begin{prop}\label{prop:conc} 
 \begin{enumerate}
 \item
The functional 
$$\mathcal{S}_2 \ni x \mapsto \mathcal{F}(x) \coloneqq \inf_{\pi \in \Pi } \mathcal{P}_1(\pi, x) $$ 
is convex and Gateaux differentiable, and the  
averaged self-overlap $ m_N \coloneqq \mathbb{E}\left[\langle Q_N \rangle_{N}\right]  $ converges in Hilbert-Schmidt norm to the derivative $ \varrho \coloneqq \nabla \mathcal{F}(0) $:
\begin{equation}\label{eq:strongHS}
    \lim_{N\to \infty} \left\|  m_N - \varrho   \right\|_2 = 0 . 
\end{equation}
\item
Under the product of the disorder probability measure and~\eqref{eq:probmeas},  the self-overlap asymptotically concentrates exponentially on its mean $ m_N $ with respect to the Hilbert-Schmidt norm 
 in the sense that for any $ \varepsilon > 0 $ there is some $ c(\varepsilon) > 0 $ such that for all $N$ large enough: 
\begin{equation}\label{eq:expconc}
	\mathbb{E}\left[\big\langle 1[\|Q_N - \varrho  \|_2 > \varepsilon ] \big\rangle_{N} \right] \leq e^{-c(\varepsilon) N}  . 
\end{equation}
\end{enumerate}
\end{prop}
Here and in the following, $ 1[\dots] $ denotes the indicator function. 
Since $ \| Q_N[\bxi] \|_2 \leq \tr Q_N[\bxi]  = 1  $, the strong convergence~\eqref{eq:strongHS} and the exponential concentration~\eqref{eq:expconc} in particular imply
\begin{equation}\label{eq:conc1}
	\lim_{N\to \infty} \mathbb{E}\left[\big\langle\left\| Q_N - m_N  \right\|_2  \big\rangle_{N} \right] = 0  ,
\end{equation}
which converts the weak concentration result in~\cite[Cor.~2.4]{MW25} to a strong convergence. 
A proof of Proposition~\ref{prop:conc} is found in Appendix~\ref{app:conc2}.  Switching off $ U $, under the a priori path measure $ \nu_N $ the typical self-overlap is the Hilbert-Schmidt operator $ \mu $ given by the kernel
\begin{equation}\label{eq:self-overlap}
 \mu(t,s) \coloneqq  \int  \xi(t)  \xi(s) \nu_1(d\xi)  = \frac{\cosh\left(\beta b ( 1 - 2 | t-s|) \right) }{\cosh \beta b} . 
 \end{equation}
 The explicit expression for the average follows from a computation, which can be found in~\cite{Leschke:2021aa}. Generally, $ \varrho $ differs from $ \mu $, but we will show as part of our main result, Theorem~\ref{thm:main} below, that they coincide in the entire replica-symmetric regime.  \\

The quantum Parisi functional for the self-overlap-corrected SK model simplifies dramatically. Its variational problem reduces from Hilbert-Schmidt-valued paths to classical paths. Moreover, as in its classical analogue \cite{AC15a}, there is a unique minimizing path within the paths
$$ \Pi(\varrho) \coloneqq \left\{ \pi \in \Pi \ | \ \pi(1) = \varrho \right\} 
$$
terminating at the limiting value of the self-overlap, $ \varrho = \nabla \mathcal{F}(0) $. 
\begin{theorem}\label{thm:QPsoc}
For the self-overlap-corrected QSK:
\begin{equation}\label{QPsoc}
  \lim_{N\to \infty} \frac{1}{N}\ \mathbb{E}\left[\ln \widehat Z_N \right] = \mathcal{F}\left(0\right) = \inf_{\pi \in \Pi(\varrho)}  \mathcal{P}_1\left(\pi,0\right) =  \inf_{\pi\in\Pi_c( \varrho )} \mathcal{P}_1\left(\pi,0\right)   = \inf_{\pi\in\Pi_c} \mathcal{P}_1\left(\pi,0\right) 
 \end{equation}
  where $ \varrho = \nabla \mathcal{F}(0) $ and 
  \begin{align*}
& \Pi_c   \coloneqq \left\{\pi \in \Pi  \ \Big| \ \begin{array}{l} 
\mbox{there is a monotone $ \pi_c : [0,1) \to [0,\infty) $ such}  \\
\mbox{that for Lebesgue a.e.\ $ s,t \in [0,1)$:} 
\end{array} \qquad \pi(m)(t,s) = \pi_c(m)  \right\} , \\
& \Pi_c(\varrho)  \coloneqq \left\{\pi \in \Pi_c \ | \ \mbox{for Lebesgue a.e.\ $ s,t \in [0,1)$:} \;\;  \pi(1)(s,t) = \langle 1 ,\varrho 1 \rangle \,   \right\}  . 
\end{align*}
Here, $ 1 \in L^2(0,1) $ stands for the constant function one.\\
There is a unique minimizer $\pi^{*}$ of $ \inf_{\pi \in \Pi_c}  \mathcal{P}_1\left(\pi,0\right)$, and $\tilde{\pi}^{*}(m) := \pi^{*}(m) \mathbbm{1}_{[0,1)}(m) + \varrho \mathbbm{1}_{\{1\}}(m) \in  \Pi_c(\varrho) $  is the unique minimizer among all paths $\pi \in \Pi(\varrho)$ with fixed endpoint.
\end{theorem}
The proof of Theorem~\ref{thm:QPsoc} is found in Subsection~\ref{sec:ProofParisi}. 
The variational characterization~\eqref{QPsoc} strengthens \cite[Thm.~2.6]{MW25}. In the self-overlap-corrected case, it completely settles the conjectured classical structure of the Parisi order parameter $ \pi $ for the QSK: the optimal path does not distinguish times $ s,t \in [0,1) $, which reflects the time-translation symmetry modulo one of the integrand in~\eqref{eq:ZNSC}. Note that the uniqueness holds not only among all classical paths and is thus stronger than currently known results for Potts models restricted to symmetric self-overlaps \cite{Chen25}.

\subsection{Quantum Almeida-Thouless line}
The annealed functional of the self-overlap-corrected QSK, which corresponds to the trivial path $ \pi_c = 0 $, also simplifies due to the normalization $ \mathbb{E}[ \widehat Z_N ]= 1 $:
\begin{equation}
     0 =   \mathcal{P}_1(0,0)  = \ln \mathbb{E}\left[ \widehat Z_N \right] . 
 \end{equation} 
 The quantum AT line hence boils down to the $b-T $-values that separate the annealed solution from a non-trivial value of $ \pi_c $ in~\eqref{QPsoc}. 
The following main result shows that this occurs at 
$$ \displaystyle T(b) \coloneqq \frac{b}{\arctanh(b)} , 
$$
with $ b \in [0,1) $ and endpoint at $ T(1) \coloneqq 0 $, cf.~Figure~\ref{Figure}.  The theorem also provides a complete description of the behavior of the replica overlap
$$
 R_N[\bxi,\bxi'](t,s) \coloneqq \frac{1}{N} \sum_{j=1}^N \xi_j(t) \xi_j'(s) ,
$$
which involves two independent copies of paths $ \bxi, \bxi' $. The onset of replica symmetry breaking is signaled by the expectation value $ \langle \langle \| R_N \|_2^2 \rangle \rangle_N  $ of the replica overlap's Hilbert-Schmidt norm under the duplication of the Gibbs measure, i.e., the product of two independent copies of the random probability measure in~\eqref{eq:probmeas}.
\begin{theorem}\label{thm:main}
\begin{enumerate}
    \item  \label{2}  
    If $  b \geq \tanh(\beta b) $, then the pressure agrees with the annealed pressure, 
\begin{equation}\label{eq:part1}
 \lim_{N\to \infty} \frac{1}{N} \ \mathbb{E} \left[ \ln  \widehat Z_N\right] = 0 .
\end{equation}
In this regime, the replica overlap vanishes, 
\begin{equation}\label{eq:noreplica}
	\lim_{N\to \infty} \mathbb{E}\left[ \langle\langle \| R_N \|_2^2 \rangle\rangle_N \right] = 0 , 
\end{equation}
and $ \varrho = \mu  $, i.e., the self-overlap concentrates on the self-overlap~\eqref{eq:self-overlap} of the a priori measure. 
    \item \label{1}  If $   b < \tanh(\beta b)  $, then the annealed solution is unstable: 
\begin{equation}\label{eq:part2} \lim_{N\to \infty}\  \frac{1}{N} \, \mathbb{E}\left[ \ln \widehat Z_N \right] < 0 =   \lim_{N\to \infty}\  \frac{1}{N} \, \ln \mathbb{E}\left[ \widehat Z_N \right] .
\end{equation}
In this glass phase, the replica overlap is not self-averaging, 
\begin{equation}\label{eq:notselfav}
	\liminf_{N\to \infty} \mathbb{E}\left[ \langle\langle \| R_N \|_2^2 \rangle\rangle_N \right] > 0.
\end{equation}
\end{enumerate}
\end{theorem}
The proof of this main result rests on a combination of techniques: 1.~A conditional second moment technique is applied to the self-overlap-constrained model discussed below.  This part relates to a quantum Hopfield model discussed in Section~\ref{sec:ProofPart1}, which also contains the proof of the first part of Theorem~\ref{thm:main}. 2.~A quantum Toninelli argument is used for the second part of the proof of Theorem~\ref{thm:main}, which is found in Section~\ref{sec:proofpart2}. 3. Statistical mechanics methods based on a Parisi formula for the self-overlap-constraint model are developed for the determination of the order parameter.

Theorem~\ref{thm:main}  yields a complete description of the phase diagram in the self-overlap-corrected QSK. Since we have information about the order parameter in all phases, including the critical line, the result even outmatches what is currently known for the classical SK model in a longitudinal magnetic field \cite{L26}. 

The phase diagram predicted by Theorem~\ref{thm:main} should be compared to that of the original QSK model. The quantum AT lines are expected to look similar, and, in particular, end at critical points on the $ T=0 $ axis, which, however, are expected to be different, see Figure~\ref{Figure}. 
In contrast to the self-overlap-corrected model, the quenched and annealed solutions may differ  in the paramagnet regime of the original QSK~\cite{Leschke:2021aa}.  

\begin{figure}[t]
    \centering
    \begin{subfigure}{0.45\textwidth}
        \centering
    \includegraphics[width=\linewidth]{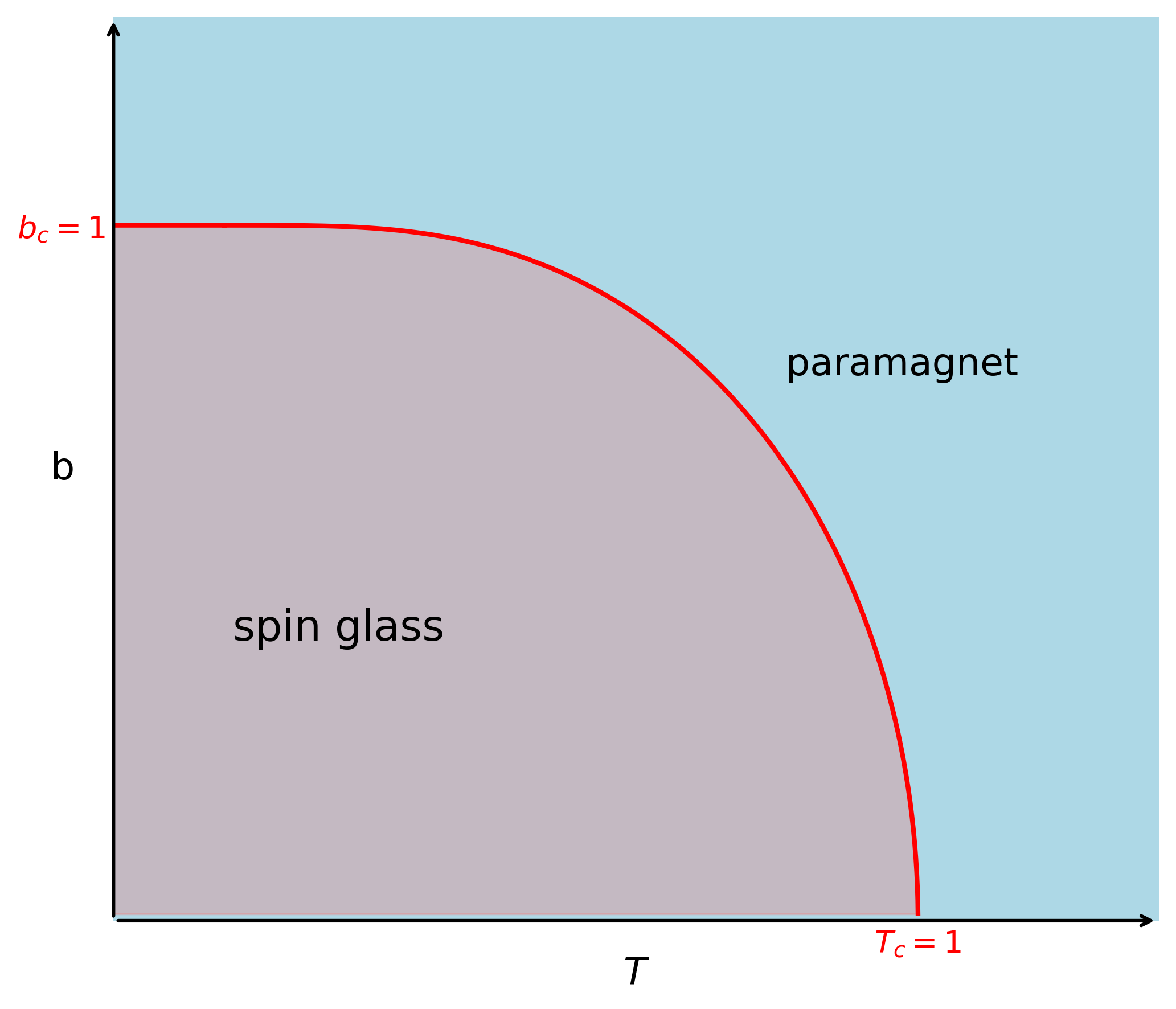}
    \end{subfigure}
    \hfill
    \begin{subfigure}{0.45\textwidth}
        \centering
        \includegraphics[width=\linewidth, height=7cm]{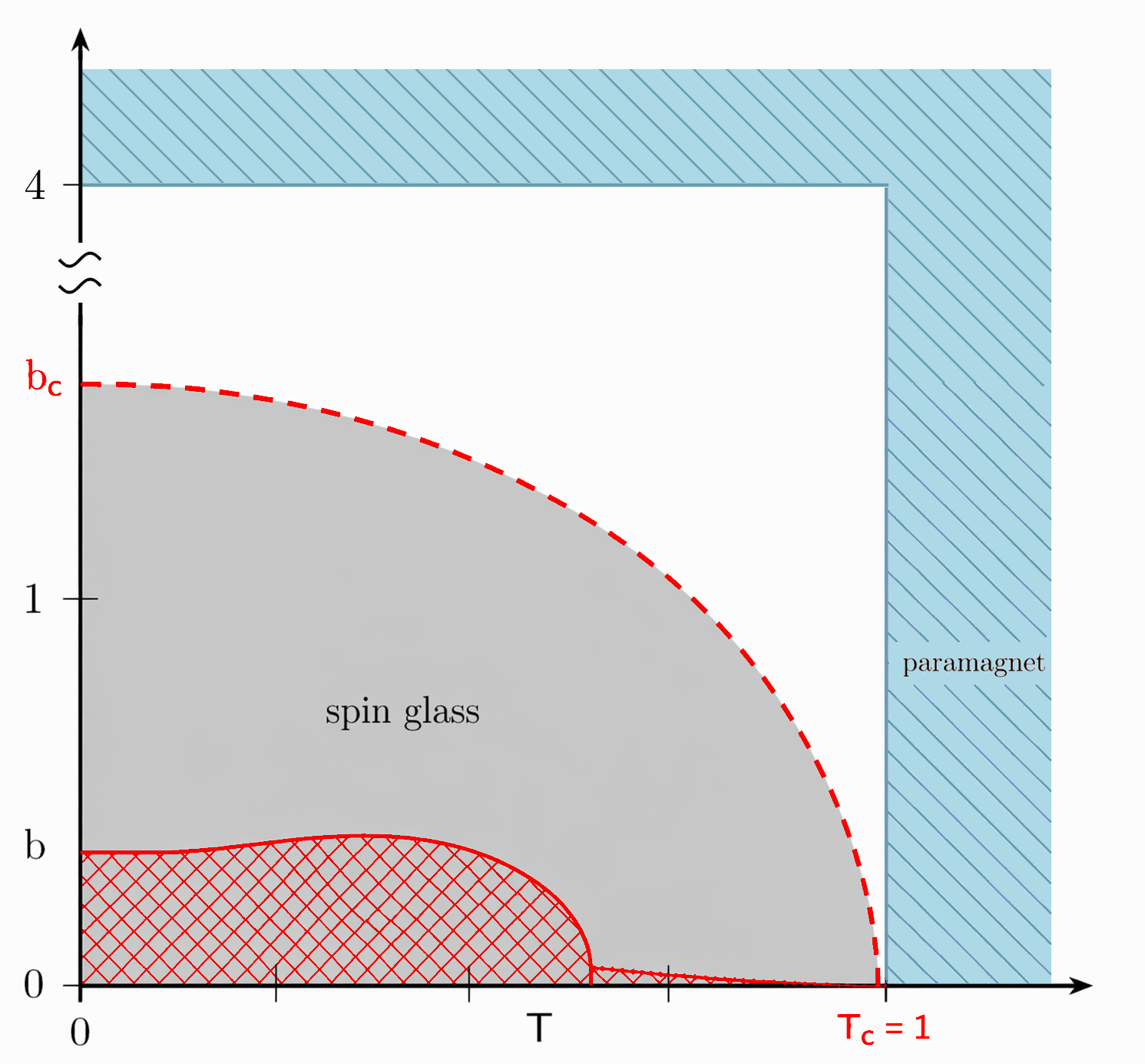}
    \end{subfigure}
    \caption{The left figure shows the quantum AT line $T(b) = b/\arctanh(b)$ in the self-overlap-corrected QSK.
In the region below the line, the annealed solution is unstable and the replica overlap is non-zero.
In the region above the line, the annealed solution agrees with the full variational problem~\eqref{QPsoc} and the replica overlap concentrates around $0$. The figure on the right combines the results from~\cite{Leschke:2021aa,Leschke:2021ab,ManaiWarzel2026} and displays the phase diagram of the original QSK model (at $ \lambda = 1 $). The dotted red line indicates the QAT line based on numerical computations. Its qualitative nature is remarkably similar to that in the self-overlap-corrected model. In the blue region $T \geq 1$ and $ b\geq 4 $, replica symmetry has been established, whereas the persistence of the glass phase has been shown in the marked red region.  }
    \label{Figure}
\end{figure}

\subsection{Self-overlap-constrained QSK}\label{sec:RScon}
The self-overlap-corrected QSK is intimately related to the self-overlap-constrained QSK. For the latter, one picks a target self-overlap $ \upsilon \in \mathcal{S}_2^+ $ and constrains the path measure $ \nu_N $ to the set  
 \begin{equation}\label{eq:constrain}
\mathcal{Q}_{N,\varepsilon}(\upsilon) \coloneqq \left\{ \bxi \ | \ \left\| Q_N[\bxi] - \upsilon \right\|_2 \leq \varepsilon \right\}  .
 \end{equation}
To deduce information about the order parameter from thermodynamic information, we study the following family of constraint partition functions
\begin{equation}\label{def:soconstrp}
    Z_{N,\varepsilon}(\lambda;\upsilon) \coloneqq 
    \int_{\mathcal{Q}_{N,\varepsilon}(\upsilon) } \mkern-10mu \exp\left( -  \int_0^1 \lambda \beta U\left(\bxi(t)\right) dt  \right) \nu_N(d\mathbf{\bxi}) 
\end{equation}
with parameter $ \lambda\in \mathbb{R} $.  This quantity still depends on $ \beta b $, which we suppress in the notation.\\

Self-overlap-constrained models as in~\eqref{def:soconstrp} have been considered in Panchenko's proof of the Parisi formula for vector-spin glasses \cite{Pan18}. In our derivation of a Parisi formula for the QSK and related models~\cite{MW25}, they did not appear. In fact, the proof in~\cite{Pan18}  relies on a covering argument based on compactness to fix the self-overlap, which does not apply to infinite-dimensional vector spin glasses such as the QSK. For a fixed  $\upsilon \in \mathcal{S}_2^{+}$ and any $ \lambda \in \mathbb{R} $, one may nevertheless show that the shifted Legendre transform
\begin{equation}\label{def:Parisicon}
         F(\lambda;\upsilon) \coloneqq \inf_{x\in \mathcal{S}_2} \left[  \inf_{\pi \in \Pi(\upsilon)}\mathcal{P}_\lambda(\pi, x) - \langle x , \upsilon\rangle \right] + \frac{\lambda^2\beta^2}{4} \left\| \upsilon\right\|_2^2 .  
     \end{equation}
yields a 
description for the limiting pressure. Before spelling this in the subsequent proposition, we record the following properties of  $ F(\lambda;\upsilon) $ that are immediate from its definition:
\begin{enumerate}
    \item Applications of Jensen's inequality yield upper and lower bounds
\begin{equation}\label{eq:oneforall}
    F(0;\upsilon) \leq F(\lambda;\upsilon) \leq F(0;\upsilon) + \frac{\lambda^2\beta^2}{4} \left\| \upsilon\right\|_2^2 \leq \frac{\lambda^2\beta^2}{4} \left\| \upsilon\right\|_2^2 ,
\end{equation}
in terms of the rate function $  F(0;\upsilon) = \inf_{x \in \mathcal{S}_2 } \mathcal{P}_{0}(0 , x ) - \langle x , \upsilon \rangle $ of observing the self-overlap $ \upsilon $ under the a priori measure $ \nu_N $ (as described by the G\"artner-Ellis theorem). The rate function $ F(0;\upsilon) $ takes values in $ [-\infty, \infty) $, and~\eqref{eq:oneforall} shows that $  F(0;\upsilon) > -\infty $ if and only if $  F(\lambda;\upsilon) > - \infty $.  
\item Since $\mathcal{S}_2 \ni x \mapsto \mathcal{F}(x) = \inf_{\pi \in \Pi(\varrho) } \mathcal{P}_1(\pi, x)   $ is convex and Gateaux-differentiable by Proposition~\ref{prop:conc}, properties of the Legendre transform imply that at $ \varrho = \nabla \mathcal{F}(0)$: 
\begin{equation}\label{eq:relation0}
    F(1; \varrho) = \mathcal{F}(0) + \frac{\beta^2}{4} \left\| \varrho\right\|_2^2 .
\end{equation}
This provides the link to the variational description~\eqref{QPsoc} of the self-overlap-corrected model.
\end{enumerate}
The variational description of the pressure of the self-overlap constraint model, which generalizes results in \cite{Pan18}, is addressed in
\begin{prop}\label{prop:constrained}
    For any fixed self-overlap $\upsilon \in \mathcal{S}_2^{+}$ at which $ F(0;\upsilon) > - \infty $, and any $ \lambda \in \mathbb{R} $,  the constrained pressure converges in the following sense
    \begin{equation}\label{eq:equalmodels}
         F(\lambda; \upsilon) = \liminf_{\varepsilon \downarrow 0}\liminf_{N \to \infty} \frac1N \mathbb{E} \ln Z_{N,\varepsilon}(\lambda;\upsilon) =  \limsup_{\varepsilon \downarrow  0}\limsup_{N \to \infty} \frac1N \mathbb{E} \ln Z_{N,\varepsilon}(\lambda;\upsilon) .
    \end{equation}
    Furthermore:
       \begin{enumerate}
     \item 
         the function $ \lambda \mapsto F(\lambda;\upsilon) $ is convex and differentiable. 
    \item 
        there exists a null sequence $\varepsilon_N^{0} \to 0$, which is independent of $\lambda$, such that for each further null sequence $\varepsilon_N \geq \varepsilon_N^{0} $ in an open neighborhood of $\lambda$:
        \begin{equation}\label{eq:Wcoupled}
             \lim_{N \to \infty} \frac1N \mathbb{E}\ln  Z_{N,\varepsilon_N}(\lambda;\upsilon) = F(\lambda;\upsilon) .
        \end{equation}
    \end{enumerate}
 \end{prop}
 The proof of this proposition, on which the characterization of the order parameter in the spin-glass phase is based, is found in Subsection~\ref{P:constrained}. \\

Under the constraint on the self-overlap, one may include in the exponential in~\eqref{def:soconstrp} the self-overlap correcting term, which then changes the limits~\eqref{eq:equalmodels} predictably. In particular, in case $ F(0;\upsilon) > - \infty $, this implies that the two-parameter family
\begin{equation}\label{def:soconstrp2}
    W_{N,\varepsilon}(\lambda_1,\lambda_2;\upsilon) \coloneqq \int_{\mathcal{Q}_{N,\varepsilon}(\upsilon) } \mkern-10mu \exp\left( -  \int_0^1 \lambda_1 \beta U\left(\bxi(t)\right) dt - N \frac{\lambda_2^2\beta^2}{4}  \left\| Q_{N}[\bxi] \right\|_2^2 \right) \nu_N(d\mathbf{\bxi}) ,
\end{equation}
has a pressure $$ \widehat F(\lambda_1,\lambda_2;\upsilon ) = \lim_{\varepsilon\downarrow 0} \lim_{N\to \infty} \frac{1}{N} \mathbb{E} \ln W_{N,\varepsilon}(\lambda_1,\lambda_2;\upsilon) ,
$$
with limits understood as in Proposition~\ref{prop:constrained}, which agrees 
with
$$
\widehat F(\lambda_1,\lambda_2;\upsilon ) \coloneqq F(\lambda_1;\upsilon) - \frac{\lambda_2^2\beta^2}{4} \| \upsilon \|_2^2 . 
$$
One of the benefits of considering the augmented quantity, $ \ln  W_{N,\varepsilon}(\lambda_1,\lambda_2;\upsilon) $, is that this quantity is evidently convex in $ \lambda_1 \in \mathbb{R} $. This will allow us to pass identities for derivatives with respect to $ \lambda_1 $ to their limits. 
 The derivative of $ \ln  W_{N,\varepsilon}(\lambda_1,\lambda_2;\upsilon) $ involves the Gibbs measure associated with~\eqref{def:soconstrp2}, which we will abbreviate by
\begin{equation}\label{def:Gibbs}
 \langle (\cdot) \rangle_{N,\varepsilon;\upsilon}^{(\lambda_1,\lambda_2)} \coloneqq  W_{N,\varepsilon}(\lambda_1,\lambda_2;\upsilon)^{-1}   \int_{\mathcal{Q}_{N,\varepsilon}(\upsilon) } (\cdot) \  \exp\left( -  \int_0^1 \lambda_1 \beta U\left(\bxi(t)\right) dt  
 - \frac{N \lambda_2^2\beta^2}{4}  \left\| Q_N[\bxi] \right\|_2^2 \right) \nu_N(d\mathbf{\bxi}) ,
\end{equation}
 as well as its duplicated version, which we denote by $ \langle\langle (\cdot) \rangle\rangle_{N,\varepsilon;\upsilon}^{(\lambda_1,\lambda_2)} $. 
 A fairly standard argument using Gaussian integration by parts (cf.~\cite[Ch. 1.2]{Pan13}) yields the following identity
 \begin{equation}\label{eq:GIP}
  \frac{1}{N} \ \frac{\partial}{\partial \lambda_1} \mathbb{E} \left[ \ln  W_{N,\varepsilon}(\lambda_1,\lambda_2;\upsilon) \right]  = \frac{\lambda_1 \beta^2 }{2} \left( \mathbb{E}\left[  \langle \left\| Q_N \right\|_2^2  \rangle_{N,\varepsilon;\upsilon}^{(\lambda_1,\lambda_2)}  \right] - \mathbb{E}\left[ \langle \langle \left\| R_N \right\|_2^2  \rangle \rangle_{N,\varepsilon;\upsilon}^{(\lambda_1,\lambda_2)}  \right] \right) . 
 \end{equation}
Up to $ \varepsilon $, the first term in the bracket on the right is  $  \| \upsilon \|_2^2 $. We will show that such relations carry over to the limit. This will be crucial for the proof of~\eqref{eq:notselfav}, i.e.,  that the replica overlap is non-trivial in the spin glass regime. 
\begin{theorem}\label{prop:diffident}
    There is a minimizer $ \pi^* \in \Pi $ of $ \inf_{\pi\in \Pi(\varrho)} \mathcal{P}_1(\pi,0) $ such that at $ \varrho = \nabla \mathcal{F}(0)$: 
    \begin{equation}\label{eq:derivCon}
         \frac{\partial \widehat F}{\partial \lambda_1}(1, 1;\varrho) = \frac{\beta^2}{2} \left( \|\varrho \|_2^2 - \int_{0}^{1} \|\pi^{*}(s) \|^2_{2} \, ds  \right) .
     \end{equation}
\end{theorem}
The proof is found in Subsection~\ref{S:diffident}.\\

The proof of Theorem~\ref{thm:main}, 
 in the replica-symmetric case, is based on a conditional second-moment analysis. The main idea is that, in the replica symmetric regime, the self-overlap-corrected QSK is equivalent to a constrained QSK in which one restricts the self-overlap to  
 the typical self-overlap $ \mu $ under the a priori path measure $ \nu_N $.  
The following theorem will then imply~\eqref{eq:part1} and~\eqref{eq:noreplica}.%
\begin{theorem}\label{thm:CRS}
	For any $0 <  \lambda_1 \leq \frac{b}{\tanh(\beta b)} $ and all $ \lambda_2 \in \mathbb{R} $
	\begin{align}\label{eq:pressureSOC}
	& \widehat F(\lambda_1,\lambda_2;\mu ) =  \frac{\lambda_1^2-\lambda_2^2}{4} \ \beta^2 \left\| \mu \right\|_2^2 , \\
    \label{eq:replica0}
	& \limsup_{\varepsilon \downarrow 0} \limsup_{N\to \infty}  \mathbb{E}\left[ \langle \langle \left\| R_N \right\|_2^2  \rangle \rangle_{N,\varepsilon;\mu}^{(\lambda_1,\lambda_2)}  \right]  = 0 . 
\end{align}
\end{theorem}
\begin{proof}
The proof of~\eqref{eq:pressureSOC} is the topic of Section~\ref{sec:ProofPart1}, where it is found in Subsection~\ref{sec:ProofCRS}. We now show how this and the identity~\eqref{eq:GIP} imply~\eqref {eq:replica0}.

Consider first the case $ 0 < \lambda_1 < \frac{b}{\tanh(\beta b)}.$ 
The first term on the right side in~\eqref{eq:GIP} with $ \upsilon = \mu $ is estimated from above by $ \frac{\lambda_1 \beta^2 }{2} \left(  \| \mu \|_2^2 + \varepsilon \right) $.  Due to convexity, one may interchange the limits and the differentiation to conclude from~\eqref{eq:pressureSOC}:
$$
 \lim_{\varepsilon \downarrow 0} \lim_{N\to \infty}  \frac{1}{N} \ \frac{\partial}{\partial \lambda_1} \mathbb{E} \left[ \ln  W_{N,\varepsilon}(\lambda_1,\lambda_2;\mu) \right] = \frac{\partial}{\partial \lambda_1}   \frac{\lambda_1^2-\lambda_2^2}{4} \ \beta^2 \left\| \mu \right\|_2^2 =  \frac{\lambda_1 \beta^2 }{2} \| \mu \|^2_2 . 
$$
This completes the proof for $0 < \lambda_1 < \frac{b}{\tanh(\beta b)}.$ 

On the critical line $ \lambda_1 = \frac{b}{\tanh(\beta b)}$, we need a few modifications. We fix $\lambda_2^0 \in \mathbb{R} $ and employ Proposition~\ref{prop:constrained}, which guarantees the existence of a null sequence $\varepsilon_N^0$ such that for every null sequence $\varepsilon_N \geq \varepsilon_N^{0}$ the coupled limit $  \lim_{N\to \infty} \frac{1}{N} \mathbb{E} \left[ \ln  W_{N,\varepsilon_N}(\lambda_1,\lambda_2;\mu) \right] $ exists in an open neighborhood of $(\frac{b}{\tanh(\beta b)},\lambda_2^0)$. By Proposition~\ref{prop:constrained}, its limit is convex and differentiable in $\lambda_1$. Moreover, for $\lambda_1 \leq \frac{b}{\tanh(\beta b)} =: \lambda_1^0 $ the limit agrees with the right side in \eqref{eq:pressureSOC}.  Hence, we still have 
$$  \lim_{N\to \infty}  \frac{1}{N} \ \frac{\partial}{\partial \lambda_1} \mathbb{E} \left[ \ln  W_{N,\varepsilon_N}\big(\lambda_1^0,\lambda_2^0;\mu\big) \right] = \frac{b}{\tanh(\beta b)} \frac{ \beta^2 }{2} \| \mu \|^2_2, $$
and applying again the identity~\eqref{eq:GIP} with $ \upsilon= \mu $, we deduce 
$  \lim_{N\to \infty}  \mathbb{E}\left[ \langle \langle \left\| R_N \right\|_2^2  \rangle \rangle_{N,\varepsilon_N;\mu}^{(\lambda_1^0,\lambda_2^0)}  \right]  = 0 $ 
for the coupled limit of the replica overlap. To remove the coupling, we proceed via contradiction and assume 
$$ \limsup_{\varepsilon \downarrow 0} \limsup_{N\to \infty}  \mathbb{E}\left[ \langle \langle \left\| R_N \right\|_2^2  \rangle \rangle_{N,\varepsilon;\mu}^{(\lambda_1^0,\lambda_2^0)}  \right] = a > 0. $$
Then one can extract a coupled null sequence $\varepsilon_N \geq \varepsilon_N^0 $ for which $\limsup_{N \to \infty} \mathbb{E}\left[ \langle \langle \left\| R_N \right\|_2^2  \rangle \rangle_{N,\varepsilon_N;\mu}^{(\lambda_1^0,\lambda_2^0)}  \right] = a $, which contradicts our prior consideration.
\end{proof}

The rest of the paper is dedicated to the proofs of the results discussed in this section.
\section{Approximations and proof of Parisi formulae}
We start by gathering basic results on the approximation of the QSK by a finite-dimensional vector-spin glass. These allow for the transfer of known results on vector-spin glasses~\cite{Pan18,Chen23} to the case of quantum glasses as was done in~\cite{MW25}. We also spell out the proofs of the simplified Parisi formula, 
Theorem~\ref{thm:QPsoc}, as well as the one for the constrained model, Proposition~\ref{prop:constrained}.

\subsection{Vector-spin glass approximation of the self-overlap-corrected QSK}\label{sec:approx}
 Proceeding as in~\cite{MW25}, we introduce the square-wave pulses 
\begin{equation}\label{eq:unitv}
	e_k\coloneqq 2^{D/2} \indfct_{(t_{k-1}, t_k]}
\end{equation}
which correspond to a dyadic decomposition $ t_k \coloneqq k 2^{-D} $, $ k \in \{1, \dots , 2^D \} $, of the unit time interval. They are orthonormal in the real Hilbert space $ L^2(0,1] $ and we abbreviate the orthogonal projection on their span by $ P_D \coloneqq \sum_{k=1}^{2^D} | e_k \rangle \langle e_k | $. The projection of any path to this subspace is
$$
\xi_j^D(s) \coloneqq \left( P_D \xi_j\right)(s) = \sum_{k=1}^{2^D} \xi_j(k)  \indfct_{(t_{k-1}, t_k]}(s) , \qquad \xi_j(k) \coloneqq 2^D  \int_{t_{k-1}}^{t_k} \xi_j(s) ds = 2^{D/2} \langle e_k , \xi_j \rangle .  
$$
Abbreviating $ \bxi^D  \coloneqq (\xi_1^D , \dots , \xi_N^D ) $, the partition function
\begin{equation}\label{eq:ZNSCD}
    \widehat Z_N^D \coloneqq
\int \exp\left( -\int_0^1 \beta U\left(\bxi^D(t)\right) dt  - N \frac{\beta^2}{4} \big\| Q_N[\bxi^D] \big\|_2^2  \right) \nu_N(d\mathbf{\bxi}) ,
\end{equation}
is that of a self-overlap-corrected vector-spin glass of SK-type. 
Its corresponding random probability  measure is given by
\begin{equation}\label{eq:probmeasD}
 \langle (\cdot) \rangle_{N}^D \coloneqq \frac{1}{\widehat Z_N^D}   \int (\cdot) \ \exp\left( -\int_0^1 \beta U\left(\bxi^D(t)\right) dt  - N \frac{\beta^2}{4} \big\| Q_N[\bxi^D] \big\|_2^2\right) \nu_N(d\mathbf{\bxi})  . 
\end{equation}

The following summarizes the results on the Parisi variational formula for the pressure of such vector-spin glasses:
\begin{enumerate}
    \item The Parisi functional corresponding to~\eqref{eq:ZNSCD} is given by the restriction of $ \mathcal{P}_1 $ from~\eqref{eq:Parisifctl} to matrix-valued paths $ \Pi^D $ and matrices $ x \in \mathcal{S}_2^D \coloneqq \{ x \in \mathcal{S}_2 \ | \ x = P_D x P_D \} $. Since $ \left(\mathcal{S}_2^+\right)^D = \mathcal{S}_2^+ \cap \mathcal{S}_2^D $,  one may take $ \Pi^D  = \{ P_D \pi P_D  \ | \ \pi \in \Pi \} $ and identify
    $
    \mathcal{P}^D_1(\pi, x) = \mathcal{P}_1(P_D\pi P_D, P_D x P_D ) $. 
    The functional 
    $$
    \mathcal{F}^D(x) \coloneqq \inf_{\pi \in \Pi^D} \mathcal{P}^D(\pi, x) $$
    which is again first defined on matrices $ x $ and then trivially extended to $ \mathcal{S}_2 $, is convex and Gateaux-differentiable. The derivative $ \varrho^D \coloneqq \nabla \mathcal{F}^D(0)$ marks the asymptotic value, on which the self-overlap concentrates under the averaged measure~\eqref{eq:probmeasD}, i.e.\
    \begin{equation}
        \lim_{N\to \infty} \left\| \mathbb{E}\left[\langle Q_N[\bxi^D]\rangle_{N}^D  \right] -   \varrho^D \right\|_2 = 0 , 
    \end{equation}
    see~\cite[Thm.~1.1]{Chen24} and note that weak and strong convergence are equivalent as long as $ D $ is finite. By the same proof technique as in Proposition~\ref{prop:conc}, one may boost this result to exponential concentration. 
    \item The pressure is expressed as the following two equivalent Parisi variational principles \cite{Chen23,Chen24}:
    \begin{equation}\label{eq:Chen23}
\lim_{N\to \infty} \frac{1}{N} \mathbb{E}\left[ \ln   \widehat Z_N^D \right] =\mathcal{F}^D(0) = \inf_{\pi \in \Pi^D(\varrho^D)} \mathcal{P}^D(\pi,0) ,
\end{equation}
where the last infimum restricts the paths $ \Pi^D $ to the endpoint $ \varrho^D $.
\end{enumerate}
\noindent
The following two approximation results have been established in \cite{MW25}:
\begin{enumerate}
    \item The pressure converges~\cite[Prop.~3.2]{MW25}:
    \begin{equation}\label{eq:ZDtoZ}
        \lim_{D\to \infty} \limsup_{N\to \infty} \frac{1}{N} \mathbb{E}\left[\left|\ln  \widehat Z_N^D  - \ln \widehat Z_N \right| \right] = 0 .
    \end{equation}
    \item 
    The Lipschitz continuity~\eqref{eq:Lipschitz} implies $ \lim_{D\to \infty} \mathcal{P}_\lambda(P_D\pi P_D,P_D xP_D)  = \mathcal{P}_\lambda(\pi, x) $ for all $ \pi \in \Pi $ and $ x \in \mathcal{S}_2$.
    One even has \cite[Lemma ~3.7]{MW25}
\begin{equation}\label{eq:functionalconf}
\lim_{D\to \infty} \mathcal{F}^D(P_D xP_D ) = \mathcal{F}(x) .
\end{equation}

\end{enumerate}

For a proof of Theorems~\ref{thm:QPsoc} and~\ref{prop:diffident}, we need the following extension of these approximation results.
\begin{prop}\label{prop:Parisiend}
    Let $ \lambda \in \mathbb{R} $. Then:
\begin{equation}\label{eq:Parisiend}
    \lim_{D\to \infty} \inf_{\pi \in \Pi^D(\varrho^D)} \mathcal{P}_\lambda^D(\pi,0) = \inf_{\pi \in \Pi(\varrho)} \mathcal{P}_\lambda(\pi,0) .
    \end{equation}
    \begin{enumerate}
    \item The limiting self-overlaps converge strongly as $ D \to \infty$:
    \begin{equation}
      \label{eq:strondD}
    \lim_{D\to \infty} \left\| \varrho^D  - \varrho \right\|_2 = 0 .
    \end{equation}
    \item Any sequence of minimizers $ \pi^D_* $ of $ \inf_{\pi \in \Pi^D(\varrho^D)} \mathcal{P}_\lambda^D(\pi,0) $ is relatively compact in $ L^2([0,1] , \mathcal{S}_2^+) $  and any limiting point $ \pi_* $ satisfies
    $$
    \inf_{\pi \in \Pi(\varrho)} \mathcal{P}_\lambda(\pi,0) = \mathcal{P}_\lambda(\pi_*,0)
    $$
    \end{enumerate}
\end{prop}
The proof is postponed to Appendix~\ref{app:ProofParisi}. 

\subsection{Proof of the simplified Parisi formula}\label{sec:ProofParisi}

\begin{proof}[Proof of Theorem~\ref{thm:QPsoc}]
The first equality in~\eqref{QPsoc} has been established in \cite[Thm. 2.3]{MW25}. Since $ \mathcal{F}(0) = \lim_{D\to\infty} \mathcal{F}^D(0) $ as a special case of~\eqref{eq:functionalconf}, the second equality in~\eqref{QPsoc} follows by combining the last identity in~\eqref{eq:Chen23} and~\eqref{eq:Parisiend}. The last equality in~\eqref{QPsoc} will follow from
$$
\inf_{\pi \in \Pi_c(\varrho) } \mathcal{P}_1(\pi,0 ) \geq \inf_{\pi \in \Pi_c} \mathcal{P}_1(\pi,0 ) \geq \inf_{\pi \in \Pi} \mathcal{P}_1(\pi,0 ) = \mathcal{F}(0) , 
$$
the previously proven identities and the third identity~\eqref{QPsoc}, i.e. 
$$
\inf_{\pi \in \Pi(\varrho) } \mathcal{P}_1(\pi,0 ) = \inf_{\pi \in \Pi_c(\varrho) } \mathcal{P}_1(\pi,0 ) .
$$
It hence remains to establish this equality.

The empirical self-overlap averaged with respect to the random Gibbs measure~\eqref{eq:probmeas} is invariant under time translations, 
$$
\langle Q_N(t,s) \rangle_N = \langle Q_N(t+\tau ,s+\tau ) \rangle_N 
$$
for any $ \tau \in [0,1) $, where additions are understood modulo one. This follows from the time-translation invariance of the a priori measure $ \nu_N$. Since $
\lim_{N\to \infty} \left\| \mathbb{E} \langle Q_N\rangle_N  - \varrho \right\|_2 = 0 $, as a consequence of Proposition~\ref{prop:conc},  the limit $ \varrho $ inherits the symmetry and hence equals its time average $ \overline{\varrho}$ given by
$$
\overline{\varrho}(t,s) = \int_0^1 \varrho(t+\tau,s+\tau) \ d\tau .
$$
We may thus apply~\cite[Thm.~2.6]{MW25} to conclude that a  minimizer in the variational problem $ \inf_{\pi \in \Pi(\varrho)} \mathcal{P}(\pi,0 )$ is found within the set of paths
$$
\Pi_2 :=\left\{ \pi \in \Pi(\varrho) \ \Big| \ \begin{array}{rl} \mbox{There are $ \kappa,\lambda :[0,1] \to [0,1] $ non-decreasing, right-} & \quad \pi(m) =   \kappa(m) \ \varrho^\perp  \\ \mbox{continuous with $ \kappa(0) = \lambda(0) = 0 $, $ \kappa(1) =\lambda(1) = 1 $:} 
& \qquad + \lambda(m) \ \langle 1 , \varrho \ 1 \rangle \ | 1 \rangle\langle 1 | 
\end{array} 
\right\} . 
$$
Here $ \varrho^\perp $ is defined through the spectral decomposition using the time-translation invariance, which yields $  \varrho =  \varrho^\perp + \langle 1 , \varrho \ 1 \rangle \ | 1 \rangle\langle 1 |  $, where the symbol $ | 1 \rangle\langle 1 | $ stands for the rank-one projection onto the constant function $ 1 \in L^2(0,1) $. 

Within the above class, we may now use standard Gaussian interpolation to show that the infimum corresponds to $ \kappa(m) = 0 $ for all $ m \in [0,1) $. 
For a proof, consider two monotone step functions 
$ \lambda, \kappa :[0,1] \to [0,1] $, and let 
$  (\mu^{\lambda,\kappa}_\alpha)_{\alpha \in \mathbb{N}^{r} }$ 
stand for the weights of the Ruelle probability cascade \cite{Pan13} corresponding to a common (refined) partition $ 0 =: m_0 < m_1 < \dots < m_{r-1} < m_{r}:= 1$ of the jump points of $ \lambda, \kappa$. 
Let $ w_{\alpha}^\lambda $ stand for a family of centered Gaussian random variables with covariance
$$
\mathbb{E}\left[w_{\alpha}^\lambda w_{\alpha'}^\lambda \right] = \lambda(\alpha \wedge \alpha' ) \ \langle 1 , \varrho 1 \rangle , \quad \mbox{with} \quad \alpha \wedge \alpha' \coloneqq \min\left\{  0\leq j \leq r  \ \big| \ \alpha_{j} = \alpha'_{j} \; \mbox{and} \;  \alpha_{j+1} \neq \alpha'_{j+1} \right\} .
$$
Moreover, let $ w_\alpha^\kappa $ stand for an independent family of 
$ L^2 $-valued, centered Gaussian processes with covariance given by the kernel
$$
\mathbb{E}\left[w_{\alpha}^\kappa(t) w_{\alpha'}^\kappa(s) \right] = \kappa(\alpha \wedge \alpha' )  \ \varrho^\perp(t,s) .
$$
By construction of the Ruelle probability cascade, the function
	\begin{align*}
		 f(p) :=
		 \mathbb{E}\bigg[ \ln \sum_{\alpha \in \mathbb{N}^{r} } \mu^{\lambda,\kappa}_\alpha\! \int \exp\Big( \beta w_\alpha^\lambda \ \langle \xi ,  1 \rangle   - \frac{\beta^2 }{2} \ \langle 1 , \xi\rangle^2 \ \langle 1 , \varrho 1 \rangle   +
         \beta \sqrt{p}   \langle \xi ,  w_\alpha^\kappa \rangle
         - \frac{\beta^2 p}{2} \langle \xi , \varrho^\perp  \xi  \rangle\Big) \nu_1(d\xi) \bigg] ,
	\end{align*}
    interpolates for $ p \in [0,1] $ between $ X_0(\lambda \langle 1 , \varrho 1 \rangle | 1 \rangle\langle 1 | , 0 ) = f(0) $ and $ X_0(\lambda \langle 1 , \varrho 1 \rangle | 1 \rangle\langle 1 | + \kappa \varrho^\perp , 0 ) = f(1) $. 
	Let $ \langle (\cdot) \rangle_p $ denote the (random) Gibbs expectation value associated with the partition function defined by the logarithm, and $ \langle\langle (\cdot) \rangle\rangle_r $ be the Gibbs expectation value of the corresponding duplicated system. 
	Then straightforward differentiation and Gaussian integration by parts yield
	\begin{equation}\label{eq:derPf}
		f'(p) = -  \frac{\beta^2}{2}\  \mathbb{E}\left[ \left\langle  \left\langle  \kappa(m_{\alpha\wedge\alpha'}) \ \langle \xi , \varrho^\perp  \xi' \rangle \right\rangle\right\rangle_r\right] ,
	\end{equation}
	where the variables with and without the prime refer to the first and second system in the duplication. 
	Note that the term $  \frac{\beta^2}{4}\  \mathbb{E}\left[ \left\langle  \langle \xi ,  \varrho^\perp \xi \rangle \right\rangle_r \right] $, which appears in the Gaussian integration by parts, cancels due to the presence of this term in the exponent. 
    Since the measure $ \nu_1 $ is invariant under time shifts $ \xi \to \xi( \cdot + \tau) $ with additions modulo one and $ \tau \in [0,1] $, and so are the terms 
    $ \langle \xi , 1 \rangle = \langle \xi( \cdot + \tau)  , 1 \rangle $ as well as $ \langle \xi , \varrho^\perp  \xi  \rangle = \langle \xi( \cdot + \tau) , \varrho^\perp  \xi( \cdot + \tau)  \rangle $, we hence conclude that for any pair $ \tau, \tau' \in [0,1)$:
    $$
    f'(p) = - \frac{\beta^2}{2}\  \mathbb{E}\left[ \left\langle  \left\langle  \kappa(m_{\alpha\wedge\alpha'}) \ \langle \xi( \cdot + \tau), \varrho^\perp  \xi'(\cdot + \tau')\rangle \right\rangle\right\rangle_r\right] .
    $$
    Averaging over $ \tau, \tau' $ with respect to the Lebesgue measure and using the fact that 
    $$
    \int_0^1  \xi(t + \tau) d\tau = \langle \xi , 1 \rangle  ,
    $$
    which is constant in $ t $ and hence orthogonal to the range of $ \varrho^\perp $, we hence conclude that $ f'(p) = 0 $ for all $ p \in [0,1]$. This implies that
    \begin{align}
    & \mathcal{P}_1\left(\lambda \langle 1 , \varrho 1 \rangle | 1 \rangle\langle 1 | + \kappa \varrho^\perp , 0 \right) \notag \\
    & \quad = X_0\left(\lambda \langle 1 , \varrho 1 \rangle | 1 \rangle\langle 1 | + \kappa \varrho^\perp , 0 \right) + \frac{\beta^2}{4} \| \varrho^\perp \|_2^2 \int_0^1 \kappa(m)^2 dm + \frac{\beta^2}{4} \langle 1 , \varrho 1 \rangle^2 \int_0^1 \lambda(m)^2 dm \notag  \\
    & \quad \geq X_0\left(\lambda \langle 1 , \varrho 1 \rangle | 1 \rangle\langle 1 | , 0 \right) + \frac{\beta^2}{4} \langle 1 , \varrho 1 \rangle^2 \int_0^1 \lambda(m)^2 dm  = \mathcal{P}_1\left(\lambda \langle 1 , \varrho 1 \rangle | 1 \rangle\langle 1 | , 0 \right)
    \end{align}
    and completes the proof of~\eqref{QPsoc}.

    Suppose $\pi_{*} \in \Pi(\varrho)$ is a minimizer. Consulting the proof of \cite[Thm. 2.6]{MW25}, we note that $\pi_{*}$ needs to agree almost everywhere with a path from $\Pi_2$ as otherwise Jensen's inequality underlying this argument becomes strict. Due to right-continuity and the fixed endpoint, we deduce that $\pi_{*} \in \Pi_2$. Similarly, the above proof also yields a strict inequality unless $\kappa(m) = 0$ almost everywhere. Hence, $\pi_{*} \in \Pi_c$, and it only remains to show uniqueness within the set of classical paths. The claimed uniqueness of the minimizer in $ \Pi_c $ follows from a straightforward application of the technique in~\cite{AC15a}, which establishes strict convexity of the Parisi functional when restricted to classical paths. 
\end{proof}

\subsection{Approximations for self-overlap-constrained models}

For $ \upsilon \in \left(\mathcal{S}_2^{+}\right)^D $,  we define 
$$
\mathcal{Q}^D_{N,\varepsilon}(\upsilon) \coloneqq \left\{ \bxi \ | \ \left\| Q_N[\bxi^D] - \upsilon  \right\|_2 \leq \varepsilon \right\} 
$$
and the approximand of the self-overlap constraint partition function
\begin{equation}\label{eq:defWapp0}
	Z_{N,\varepsilon}^D(\lambda;\upsilon)  \coloneqq  \int_{\mathcal{Q}^D_{N,\varepsilon}(\upsilon) } \mkern-10mu  \exp\left( -\int_0^1 \lambda\beta U\left(\bxi^D(t)\right) dt  \right) \nu_N(d\mathbf{\bxi})  .
\end{equation}

The Parisi functional corresponding to the constraint partition function builds on the restriction of $ \mathcal{P}_{\lambda}(\pi,x) $ to matrix-valued paths $ \pi \in \Pi^D(\upsilon) $, which end in $ \upsilon \in (\mathcal{S}_2^+)^D$, and $ x \in \mathcal{S}_2^D$. The pressure is expressed in terms of a Legendre transformation
\begin{equation}\label{def:FLeg}
    F^D(\lambda;\upsilon) \coloneqq \inf_{x \in \mathcal{S}_2^D }\left[ \inf_{\pi \in \Pi^D(\upsilon) }  \mathcal{P}^D_{\lambda}(\pi , x ) - \langle x , \upsilon \rangle \right] + \frac{\lambda^2\beta^2 }{4 }  \left\| \upsilon \right\|_2^2 .
\end{equation}
At $ \lambda = 0 $, this functional reduces to the rate function $  F^D(0;\upsilon) = \inf_{x \in \mathcal{S}_2^D } \mathcal{P}^D_{0}(0 , x ) - \langle x , \upsilon \rangle $ of observing the self-overlap $ \upsilon $ under the apriori measure on the square-wave pulses. Applications of Jensen's inequality show that~\eqref{eq:oneforall} extends verbatim to $F^D $. In the limit $ D \to \infty $, the functionals $ F^D $ approximate~\eqref{def:Parisicon} in the following sense.
\begin{lemma}\label{lem:FDtoF}
    For any $ \upsilon \in \mathcal{S}_2^+ $ at which $ F(0;\upsilon) > - \infty $, and for any $ \lambda \in \mathbb{R} $, one has
    \begin{equation}
        \lim_{D\to \infty} F^D(\lambda;P_D \upsilon P_D) = F(\lambda;\upsilon) . 
    \end{equation}
\end{lemma}
The proof is spelled out in Appendix~\ref{S:FDtoF}. 
\begin{rem} We caution the reader that 
intuitive claims concerning the smoothness of $ F $ are not all correct in the infinite-dimensional setting. For instance, $F$ is not continuous in $\upsilon $ with respect to the Hilbert-Schmidt norm. If $\varrho \neq P_D \varrho P_D$ with $ P_D $ a finite-dimensional projection such as those considered in Subsection~\ref{sec:approx}, one  has $F(\lambda;P_D \varrho P_D) = - \infty$. The latter is a consequence of the path-integral being only supported on self-overlaps with $\tr \varrho = 1$. The continuity even fails if one restricts the functional to the interior of the set of density matrices.
\end{rem}

The following summarizes Panchenko's results \cite{Pan18} on the variational description of the pressure for the constrained vector-spin glass approximants:
\begin{enumerate}
    \item 
For any finite $ D \in \mathbb{N} $, any $ \upsilon \in \left(\mathcal{S}_2^{+}\right)^D  $ at which  $ F^D(0;\upsilon) > - \infty 
$, 
and any $\lambda \in \mathbb{R}$, one has \cite[Thm. 2]{Pan18}
\begin{equation}\label{eq:Panchenko00}
   \liminf_{\varepsilon\downarrow 0 } \liminf_{N\to \infty} \frac{1}{N} \mathbb{E}\left[\ln Z_{N,\varepsilon}^D(\lambda;\upsilon)  \right]  \geq F^D(\lambda;\upsilon) .
\end{equation}
Under the same assumptions, for any $ \varepsilon> 0$:
\begin{align}\label{eq:Panchenko01}
  \limsup_{N\to \infty} \frac{1}{N} \mathbb{E}\left[\ln Z_{N,\varepsilon}^D(\lambda;\upsilon)  \right] \leq \sup_{\upsilon' \in B^D_\varepsilon(\upsilon) } F^D(\lambda;\upsilon')  
\end{align}
with $B^D_\varepsilon(\upsilon)  \coloneqq \{ \upsilon' \in \mathcal{S}_2^D \ | \ \| \upsilon' - \upsilon \|_2 < \varepsilon \} $, see~\cite[Lemma 2, Proof of Lemma 3]{Pan18}. The proof of~\eqref{eq:Panchenko01} is based on Guerra interpolation and a covering argument, which works for finite $ D $. 
By the Lipschitz continuity~\eqref{eq:Lipschitz} of the Parisi functional, this implies 
\begin{align}\label{eq:Panchenko0}
 F^D(\lambda;\upsilon) = \liminf_{\varepsilon\downarrow 0 } \liminf_{N\to \infty} \frac{1}{N} \mathbb{E}\left[\ln Z_{N,\varepsilon}^D(\lambda;\upsilon)  \right]  = \limsup_{\varepsilon\downarrow 0 }  \limsup_{N\to \infty} \frac{1}{N} \mathbb{E}\left[\ln Z_{N,\varepsilon}^D(\lambda;\upsilon)  \right] . 
\end{align}
\item If $ F^D(0;\upsilon) > - \infty $, the function $ \lambda \mapsto F^D(\lambda;\upsilon) $ is convex and differentiable. Convexity follows directly from that of the prelimit $ \mathbb{E}\left[\ln Z_{N,\varepsilon}^D(\lambda;\upsilon)  \right] $. The differentiability follows from the argument in~\cite{Pan08} (see also~\cite[Proof of Thm. 1.1]{Chen23} or \cite[Proof of Cor. 2.4]{MW25}), which shows that the infimum of the differentiable functions in the right-hand side of~\eqref{eq:Panchenko0} remains differentiable due to the convexity of the limit. 
\item At $ \upsilon = \varrho^D = \nabla \mathcal{F}^D(0) $, we have
$F^D(1;\varrho^D) = \mathcal{F}^D(0) + \frac{\beta^2}{4}  \left\| \varrho^D\right\|_2^2$, similarly as in~\eqref{eq:relation0}.
\end{enumerate}

Panchenko's analysis \cite{Pan18} also provides information about the relation of the empirical replica overlap and minimizers of the Parisi functional. 
For a proof of Theorem~\ref{prop:diffident}, we require the following result, which refers to the Gibbs measure $ \langle (\cdot) \rangle_{N,\varepsilon, \upsilon}^{(D, \lambda_1,\lambda_2)}$,  defined as in~\eqref{def:Gibbs} with the paths restricted to square-wave pulses. Its duplicated version is denoted by $ \langle\langle  (\cdot) \rangle \rangle_{N,\varepsilon, \upsilon}^{(D, \lambda_1,\lambda_2)}$. 
 \begin{lemma}\label{lem:replica}
     For any fixed $D \in \nn$, there exists a Parisi measure $\pi_{*}^D\in \Pi^D({\varrho}^D)$, i.e., $\inf_{\pi \in \Pi^D} \mathcal{P}^D(\pi,0) = \mathcal{P}^D(\pi_{*}^D,0)$, such that 
\begin{equation}\label{eq:replicaD}
         \lim_{N\to \infty} \mathbb{E}\left[ \langle\langle \| R_N \|_2^2 \rangle\rangle_{N,\varepsilon_N, \varrho^D}^{(D, 1,1)} \right] = \int_{0}^{1} \|\pi_{*}^D(m) \|^2_{2} \, dm
     \end{equation}
     for any sequence $\varepsilon_N \downarrow 0$ converging slowly enough. Moreover,
\begin{equation}\label{eq:derivD}
         \frac{\partial F^D}{\partial \lambda}(1;\varrho^D) = \frac{\beta^2}{2} \left( \|\varrho^D \|_2^2 - \int_{0}^{1} \|\pi_{*}^D(m) \|^2_{2} \, dm  \right).
     \end{equation}
 \end{lemma}
As will be explained in Appendix~\ref{app:Panch}, these results can be easily derived from~\cite{Pan18} .

\subsection{Proof of Proposition~\ref{prop:constrained}}\label{P:constrained}

To spell out a proof of the Parisi formula for the self-overlap-constrained QSK,
we first recall from~\cite{MW25} that we may restrict the path measures $ \nu_N$ to paths whose self-overlap is restricted to 
\begin{equation}\label{eq:restSO}
 \mathcal{S}_{N,\varepsilon}^D \coloneqq \Big\{\bxi \ \big| \  \frac{1}{N} \sum_{j=1}^N \langle \xi_j , (1 -P_D) \xi_j \rangle < \varepsilon  \Big\}.
 \end{equation}
The following large-deviation result captures the super-exponential concentration under the measure $ \nu_N $.
\begin{lemma}[=Lemma 3.3 in \cite{MW25}]\label{lem:LD} For any $\varepsilon > 0, L \in \nn$, there exists a $D(\varepsilon, L) \in \nn$ such that for all $D \geq D(\varepsilon, L)$ and $ N \in \mathbb{N}$:
\begin{equation}\label{eq:LDfree}
    \nu_N((\mathcal{S}_{N,\varepsilon}^D)^c) \leq e^{-NL}.
\end{equation}
\end{lemma}
A second preparation is the  following straightforward extension of~\cite[Prop.~3.5]{MW25}, which addresses
$$ Z_N^D(A) \coloneqq  \int_{A } \exp\left( -  \int_0^1 \lambda \beta U\left(\bxi^D(t)\right) dt  
  \right) \nu_N(d\mathbf{\bxi}),
 $$
 and its limit $Z_N(A) $
 for a measurable set $ A $ in path space.
\begin{prop}\label{prop:35}
There is a constant $ C \in (0,\infty) $ such that for any measurable set  $ A$ and all $ \varepsilon > 0 $, $ D , N \in \mathbb{N} $:
$$
\left| \mathbb{E}\left[ \ln Z_N(A\cap \mathcal{S}_{N,\varepsilon}^D) - \ln Z^D_N(A\cap \mathcal{S}_{N,\varepsilon}^D) \right] \right| \leq C N \sqrt{\varepsilon} . 
$$
\end{prop}
The proof proceeds as in~\cite[Prop.~3.5]{MW25} by standard Gaussian interpolation. 

\begin{proof}[Proof of Proposition~\ref{prop:constrained}]
Note that $Z_N(B_\varepsilon(\upsilon)) =  Z_{N,\varepsilon}(\lambda;\upsilon) $ for balls $ B_\varepsilon(\upsilon) = \{ \upsilon' \ | \ \| \upsilon' - \upsilon \|_2 < \varepsilon \} $. 
 For the set $
 \mathcal{S}_{N,\varepsilon}^D  
$  from~\eqref{eq:restSO}, we proceed as in~\cite[Eq.~(3.9)]{MW25} to conclude that for any $ D, N $ and $ \varepsilon> 0$:
\begin{equation*}
     0 \leq  \mathbb{E} \ln \frac{Z_N(B_\varepsilon(\upsilon))}{Z_N(B_\varepsilon(\upsilon)\cap  \mathcal{S}_{N,\varepsilon}^D )} \leq e^{N\lambda^2\beta^2/2} \frac{\nu_N\left((\mathcal{S}_{N,\varepsilon}^D )^c\right)}{\nu_N(B_\varepsilon(\upsilon)\cap  \mathcal{S}_{N,\varepsilon}^D)}  .
\end{equation*}
The assumption $  F(0;\upsilon) > - \infty $ and  the concentration bound from Lemma~\ref{lem:LD} hence imply that 
 for all $ \varepsilon > 0 $ there is $ D(\varepsilon) > 0 $ such that for all $ \lambda \in \mathbb{R} $, $N \in \mathbb{N} $ and $ D \geq D(\varepsilon)$:
\begin{align}
\frac1N \mathbb{E} \ln Z_{N,\varepsilon}(\lambda;\upsilon) & \leq  \frac1N \mathbb{E} \ln Z_N(B_{3\varepsilon}(P_D\upsilon P_D) \cap \mathcal{S}_{N,\varepsilon}^D ) + o_N(1) \notag \\
& \leq \frac1N \mathbb{E} \ln  Z_{N,3\varepsilon}^D(\lambda; P_D \upsilon  P_D) + \mathcal{O}(\sqrt{\varepsilon}) + o_N(1) ,
\end{align}
where the last inequality is by Proposition~\ref{prop:35},  $ \mathcal{O}(\sqrt{\varepsilon}) $ is independent of $ N $, and $ o_N(1) $ is a null sequence in $ N $, which is independent of $ D$.
This implies
\begin{align*}
\limsup_{N\to \infty } \frac1N \mathbb{E} \ln Z_{N,\varepsilon}(\lambda;\upsilon) & \leq \limsup_{N\to \infty }  \frac1N \mathbb{E} \ln  Z_{N,3\varepsilon}^D(\lambda; P_D \upsilon  P_D) +  \mathcal{O}(\sqrt{\varepsilon}) \\
& \leq \sup_{\upsilon' \in B^D_{3\varepsilon}(P_D\upsilon P_D) } F^{D}(\lambda; \upsilon') +  \mathcal{O}(\sqrt{\varepsilon}) ,
\end{align*}
 where the last step is by~\eqref{eq:Panchenko01}. 

  The idea is to control $ F^D(\lambda;\upsilon') $ uniformly for large $ D $ by using near minimizers of the limiting functional as trial configurations for all approximants. To do so, let $ \delta > 0 $ be arbitrary. As in the proof of Lemma~\ref{lem:FDtoF}, we may find   $x_{\delta} = P_D x_{\delta}P_D \in \mathcal{S}_2^D$ and $\pi_{\delta} = P_D \pi_{\delta}P_D \in \Pi^D(P_D \upsilon P_D) $ for all $ D \geq D(\varepsilon) $ large enough at which 
    $$  \inf_{\pi \in \Pi(\upsilon)} \left[   \mathcal{P}_\lambda\left(\pi, x \right) - \langle x, \upsilon\rangle  \right] \geq  \mathcal{P}_\lambda\left(\pi_{\delta}, x_{\delta} \right) - \langle x_{\delta}, \upsilon \rangle  -\delta.$$
   We cannot insert  $\pi_{\delta} \in (\mathcal{S}_2^+)^D$ into the approximating Parisi functional due to the mismatch in the endpoint. Instead, we use Lemma~\ref{lem:switch} below to distort $\pi_{\delta}$ to end at $\upsilon' $. The Lipschitz-continuity~\eqref{eq:Lipschitz} of the Parisi functional then implies 
   $$ F^D(\lambda;\upsilon') \leq F(\lambda;\upsilon) + \delta + o_\varepsilon(1) +  \frac{\lambda^2\beta^2}{4}\left(  \|\upsilon \|_2^2 - \|\upsilon' \|_2^2 \right) + \langle x_{\delta}, P_D \upsilon P_D -\upsilon \rangle \leq F(\lambda;\upsilon) + \delta + o_\varepsilon(1) $$
   for all $\upsilon' \in B^D_{3\varepsilon}(P_D \upsilon P_D) $. Here $ o_\varepsilon(1)  $ denotes a sequence, which is independent of $ D $ and converges to zero as $ \varepsilon \downarrow 0 $. Even though this estimate is only true for $D \geq D(\varepsilon)$ large enough, the uniformity of $ o_\varepsilon(1) $ in $ D $ allows us to take a coupled limit $\varepsilon\downarrow 0, D \to \infty$,  to conclude
   $$ \limsup_{\varepsilon \to 0}\limsup_{N \to \infty} \frac1N \mathbb{E} \ln Z_{N,\varepsilon}(\lambda; \upsilon) \leq F(\lambda;\upsilon) + \delta . $$
   Since $\delta > 0$ was arbitrary, the proof of an upper bound for~\eqref{eq:equalmodels} is completed.

   For a complementing lower bound, we again restrict to the set $
 \mathcal{S}_{N,\varepsilon}^D  
$  from~\eqref{eq:restSO} to conclude that for all $ N , D $ and $ \lambda \in \mathbb{R}$:
 \begin{equation}
     \frac1N \mathbb{E} \ln 
     Z_{N,2\varepsilon}(\lambda;\upsilon) \geq \frac1N \mathbb{E} \ln Z_{N}(B_{\varepsilon}(P_D\upsilon P_D) \cap \mathcal{S}_{N,\varepsilon}^D ) \geq \frac1N \mathbb{E} \ln Z^D_{N,\varepsilon}(\lambda;P_D\upsilon P_D) + \mathcal{O}(\sqrt{\varepsilon}) , 
 \end{equation}
    where the last inequality is again by Proposition~\ref{prop:35}.  
  Taking first $ N \to \infty $ and then $ \varepsilon \downarrow 0 $, we hence conclude using~\eqref{eq:Panchenko00} for all $ D $:
  \begin{equation}
      \liminf_{\varepsilon\downarrow 0} \liminf_{N\to \infty}  \frac1N \mathbb{E} \ln 
     Z_{N,2\varepsilon}(\lambda;\upsilon) \geq F^D(\lambda;P_D\upsilon P_D) .
  \end{equation}
    Taking the subsequent limit $ D \to \infty $ and using Lemma~\ref{lem:FDtoF} concludes the proof of~\eqref{eq:equalmodels}.\\

    \noindent
    1.~Convexity is evident from the convexity of the prelimit. Since the constrained Parisi formula represents the quantity as an infimum, the differentiability in $\lambda$ follows from the classical Panchenko argument \cite{Pan08}. We refrain from repeating this well-known argument, see also the proof of Proposition~1.6 in \cite{MW25}.

\noindent 2.~ Let us first fix  $\lambda$. Since the quantities are finite, for every $k \in \nn$, we find an $\varepsilon_k < \frac1k $ such that
$$  |F(\lambda;\upsilon) -\liminf_{N \to \infty} \frac1N \mathbb{E} \ln Z_{N,\varepsilon_k}(\lambda;\upsilon)| < \frac1k, \quad |F(\lambda;\upsilon) -\limsup_{N \to \infty} \frac1N \mathbb{E} \ln Z_{N,\varepsilon_k}(\lambda; \upsilon)| < \frac1k. $$
Hence, there is an $N_k \in \nn$ such that for all $N \geq N_k$
$$    |F(\lambda;\upsilon) - \frac1N \mathbb{E} \ln Z_{N,\varepsilon_k}(\lambda;\upsilon)| < \frac2k.  $$
We may assume without loss of generality that $N_{k+1} > N_k$. The sequence
$$ \varepsilon_N^{0} := \begin{cases} 1 & \text{if }N < N_1, \\
\frac1k & \text{if } N_k \leq N < N_{k+1},
\end{cases}
$$
has the desired properties at $ \lambda$. This construction can easily be extended to cover a finite number of different points $\lambda$, and by a diagonal argument, one may even find a single sequence $\varepsilon_N^0$ for a countable set  such as all rational points $\lambda \in \mathbb{Q} $.  Since the pressure is an equicontinuous function on any compact subset, such a null sequence $\varepsilon_N^0$ has the desired property for all reals as well. Due to monotonicity in $\varepsilon$, the convergence property transfers to every further null sequence $\varepsilon_N \geq \varepsilon^0_N$. 
\end{proof}

The previous proof relied on the following observation, 
which allows us to map a path $\pi \in \Pi(\varrho)$ to another path $\pi' \in \Pi(\varrho')$ with a different endpoint in a continuous manner.

\begin{lemma}\label{lem:switch}
    Let $\pi \in \Pi(\varrho)$ and $\varrho,\varrho'\in \mathcal{S}_2^+$. Then, there exists a path $\pi' \in \Pi(\varrho')$ such that
    \begin{equation}\label{eq:pathswitch}
        \sup_{m\in[0,1]}\| \pi(m) - \pi'(m)\|_2 \leq \sqrt{\|\varrho - \varrho'\|_2} \big(\sqrt{\|\varrho \|_2} + \sqrt{\|\varrho' \|_2 } \big)
    \end{equation}
\end{lemma}

\begin{proof} 
 For $ \delta > 0 $, we define the family of operators 
    $$ C_\delta(m) := (\varrho+\delta)^{-\frac12} (\pi(m) +\delta) (\varrho+\delta) ^{-\frac12}, \quad 0 < m < 1, \quad C_\delta(0) := 0, \, C_\delta(1) := \mathbbm{1} , $$
    where addition by $ \delta $ is understood as $ \delta  \mathbbm{1} $ throughout.  
    Since $\varrho + \delta \geq \delta > 0$, the inverse square roots are well defined by spectral calculus. Consulting the quadratic form, one easily verifies that these operators are non-negative contractions, which are monotone by the monotonicity of the path $\pi$,
    \begin{equation}\label{eq:monoton}
        0 \leq C_\delta(m') \leq C_\delta(m) \leq \mathbbm{1} \quad \mbox{for all $0 \leq m' \leq m\leq 1 $.}
    \end{equation}
    By the Banach-Alaoglu theorem, the unit ball of 
    $ B(L^2(0,1)) $ is compact in the weak operator topology. Hence, there is a sequence $ \delta_n \downarrow 0 $ and a family of non-negative contractions $ C(m) $ for which the quadratic form $ C_{\delta_n}(m) $ converges -- first, for all $ m $ in a countable dense set of $ [0,1] $, and then uniquely extended to all $ m \in [0,1] $ by the monotonicity~\eqref{eq:monoton}, which is inherited by the limit.  The distorted path $\pi'$ is then defined as
    $$ \pi'(m) := (\varrho')^{\frac12} C(m) (\varrho')^{\frac12} $$
    and by construction $\pi' \in \Pi(\varrho').$ By Hölder's inequality for Schatten $ p $-norms $ \| \cdot \|_p $, we have
    \begin{align*} 
    \| \pi'(m) - \pi(m)\|_2 &= \| \sqrt{\varrho'} C(m) \sqrt{\varrho'}  - \sqrt{\varrho} C(m) \sqrt{\varrho} \|_2 \\ &\leq \| (\varrho')^{\frac12} C(m) \|_4 \| \sqrt{\varrho'}  - \sqrt{\varrho}\|_4  + \| \sqrt{\varrho} C(m) \|_4 \| \sqrt{\varrho'}  - \sqrt{\varrho}\|_4 \\
    & \leq \| \sqrt{\varrho'}  - \sqrt{\varrho}\|_4  ( \| \sqrt{\varrho'} \|_4 + \| \sqrt{\varrho} \|_4) =\| \sqrt{\varrho'}  - \sqrt{\varrho}\|_4  (\sqrt{\|\varrho \|_2} + \sqrt{\|\varrho' \|_2}),
    \end{align*}
    where we used that $C(m)$ is a contraction. It remains to bound $\| \sqrt{\varrho'}  - \sqrt{\varrho}\|_4 $. To this end we set $D \coloneqq \sqrt{\varrho'}  - \sqrt{\varrho} $ and $S \coloneqq \sqrt{\varrho'}  + \sqrt{\varrho} $ and note $S+D, S-D \geq 0$. Let us denote by $d_\alpha$ and $\psi_\alpha$ the eigenvalues and eigenvectors of the compact operator $D$. With the corresponding matrix elements $S_{\alpha \beta}$ we obtain
    \begin{align*}
        \| \varrho - \varrho' \|_2^2 = \frac14 \|SD + DS \|_2^2 = \frac14 \sum_{\alpha,\beta} |d_\alpha + d_\beta|^2 |S_{\alpha\beta}|^2  \geq \sum_{\alpha} d_\alpha^2 S_{\alpha\alpha}^2 \geq \sum_{\alpha} d_\alpha^4 = \| \sqrt{\varrho'}  - \sqrt{\varrho} \|^4_4  ,
    \end{align*}
    where the last inequality uses $S \pm D \geq 0$ which implies $S_{\alpha\alpha} \geq |d_\alpha|$. This completes the proof.  
\end{proof}

\subsection{Proof of Theorem~\ref{prop:diffident}}\label{S:diffident}

It seems difficult to give a direct proof for an analogue of Lemma~\ref{lem:replica} in the infinite-dimensional setup. We circumvent this issue by extending \eqref{eq:derivD} via a convexity argument. Compared to the usual argument, we rely on the following small twist that significantly simplifies the subsequent proof.

\begin{lemma}~\label{lem:convex}
    Let $G_n : \rr \to \rr$ be a sequence of convex functions with pointwise limit $G$. Let $F$ be a further differentiable convex function. Suppose that $G_n'(0)$ converges as $n \to \infty$ and that $G(x) \leq  F(x)$ and $G(0) = F(0)$. Then, $F'(0) = \lim_{n \to \infty} G_n'(0)$.
\end{lemma}
\begin{proof}
    $G$ is convex and thus the left and right derivatives exist and satisfy $G'_{-}(0) \leq G'_{+}(0)$. Considering difference quotients, the convexity implies $G'_{-}(0) \leq \lim_{n \to \infty} G_n'(0) \leq G'_{+}(0).$ The conditions $G(x) \leq  F(x)$ and $G(0) = F(0)$ further yield $G'_{-}(0) \geq F'(0)$ and $G'_{+}(0) \leq F'(0)$. This is only possible if $G'_{-}(0) = F'(0) = G'_{+}(0)$, that is $F'(0) = \lim_{n \to \infty} G_n'(0)$.
\end{proof}

We are now ready to spell out the proof of Theorem~\ref{prop:diffident}.
\begin{proof}[Proof of Theorem~\ref{prop:diffident} ]
    Let $ \varrho^D = \nabla \mathcal{F}^D(0) $ and $ \pi^D_*$ be a sequence of minimizers of 
    $$ \inf_{\pi \in \Pi^D(\varrho^D)} \mathcal{P}_1^D(\pi,0) = F^D(1;\varrho^D) - \frac{\beta^2}{4} \| \varrho^D \|_2^2 =: \widehat F^D(1,1;\varrho^D) , $$
    and likewise $ \widehat F^D(\lambda_1,\lambda_2;\varrho^D) \coloneqq F^D(\lambda_1;\varrho^D) -\frac{\lambda_2^2\beta^2}{4} \| \varrho^D \|_2^2 $.  By Proposition~\ref{prop:Parisiend}, we find a subsequence $D_k$ along which
    \begin{equation}\label{eq:subseqconpi}
    \lim_{k\to \infty} \int_0^1 \left\| \pi_{*}^{D_k}(m) - \pi_*(m) \right\|_2 dm = 0 , \quad\mbox{and}\quad \lim_{k\to\infty} \| \varrho^{D_k} - \varrho \|_2 =0 . 
    \end{equation}
    We now set
    $$ G(\lambda) \coloneqq \liminf_{k \to \infty} \widehat F^{D_k}(\lambda, 1;\varrho^{D_k}), $$
    and claim that it is enough to prove
    \begin{equation}\label{eq:key}
        G(\lambda) \leq \widehat{F}(\lambda, 1;\varrho)
    \end{equation} 
     Consider first the point $\lambda = 1$. By the second assertion of Proposition~\ref{prop:conc} and the corresponding Parisi formulae~\eqref{eq:Chen23} and~\eqref{QPsoc}, we have
     $$  \widehat F^{D_k}(1, 1;\varrho^{D_k}) = \lim_{N \to \infty }  \frac1N \mathbb{E} \ln \widehat{Z}^{D_k}_N, \qquad \widehat F(1, 1;\varrho) = \lim_{N \to \infty }  \frac1N \mathbb{E} \ln \widehat{Z}_N.   $$
     In view of \eqref{eq:ZDtoZ}, we obtain  $G(1) = \widehat{F}(1,1;\varrho)$. Using~\eqref{eq:derivD} and~\eqref{eq:subseqconpi}, we also conclude:
    $$  \lim_{k \to \infty} \frac{\partial F^{D_k}}{\partial \lambda}(1;\varrho^{D_k} ) = \frac{\beta^2}{2} \left( \|\varrho\|_2^2 - \int_{0}^{1} \|\pi^{*}(m) \|^2_{2} \, dm  \right). $$ 
    In view of Lemma~\ref{lem:convex}, the proof is thus complete once  \eqref{eq:key} is established. 

    The claim \eqref{eq:key} is a simple consequence of the constrained Parisi formulae, Proposition~\ref{prop:constrained}. 
    For  fixed $\lambda \in \mathbb{R}$, it yields
    $$ \widehat{F}(\lambda, 1;\varrho) = \left(\lambda^2 - 1 \right)\frac{\beta^2}{4} \|\varrho \|_2^2 + \inf_{x \in \mathcal{S}_2} \inf_{\pi \in \Pi(\varrho)} \left[   \mathcal{P}_{\lambda}\left(\pi, x \right) - \langle x, \varrho \rangle  \right],  $$
    and analogous formulae hold for the finite-dimensional approximands by~\eqref{eq:Panchenko0}. For any  $\varepsilon > 0$, we choose $x_{\varepsilon} \in \mathcal{S}_2 $ and $ \pi_{\varepsilon} \in \Pi(\varrho) $ which approximate the infima such that 
    $$
    \inf_{x \in \mathcal{S}_2} \inf_{\pi \in \Pi(\varrho)} \left[   \mathcal{P}_{\lambda}\left(\pi, x \right) - \langle x, \varrho \rangle  \right] \geq  \mathcal{P}_{\lambda}\left(\pi_{\varepsilon}, x_{\varepsilon} \right) - \langle x_{\varepsilon}, \varrho \rangle  - \varepsilon .
    $$
    Recall from~\eqref{eq:functionalconf} that $\lim_{D \to \infty} \mathcal{P}^{D}_{\lambda}\left(P_D\pi P_D, P_DxP_D \right)  = \mathcal{P}_{\lambda}\left(\pi, x \right)  $ for any fixed $x, \pi$, so that 
    \begin{align*}
        & \mathcal{P}_{\lambda}\left(\pi_{\varepsilon}, x_{\varepsilon} \right) - \langle x_{\varepsilon}, \varrho \rangle  = \lim_{k \to \infty} \left[\mathcal{P}^{D_k}_{\lambda}\left(P_{D_k}\pi_{\varepsilon} P_{D_k}, P_{D_k} x_{\varepsilon} P_{D_k} \right) - \langle P_{D_k} x_{\varepsilon} P_{D_k}\, \varrho^{D_k} \rangle \right] \\
        & \geq \liminf_{k \to \infty} \inf_{x \in \mathcal{S}_2^{D_k}} \inf_{\pi \in \Pi^{D_k}(P_{D_k}\varrho P_{D_k})} \left[   \mathcal{P}^{D_k}_{\lambda}\left(\pi, x \right) - \langle x, \varrho^{D_k} \rangle \right] \\
        &= \liminf_{k \to \infty} \inf_{x \in \mathcal{S}_2^{D_k}} \inf_{\pi \in \Pi^{D_k}(\varrho^{D_k}) }\left[   \mathcal{P}^{D_k}_{\lambda}\left(\pi, x \right) - \langle x, \varrho^{D_k} \rangle \right] .
    \end{align*}
    The final identity switches the endpoint. It follows from the observations $\| \varrho^{D_k} - P_{D_k} \varrho P_{D_k} \|_2 \to 0$, the estimate in Lemma~\ref{lem:switch} and the Lipschitz-continuity~\eqref{eq:Lipschitz}  of the Parisi functional $\mathcal{P}_{\lambda}(\pi,x)$. Recall that  $\mathcal{P}_{\lambda}^D(\pi ,x) = \mathcal{P}_{\lambda}(\pi,x)$ for $x = P_D x P_D$ and $\pi = P_D \pi P_D$ guarantees a uniform Lipschitz estimate in dimension $D$. This concludes the proof of~\eqref{eq:key}, and hence of the theorem.
\end{proof}
\section{Conditional second moment analysis via quantum Hopfield models}\label{sec:ProofPart1}
This section is dedicated to the proof of the main result, Theorem~\ref{thm:CRS}, on the self-overlap constraint QSK, on which the proof of the replica-symmetric part of our main result on the quantum AT line in the self-overlap-corrected QSK, Theorem~\ref{thm:main}, rests. The latter will be proven at the end of this section.
\subsection{Strategy for a proof of Theorem~\ref{thm:CRS}}
Under the self-overlap constraint, the second term in the exponential in~\eqref{def:soconstrp2} is approximately constant. For a proof of 
Theorem~\ref{thm:CRS} we may therefore restrict ourselves without loss of generality to the case $ \lambda_1 = \lambda_2\eqqcolon \lambda $, i.e., we show that for any  $0 <  \lambda \leq \frac{b}{\tanh(\beta b)} $: 
	\begin{equation}\label{eq:pressureSOC2}
	\lim_{\varepsilon \downarrow 0} \lim_{N\to \infty} \frac{1}{N} \ \mathbb{E} \left[ \ln  W_{N,\varepsilon}(\lambda,\lambda;\mu) \right] = 0 . 
	\end{equation}
The main idea for a proof of this is the second moment test. Instead of carrying out the second moment argument directly for $ \ln W_{N,\varepsilon}(\lambda,\lambda;\mu) $, we 
carry the argument for the discretized version using the projected paths 
 $ \bxi^D  =(\xi_1^D , \dots , \xi_N^D ) $ from Subsection~\ref{sec:approx}, and set
\begin{equation}\label{eq:defWapp}
	W_{N,\varepsilon}^D(\lambda)  \coloneqq  \int_{\mathcal{Q}_{N,\varepsilon}(\mu) } \mkern-10mu  \exp\left( -\int_0^1 \lambda\beta U\left(\bxi^D(t)\right) dt  - \frac{N \lambda^2\beta^2}{4} \left\| Q_N[\bxi^D] \right\|_2^2 \right) \nu_N(d\mathbf{\bxi}).
\end{equation}
A straightforward generalization of \cite[Prop. 3.2]{MW25} shows that,  in the thermodynamic limit $ N \to \infty $, when the discretization parameter $ D $ goes to infinity subsequently, the pressure corresponding to $ W_{N,\varepsilon}(\lambda,\lambda) $  is well approximated by that of $ W_{N,\varepsilon}^D(\lambda)  $. 
\begin{prop}\cite[cf.~Prop. 3.2]{MW25}\label{prop:approx1}
For any $ \varepsilon > 0 $ and all $ \lambda \in \mathbb{R} $: 
$$
\limsup_{D\to \infty} \lim_{N\to \infty} \frac{1}{N} \left|  \mathbb{E}\left[ \ln  W_{N,\varepsilon}(\lambda,\lambda) -  \ln W_{N,\varepsilon}^D(\lambda) \right] \right| = 0 . 
$$
\end{prop}

To carry out the second moment analysis for $ W_{N,\varepsilon}^D(\lambda)  $, we first note that its average is close to one for any~$ D $ in the thermodynamic limit. This follows from
\begin{equation}\label{eq:1stmom}
\mathbb{E}\left[ W_{N,\varepsilon}^D(\lambda) \right] = \nu_N\left( \mathcal{Q}_{N,\varepsilon}(\mu) \right) = 1 - \nu_N\left( \left\| Q_N - \mu \right\|_2 > \varepsilon \right)  ,
\end{equation}
together with the following concentration result. 
 \begin{prop}\label{prop:freeconc}
 For any $ N \in \mathbbm{N} $ and   $ \varepsilon > 0 $:
 \begin{equation}
 \nu_N\left( \left\| Q_N - \mu \right\|_2 > \varepsilon \right) \leq 2 \exp\left(- \frac{3 N}{2} \frac{\varepsilon^2}{3 \sigma^2 + 2 \varepsilon}  \right) 
 \end{equation}
 with $  \sigma^2 \coloneqq  \frac{1}{2} \left[ 1+ \left(\tanh(\beta b) \right)^2  - \frac{\tanh(\beta b)}{\beta b}   \right] > 0 $. 
 \end{prop}
 This results from an application of Pinelis's concentration \cite{Pin94} of independent vectors in a real Hilbert space.  We spell it out in Appendix~\ref{app:confree}.\\

The second moment method applied to $ W_{N,\varepsilon}^D(\lambda) $ also requires an upper bound on
\begin{align}\label{eq:Z2}
\mathbb{E}\left[W_{N,\varepsilon}^D(\lambda)^2\right]  = &   \int_{\mathcal{Q}_{N,\varepsilon}(\mu) }\int_{\mathcal{Q}_{N,\varepsilon}(\mu) } \exp\left( \frac{\lambda^2\beta^2}{2 N} \frac{1}{2^{2D} }\sum_{k,k'=1}^{2^D}  \left(\bxi(k') \cdot \bta(k) \right)^2 \right) \nu_N(d\mathbf{\bxi})\nu_N(d\bta)  \notag  \\
\leq & \int_{\mathcal{Q}_{N,\varepsilon}(\mu) } \int_{\mathcal{Q}_{N,\varepsilon}(\mu) }   \exp\left( \frac{\lambda^2\beta^2}{2 N 2^{D}} \sum_{k=1}^{2^D} \int_0^1 \left(\bxi(s) \cdot \bta(k) \right)^2 ds  \right) \nu_N(d\mathbf{\bxi}) \nu_N(d\bta) \notag \\
\leq & \int_{\mathcal{Q}_{N,\varepsilon}(\mu) } \int   \exp\left( \frac{\lambda^2\beta^2}{2 N 2^{D}} \sum_{k=1}^{2^D} \int_0^1 \left(\bxi(s) \cdot \bta(k) \right)^2 ds  \right) \nu_N(d\mathbf{\bxi}) \nu_N(d\bta)  .
\end{align}
Here and in the following $ \bxi(k') \cdot \bta(k) \coloneqq \sum_{j=1}^N \xi_j(k') \cdot \eta_j(k)  $ stands for the usual Euclidean scalar product. 
The second inequality is Jensen applied to the time averages $  \xi_j(k') = 2^D  \int_{t_{k'-1}}^{t_{k'}} \xi_j(s) ds $ inside the square.  The third inequality drops the restriction in the $ d\bxi $-integration.

The $ \nu_N(d\bxi)$-integration coincides up to a factor with the partition function of a biased quantum Hopfield model. In the next subsection, we use an adaptation of arguments in~\cite{Sherbina20} (see also \cite{Manai25}) to derive a variational expression for its pressure. This will subsequently be used to finish the second moment analysis. 

\subsection{Biased quantum Hopfield models}
The innermost integral  on the right-hand side of~\eqref{eq:Z2} represents the normalized partition function at inverse temperature $\beta >  0 $ of a  quantum Hopfield model with $M =  2^D  $ patterns $ \bta(k)  $ enumerated by $ k \in \{ 1, \dots , M \} $ and coupling constant 
$$
g\coloneqq \frac{\lambda^2\beta }{2} \geq  0 
$$
 that is, with Hamiltonian
 \begin{equation}\label{eq:HopfH}
 	 \widehat H_N^{(M)} \equiv  \widehat H_N^{(M)}[\bta]  \coloneqq - \frac{g}{ N M} \sum_{k=1}^M \Big( \sum_{j=1}^N   \eta_j(k)  \ S_j^z \Big)^2   -  b \sum_{j=1}^N S_j^x  \qquad \mbox{on}\quad  \bigotimes_{j=1}^N \mathbbm{C}^2 .
 \end{equation}
 Had the patterns $ (\eta_1(k), \dots , \eta_N(k) )  $ 
been independent and identically distributed, it would coincide with the standard quantum Hopfield model whose free energy is studied in \cite{Sherbina20}.  To determine the pressure 
\begin{align}
\widehat \varphi_N^{(M)}   \equiv \widehat \varphi_N^{(M)} [\bta]  \coloneqq &  \frac{1}{ N} \ln \tr e^{- \beta \widehat H_N^{(M)}[\bta] }  \\
= &  \frac{1}{ N}  \ln\Big[ (2 \cosh(\beta b) )^N      \int   \exp\left( \frac{\beta g}{N M} \sum_{k=1}^{M} \int_0^1 \left(\bxi(s) \cdot \bta(k) \right)^2 ds  \right) \nu_N(d\mathbf{\bxi}) \Big] \notag 
\end{align}
asymptotically as $ N \to \infty $ for $1_{\mathcal{Q}_{N,\varepsilon}(\mu)} \nu_N $-typical realizations of a large number $ M = 2^D $ of patterns, which are bounded, $ \eta_j(k)^2 \leq 1 $, we trace the steps in \cite{Sherbina20} and  derive a variational principle. The latter is expressed in terms of the pressure functional
\begin{equation}\label{eq:CWfunc}
\widehat \varphi_N^{(M)}(\boldsymbol{y}) \equiv \widehat \varphi_N^{(M)}(\boldsymbol{y})[\bta] \coloneqq \frac{1}{N} \sum_{j= 1}^N \ln 2 \cosh\sqrt{(\beta b)^2 + \left( \frac{2 \beta g}{M} \sum_{k=1}^M \eta_j(k) y(k) \right)^2} - \frac{\beta g}{M} \sum_{k=1}^M  y(k)^2 
\end{equation} 
with $ \boldsymbol{y} =(y(1), \dots , y(M)) \in \mathbb{R}^M $, which is of quantum Curie-Weiss type. 
\begin{prop}\label{prop:QHM}
Consider the Hopfield Hamiltonian~\eqref{eq:HopfH} with patterns $ \bta(1), \dots , \bta(M)  $ satisfying $$ \eta_j(k)^2 \leq 1 $$ for all $ j \in \{1, \dots , N \} $ and $k \in \{ 1, \dots , M \}$. Then for any $ \beta > 0  $, $ b , g\geq 0 $ there is some $ C < \infty $, which is independent of $ \bta $,  such that for any $ N , M \in \mathbb{N} $:
\begin{equation}\label{eq:HtoCW}
    \Big| \widehat  \varphi_N^{(M)}[\bta]  - \sup_{ \boldsymbol{y}} \widehat \varphi_N^{(M)}(\boldsymbol{y})[\bta]   \Big| \leq C \max\left\{  \frac{M}{\sqrt{N}} , \frac{1}{N^{1/3}} \right\} .
\end{equation}
\end{prop}
\begin{proof}
    We abbreviate the magnetization operator corresponding to the $ k$th pattern by 
    $$
    M(k) \coloneqq \sum_{j=1}^N \eta_j(k) S_j^{z} ,
    $$
    and introduce a collection $ \gamma(k)$, $ k \in \{ 1,\dots , M \} $, of independent and identically distributed standard normal random variables, which serve as an auxiliary random field in 
    \begin{equation}\label{Hopfwf}
    \widehat H_N^{(M)}(\gamma) \coloneqq \widehat H_N^{(M)} - \frac{1}{\sqrt{N}}  \sum_{k=1}^M \gamma(k) M(k) .
    \end{equation}
    If $ \psi_N^{(M)}(\gamma) \coloneqq N^{-1} \ln \tr e^{- \beta  \widehat H_N^{(M)}(\gamma) } $ denotes the corresponding pressure, we have $ \psi_N^{(M)}(0)= \widehat \varphi_N^{(M)} $. If $ \mathbb{E}_\gamma[\cdot] $ stands for the joint expectation value over the auxiliary random variables and $ \langle \cdot \rangle_\gamma $ denotes the Gibbs expectation value with respect to $ \widehat H_N^{(M)}(\gamma)$ at inverse  temperature $\beta$, the Peierls-Bogoliubov inequality guarantees that
    \begin{align}\label{eq:JBin}
        0 =   \frac{\beta}{\sqrt{N}^3 }  \mathbb{E}_\gamma\left[ \sum_{k=1}^M \gamma(k) \  \langle M(k)   \rangle_0 \right] & \leq  \mathbb{E}_\gamma\left[ \psi_N^{(M)}(\gamma)   -  \varphi_N^{(M)} \right]
         \leq \frac{\beta}{\sqrt{N}^3 } \ \mathbb{E}_\gamma\left[ \sum_{k=1}^M \gamma(k) \langle M(k) \rangle_\gamma\right] \notag \\
        & \leq \frac{\beta}{\sqrt{N}^3 }  \ \mathbb{E}_\gamma\left[ \sum_{k=1}^M \gamma(k)^2 \right]^{1/2} \mathbb{E}_\gamma\left[ \sum_{k=1}^M \langle M(k) \rangle_\gamma^2 \right]^{1/2} \! \leq\  \beta \frac{M}{\sqrt{N} }  .
    \end{align}
    Hence, if 
    $ M / \sqrt{N} $ is small, the addition of the random field in~\eqref{Hopfwf} does not change the pressure. 
    To study $ \psi_N^{(M)}(\gamma)  $, we now linearize the interaction and introduce
    \begin{align}\label{eq:compHam}
    \widehat H_N^{(M)}(\gamma,\boldsymbol{y}) \coloneqq &  \ \widehat H_N^{(M)}(\gamma) + \frac{g}{M N} \sum_{k=1}^M \left( M(k) -  N y(k)\right)^2  \\
    = & \ - \sum_{k=1}^M M(k) \left( \frac{2 gy(k)}{M}  +\frac{\gamma(k)}{\sqrt{N}} \right)   -  b \sum_{j=1}^N S_j^x +   \frac{gN }{M} \sum_{k=1}^M y(k)^2  . \notag 
    \end{align}
    Its pressure at inverse  temperature $\beta$ will be denoted $  \psi_N^{(M)}(\gamma, \boldsymbol{y}) $ such that  $  \psi_N^{(M)}(0, \boldsymbol{y})  = \widehat \varphi_N^{(M)}(\boldsymbol{y}) $, cf.~\eqref{eq:CWfunc}.
    Analogously to~\eqref{eq:JBin}, the Peierls-Bogoliubov inequality guarantees that the pressure is unchanged by the presence of the random field:
    \begin{equation}\label{eq:JBin2}
  \mathbb{E}_\gamma\left[  \sup_{ \boldsymbol{y}} \left| \psi_N^{(M)}(\gamma, \boldsymbol{y})   -  \psi_N^{(M)}(0, \boldsymbol{y}) \right| \right]\leq \beta   \frac{M}{\sqrt{N} }  .
    \end{equation}
    Another application of the Peierls-Bogoliubov inequality compares the effect of the addition of the last term in~\eqref{eq:compHam} on the pressure,
$$
    	0 \leq  \psi_N^{(M)}(\gamma)  -  \psi_N^{(M)}(\gamma, \boldsymbol{y})  \leq  \frac{\beta g}{M N^2} \sum_{k=1}^M \langle \left( M(k) -  N y(k)\right)^2 \rangle_\gamma ,
$$
    where the expectation value on the right-hand side is with respect to the Gibbs measure of $ H_N^{(M)}(\gamma)$. Optimizing over $  \boldsymbol{y} \in \mathbb{R}^M $ yields
    \begin{equation}
    	0 \leq  \psi_N^{(M)}(\gamma)  -  \sup_{ \boldsymbol{y}}  \psi_N^{(M)}(\gamma, \boldsymbol{y})  \leq  \frac{\beta g}{M N^2}  \sum_{k=1}^M \langle \left( M(k) -  \langle M(k) \rangle_\gamma \right)^2 \rangle_\gamma .
    \end{equation}
    As in~\cite{Sherbina20}, we now split the last term into two contributions. To do so, we evaluate the Gibbs expectation value $ \langle \cdot \rangle_\gamma $ in terms of the eigenvalues $ (E_j ) $ and corresponding normalized eigenvectors $ (\psi_j) $ of $ -\beta \widehat{H}_N^{(M)}(\gamma) $. Abbreviating the partition function by $ Z \coloneqq \tr e^{- \beta \widehat{H}_N^{(M)}(\gamma)} $ and 
    $$
    A_{ij}(k)  \coloneqq \langle \psi_i , \left(M(k) -  \langle M(k) \rangle_\gamma \right)  \psi_j \rangle ,
    $$
    we rewrite
	\begin{equation}\label{eq:eigenexp}
    \sum_{k=1}^M \langle \left( M(k) -  \langle M(k) \rangle_\gamma\right)^2 \rangle_\gamma  =  \frac{1}{Z}   \sum_{k=1}^M \sum_{i,j}    \frac{1}{2} \left( e^{E_i} +  e^{E_j} \right)    A_{ij}(k)^2 . 
    	\end{equation}
    The inequality
    $$
     \frac{1}{2} \left( e^{s} +  e^{s'} \right) \leq  \frac{ e^{s} -  e^{s'}}{s-s'} +  \frac{1}{2}  \left| e^{s} -  e^{s'} \right| , 
    $$
     is valid for all $ s, s' \in \mathbb{R} $ (even at $ s= s' $ when interpreting the first term on the right-hand side as a derivative at $ s $, i.e. as $ e^s $). Applied to~\eqref{eq:eigenexp}, we hence bound~\eqref{eq:eigenexp} by a sum of two terms. 
    The first equals
    \begin{equation}\label{eq:1st}
    	 \frac{1}{Z}   \sum_{k=1}^M \sum_{i,j}  \frac{e^{E_i} -  e^{E_j}}{E_i - E_j}    A_{ij}(k)^2  =  \frac{N^2}{\beta^2}  \sum_{k=1}^M   \frac{\partial^2 \psi_N^{(M)}(\gamma)}{\partial \gamma(k)^2} .
    \end{equation}
    The second term equals
    \begin{align}
     &  \frac{1}{ 2 Z}   \sum_{k=1}^M \sum_{i,j}   \left| e^{E_i} -  e^{E_j} \right|     A_{ij}(k)^2  \notag \\
      & \leq    \left( \frac{1}{2 Z}   \sum_{k=1}^M \sum_{i,j}   \frac{e^{E_i} -  e^{E_j}}{E_i - E_j}   A_{ij}(k)^2   \right)^{2/3}  \left( \frac{1}{ 2Z}   \sum_{k=1}^M \sum_{i,j}   \left( e^{E_i} +  e^{E_j} \right) (E_i - E_j)^2    A_{ij}(k)^2   \right)^{1/3} .
           \end{align}
           The last step is an application of H\"older's inequality, together with the estimate $  \left| e^{E_i} -  e^{E_j} \right| \leq  e^{E_i} +  e^{E_j}$ for the second term. 
           By~\eqref{eq:1st},  the first bracket equals
      	$ \frac{N^2}{2 \beta^2}  \sum_{k=1}^M   \frac{\partial^2\psi_N^{(M)}(\gamma)}{\partial \gamma(k)^2} $.  An explicit computation shows that the term in the second bracket is related to a Gibbs expectation of a commutator
	\begin{align}
		& \frac{1}{  2  Z}   \sum_{k=1}^M \sum_{i,j}   \left( e^{E_i} +  e^{E_j} \right) (E_j - E_i)^2    A_{ij}(k)^2 =  - \sum_{k=1}^M \left\langle [ M(k) , \beta \widehat H_N^{(M)}(\gamma) ]^2 \right\rangle_\gamma \notag  \\
		& = - \sum_{k=1}^M   \left\langle \left( \sum_{j=1}^N \beta b  \eta_j(k) [ S_j^z , S_j^x  ] \right)^2 \right\rangle_\gamma = 4 \beta^2 b^2   \sum_{k=1}^M  \sum_{i,j=1}^N \eta_i(k) \eta_j(k)  \left\langle S_i^y S_j^y  \right\rangle_\gamma  \leq 4 \beta^2 b^2 M N^2 . 
	\end{align}
	To bound the first term, we average over the Gaussian field and use Gaussian integration by parts
	    \begin{align}
	    	& \frac{N^2}{\beta^2}  \sum_{k=1}^M \mathbb{E}_\gamma\left[  \frac{\partial^2 \psi_N^{(M)}(\gamma)}{\partial \gamma(k)^2} \right] = \frac{N^2}{\beta^2}  \sum_{k=1}^M  \mathbb{E}_\gamma\left[ \gamma(k)  \frac{\partial \psi_N^{(M)}(\gamma)}{\partial \gamma(k)}  \right]  \notag \\
		& = \frac{N}{\beta} \sum_{k=1}^M   \mathbb{E}_\gamma\left[ \frac{\gamma(k) }{\sqrt{N}} \langle M(k) \rangle_\gamma \right] \leq \frac{\sqrt{N}}{\beta} \sum_{k=1}^M  \mathbb{E}_\gamma\left[ \gamma(k)^2\right]^{1/2}  \mathbb{E}_\gamma\left[ \left( \sum_{j=1}^N \eta_j(k) \langle S_j^z \rangle_\gamma \right)^2 \right]^{1/2} \leq \frac{N^{3/2} M}{\beta} . 
	    \end{align}
	    In summary, we have thus proved that
	    \begin{align}\label{eq:Jlast}
	    	 \left|  \mathbb{E}_\gamma\left[  \psi_N^{(M)}(\gamma) \right] - \mathbb{E}_\gamma\left[ \sup_{ \boldsymbol{y}}  \psi_N^{(M)}(\gamma,\boldsymbol{y}) \right]  \right|  \leq  &   \frac{\beta g}{M N^2}  \left[ \frac{N^{3/2} M}{\beta}  + \left( \frac{N^{3/2}M}{2 \beta}  \right)^{2/3} \left( 4 \beta^2b^2 M N^2 \right)^{1/3}\right]  \notag \\
		   = &  \frac{g}{\sqrt{N} } + \frac{\beta gb^{2/3} }{N^{1/3}}  .
	    \end{align}
	    The above is now used to estimate
	    \begin{align}
	    \left|  \widehat \varphi_N^{(M)} - \sup_{ \boldsymbol{y}} \widehat \varphi_N^{(M)}(\boldsymbol{y}) \right|  & \leq   \left|   \widehat  \varphi_N^{(M)} - \mathbb{E}_\gamma\left[ \psi_N^{(M)}(\gamma)\right]  \right| \notag  \\
	&    +    \mathbb{E}_\gamma\left[ \sup_{ \boldsymbol{y}} \left| \psi_N^{(M)}(\gamma, \boldsymbol{y})   -  \psi_N^{(M)}(0, \boldsymbol{y}) \right|  \right]  + \left|  \mathbb{E}_\gamma\left[ \psi_N^{(M)}(\gamma)\right]  - \mathbb{E}_\gamma\left[ \sup_{ \boldsymbol{y}} \psi_N^{(M)}(\gamma,\boldsymbol{y}) \right]  \right|  \notag  \\
	& \leq  \frac{2\beta M}{\sqrt{N}} +   \frac{g}{\sqrt{N} } +\frac{\beta gb^{2/3} }{N^{1/3}} ,
	    \end{align}
	    where the last line is a combination of~\eqref{eq:JBin}, \eqref{eq:JBin2}, and~\eqref{eq:Jlast}. This completes the proof. 
\end{proof}

It remains to establish an upper bound on the variational problem in~\eqref{eq:HtoCW}. This is achieved in the following, which makes use of the empirical self-overlap matrix 
$$ Q_N^{(M)}(k,k') \coloneqq \frac{1}{N} \sum_{j=1}^N  \eta_j(k)  \eta_j(k') $$ 
of the patterns. 
We will denote the operator norm of this symmetric matrix by~$ \big\| Q_N^{(M)} \big\| $, occasionally dropping its dependence on $ \bta $. 
\begin{lemma}\label{lem:morm}
For any $ N , M \in \mathbb{N} $ and patterns $ \bta = ( \bta(1), \dots , \bta(M) )  \in \mathbb{R}^{M\times N} $: 
\begin{equation}\label{eq:QCWvar}
 \sup_{ \boldsymbol{y}} \widehat{\varphi}_N^{(M)}(\boldsymbol{y})[\bta]    \leq \sup_{x \geq 0 }  \left[ \ln 2 \cosh\sqrt{(\beta b)^2 + x  \frac{4 \beta g}{M} \big\| Q_N^{(M)}[\bta]   \big\|  } - x  \right] 
\end{equation}
In case 
\begin{equation}\label{eq:para}
  \frac{2 g\tanh(\beta b)}{b   } \frac{\big\| Q_N^{(M)}  \big\|}{M} \leq 1 , 
 \end{equation}
the supremum on the right side is a maximum at $ x = 0 $ and hence equals $ \ln \left( 2 \cosh(\beta b)\right)  $. 
\end{lemma}
\begin{proof}
	The function $h:  (0,\infty) \ni  \xi  \mapsto \ln \cosh \sqrt{\xi }  $ is concave as can be checked by calculating its second derivative, $ h''(\xi) = \big( 1 - \sinh(2\sqrt{\xi}) / (2 \sqrt{\xi}) \big) / (4 \xi \cosh^2(\sqrt{\xi}) ) \leq 0 $. We may therefore pull the empirical average inside and estimate using Jensen's inequality
	\begin{align}
		 \widehat{\varphi}_N^{(M)}(\boldsymbol{y})  \leq  & \ln 2 \cosh\sqrt{(\beta b)^2 +  \frac{1}{N} \sum_{j= 1}^N \left( \frac{2\beta g}{M} \sum_{k=1}^M \eta_j(k) y(k) \right)^2} - \frac{\beta g}{M}  \sum_{k=1}^M  y(k)^2 \notag \\
	 =  & \ln 2 \cosh\sqrt{(\beta b)^2 +  \left( \frac{2\beta g}{M}  \right)^2 \langle \boldsymbol{y} ,  Q_N^{(M)} \boldsymbol{y} \rangle } -\frac{\beta g}{M} \| \boldsymbol{y} \|^2  .
	\end{align}
	The bound~\eqref{eq:QCWvar} then follows from the estimate $  \langle \boldsymbol{y} ,  Q_N^{(M)} \boldsymbol{y} \rangle  \leq  \big\| Q_N^{(M)} \big\|  \| \boldsymbol{y} \|^2 $. 
	
	For a proof of the second claim, we compute the derivative of $ h:  (0,\infty) \ni  \xi  \mapsto \ln \cosh \sqrt{a^2 + b^2 \xi }   - \xi $ with $ a,b \in \mathbb{R} $, which equals $ h'(\xi) = \frac{b^2}{2 \sqrt{a^2 + b^2 \xi} } \tanh \sqrt{a^2 + b^2 \xi}  - 1 $. Since $ h $ is concave by the same reasoning as above, it follows that $ h'(\xi) \leq h'(0) = \frac{b^2 \tanh a }{2 a } - 1 \leq 0 $ in case~\eqref{eq:para} holds. 
\end{proof}

\subsection{Proof of Theorem~\ref{thm:CRS}}\label{sec:ProofCRS}
By Proposition~\ref{prop:approx1},  it remains to investigate the pressure of the approximand, 
$$ \varphi_{N,\varepsilon}^{D}(\lambda)  \coloneqq \frac{1}{N} \ln W_{N,\varepsilon}^D(\lambda)  .
$$
Jensen's inequality together with~\eqref{eq:1stmom} implies
\begin{equation}\label{eq:secondmomstart}
 \limsup_{N\to \infty }  \mathbb{E}\left[ \varphi_{N,\varepsilon} ^{D}(\lambda)  \right]  \leq  \limsup_{N\to \infty }  \frac{1}{N} \ln  \mathbb{E}\left[ W_{N,\varepsilon}^D(\lambda) \right] \leq 0 .
\end{equation}
For a matching lower bound, we use the following variant of the standard second-moment bound, cf.~\cite{Tal11b,Leschke:2021aa}.
\begin{prop}
For any $ \delta > 0 $ and $ l \in [0,1] $ for which 
\begin{equation}\label{ass:PZ}
 e^{-l \delta N}  \leq  \mathbb{E}\left[ W_{N,\varepsilon}^D(\lambda)\right]  
 \end{equation}
 at fixed parameters of $ W_{N,\varepsilon}^D(\lambda) $, we have 
\begin{equation}\label{eq:PZ}
\exp\left( N \left( \mathbb{E}\left[\varphi_{N,\varepsilon}^{D}(\lambda) \right] + \delta \right)   \right) \geq \frac{\left( \mathbb{E}\left[ W_{N,\varepsilon}^D(\lambda)\right]  - e^{-l \delta N} \right)^2}{\mathbb{E}\left[ W_{N,\varepsilon}^D(\lambda)^2\right] } - \mathbb{P}\Big( \left| \varphi_{N,\varepsilon}^{D}(\lambda)  - \mathbb{E}\left[ \varphi_{N,\varepsilon}^{D}(\lambda) \right]  \right| >  (1-l) \delta  \Big) . 
\end{equation}
\end{prop}
\begin{proof}
	Since the exponent in the left side is non-random, we may lower bound it by 
	\begin{align} 
	 1\left[\mathbb{E}\left[\varphi_{N,\varepsilon}^{D}(\lambda)  \right] + \delta \geq 0 \right] &  = \mathbb{P}\left( \mathbb{E}\left[\varphi_{N,\varepsilon}^{D}(\lambda) \right] + \delta \geq 0 \right) \notag \\
	 & \geq\mathbb{P}\left( \varphi_{N,\varepsilon}^{D}(\lambda) + l \delta  \geq 0 \right)  - \mathbb{P}\left( \varphi_{N,\varepsilon}^{D}(\lambda)- \mathbb{E}\left[ \varphi_{N,\varepsilon}^{D}(\lambda) \right]  \geq (1-l) \delta \right) .
	 \end{align}
	 The inequality, which is valid for all $ l \in [0,1] $ is based on inclusion-exclusion and writing $  \mathbb{E}[ \varphi_{N,\varepsilon}^{D}(\lambda)  ]   =  \mathbb{E}[\varphi_{N,\varepsilon}^{D}(\lambda)] - \varphi_{N,\varepsilon}^{D}(\lambda) + \varphi_{N,\varepsilon}^{D}(\lambda) $. 	 Under the assumption~\eqref{ass:PZ} we may now apply the Paley-Zygmund inequality to the first term in the right side:
	 $$
	 \mathbb{P}\left(  W_{N,\varepsilon}^D(\lambda) \geq e^{- l \delta N}   \right) \geq \left( 1- \frac{  e^{-l \delta N}  }{ \mathbb{E}\left[  W_{N,\varepsilon}^D(\lambda)\right] } \right)^2 \frac{ \mathbb{E}\left[  W_{N,\varepsilon}^D(\lambda)\right]^2}{ \mathbb{E}\left[ W_{N,\varepsilon}^D(\lambda)^2\right] } .
	 $$
	 This completes the proof. 
\end{proof}

  In our set-up, the assumption~\eqref{ass:PZ} holds for any $l \delta > 0$ at sufficiently large $N$ by Proposition~\ref{prop:freeconc}. 

To bound the second term in~\eqref{eq:PZ}, we use the following standard Gaussian concentration estimate, which is based on the fact that the Lipschitz constant of the pressure $ \varphi_{N,\varepsilon}^{D}(\lambda) $ as a function of the basic random variables $ ( g_{jk} ) $ in~\eqref{eq:SK} is bounded by $ \lambda \beta $. 
\begin{prop}[cf. \cite{Crawford:2007aa,MW25}] \label{prop:Gaussc}
There are $ c, C \in (0,\infty) $ such that 
for any $ \varepsilon,  b  , \lambda \geq 0 $, $ t , \beta > 0 $ and all $ N , D \in \mathbb{N} $ 
\begin{equation}
	\mathbb{P}\left( \left| \varphi_{N,\varepsilon}^{D}(\lambda)  - \mathbb{E}\left[ \varphi_{N,\varepsilon}^{D}(\lambda) \right]  \right| > \beta\lambda  t  \right) \leq C e^{- c t^2 N} 
\end{equation}
\end{prop}

We are now ready to complete the proof of Theorem~\ref{thm:CRS}. 
\begin{proof}[Proof of Theorem~\ref{thm:CRS}, Eq.~\eqref{eq:pressureSOC}] It suffices to treat the case $ \lambda_1 = \lambda_2 = \lambda $. By the continuity of the free energy in $\lambda$, we may restrict ourselves to $ 0 <  \lambda < \frac{b}{\tanh(\beta b)} $ away from the critical line. 

Let $ \delta > 0 $ be arbitrary. By Proposition~\ref{prop:approx1} and \eqref{eq:secondmomstart} as well as monotonicity in $ \varepsilon $, it remains to show that for some $ \varepsilon > 0 $ small enough:
\begin{equation}\label{eq:aim}
	\liminf_{D\to \infty} \liminf_{N\to \infty} \mathbb{E}\left[  \varphi_{N,\varepsilon}^{D}(\lambda)  \right]  \geq - \delta . 
\end{equation}
We will pick
\begin{equation}
0 < \varepsilon < \frac{\tanh(\beta b) }{\lambda^2 \beta b} \left[  \left( \frac{b}{\tanh(\beta b)}\right)^2 - \lambda^2\right] .
\end{equation}
The reason for this choice is the following implied estimate 
on the operator norm of the empirical self-overlap matrix, which features in Lemma~\ref{lem:morm}. In the present case, we have
$$
\frac{1}{2^D} Q_N^{(2^D)}(k,k' ) = \frac{1}{N} \sum_{j=1}^N \langle e_k , \eta_j\rangle \langle \eta_j , e_{k'} \rangle , \quad k , k' \in \{1, \dots , 2^D\} ,
$$
where $ \langle e_k , \eta_j\rangle  $ is the scalar product in $ L^2(0,1) $, cf.~\eqref{eq:unitv}. Since the right side equals the projection of $ Q_N[\bta] $  from~\eqref{def:SO} to the subspace corresponding to $ P_D $, its operator norm is estimated as follows:
\begin{align}
\frac{1}{2^D} \left\| Q_N^{(2^D)} \right\| & = \left\| P_D Q_N[\bta] P_D  \right\| \leq  \left\| P_D \mu P_D  \right\| +  \left\| P_D\left(  Q_N[\bta] - \mu\right)  P_D  \right\| \notag \\
& \leq  \left\| \mu  \right\| +  \left\| Q_N[\bta] - \mu \right\|_2 . 
\end{align} 
On the event $ \mathcal{Q}_{N,\varepsilon}(\mu) $, the last term is bounded by $ \varepsilon $. Since $ \mu $ is a symmetric integral operator on $ L^2(0,1) $, its operator norm is bounded with the help of the Schur test:
\begin{equation}
  \left\| \mu  \right\| \leq \sup_{t \in (0,1)} \int_0^1 \left| \mu(t,s)\right| ds =  \int_0^1  \frac{\cosh\left(\beta b ( 1 - 2 |s|) \right) }{\cosh \beta b} ds = \frac{\tanh(\beta b)}{\beta b } ,
\end{equation}
where the last step follows by explicit computation~\cite{Leschke:2021aa}.

Consequently, for any $ \bta $ in the support of the measure $\nu_N $ restricted to $  \mathcal{Q}_{N,\varepsilon}(\mu)  $, the matrix $ Q_N^{(2^D)}(k,k' ) $ satisfies the condition~\eqref{eq:para} with $ M=2^D $ and $ g= \beta\lambda^2 / 2 $.  Combining~\eqref{eq:Z2} with Proposition~\ref{prop:QHM} as well as Lemma~\ref{lem:morm}, we thus conclude 
\begin{align}
\mathbb{E}\left[W_{N,\varepsilon}^D(\lambda)^2\right] & \leq \int_{ \mathcal{Q}_{N,\varepsilon}(\mu)} e^{N \left(\widehat \varphi_N^{(2^D)}[\bta] - \ln \left( 2 \cosh(\beta b)\right) \right)} \nu_N(d\bta) \leq e^{ C N \max\left\{  \frac{2^D}{\sqrt{N}} , \frac{1}{N^{1/3}} \right\}  }  \int_{ \mathcal{Q}_{N,\varepsilon}(\mu)} \nu_N(d\bta)  \notag \\
& \leq \exp\left(  C \max\left\{  2^D \sqrt{N} , N^{2/3} \right\}  \right)  ,
\end{align}
where $ C < \infty $ is the constant from  Proposition~\ref{prop:QHM}. We now plug this estimate into the Paley-Zygmund bound~\eqref{eq:PZ} in which we choose $ l = 1/2 $. Clearly, \eqref{ass:PZ} is then satisfied for all $ N $ large enough by Proposition~\ref{prop:freeconc}. Moreover, Proposition~\ref{prop:Gaussc} guarantees that the negative term in the right side of~\eqref{eq:PZ} is exponentially small in $ N $ for any $ \delta > 0 $. Consequently,~\eqref{eq:PZ}  implies that for any fixed $ D $:
\begin{equation}
	\liminf_{N\to \infty}  \mathbb{E}\left[ \varphi_{N,\varepsilon}^{D}(\lambda) \right]  \geq - \delta . 
\end{equation}
Since $ \delta > 0 $ was arbitrary, this concludes the proof of~\eqref{eq:aim}. 
\end{proof}

\subsection{Proof of Theorem~\ref{thm:main} -- Part~\ref{2}}

We start the proof by establishing the equality of the quenched and annealed pressure in the thermodynamic limit.
 \begin{proof}[Proof of Theorem~\ref{thm:main}, Eq.~\eqref{eq:part1}] 
 Since $ \widehat Z_N  \geq W_{N,\varepsilon}(1,1) $, we conclude from Theorem~\ref{thm:CRS} that under the condition $ b > \tanh(\beta b) $:
 \begin{align}\label{eq:claim1}
 0 = \lim_{\varepsilon \downarrow 0} \lim_{N\to \infty} \frac{1}{N} \ \mathbb{E} \left[ \ln  W_{N,\varepsilon}(1,1) \right] 
 \leq  \lim_{N\to \infty} \frac{1}{N} \ \mathbb{E} \left[ \ln  \widehat Z_N\right] \leq  \lim_{N\to \infty} \frac{1}{N} \ \ln \mathbb{E}\left[ \widehat Z_N\right] = 0 . 
 \end{align}
By continuity, this extends to the critical line $ b = \tanh(\beta b).$ This completes the proof of~\eqref{eq:part1}. 
 \end{proof}
 
 In a second step, we conclude from Proposition~\ref{prop:conc}  that the pressure corresponding to the partition function
 \begin{equation}
 \widehat Z_{N,\varepsilon} \coloneqq \int_{\mathcal{Q}_{N,\varepsilon}(m_N)}\mkern-10mu  \exp\left( -\int_0^1 \beta U\left(\bxi(t)\right) dt  - \frac{\beta^2N }{4} \left\| Q_N[\bxi] \right\|_2^2\right) \nu_N(d\mathbf{\bxi}) ,
 \end{equation}
 in which we restrict the self-overlap around its mean 
 $ m_N \coloneqq  \mathbb{E}\left[\langle Q_N \rangle_{N}\right]   $, 
 coincides with that of the unrestricted partition function $  \widehat Z_{N} $.
 \begin{lemma}\label{lem:restralso0}
 For any $ \varepsilon > 0 $: 
 $\quad \displaystyle
 \limsup_{N\to \infty} \frac{1}{N} \left| \mathbb{E}\left[ \ln  \widehat Z_{N,\varepsilon} -  \ln \widehat Z_{N}\right] \right| = 0 $.
 \end{lemma} 
 \begin{proof}
 By standard Gaussian concentration of measure, there are universal constants $ C, c \in (0,1) $ such that for any $ \delta > 0 $, and all $ N \in \mathbb{N} $, $ \varepsilon > 0 $:
\begin{equation}\label{eq:conc2}
 \mathbb{P}\left(  \frac{1}{N} \left| \ln  \widehat Z_{N,\varepsilon} -  \mathbb{E}\left[ \ln  \widehat Z_{N,\varepsilon}\right]  \right| > \beta \delta  \right) \leq C e^{-c  \delta^2 N} ,
 \end{equation}
 and similarly for $   \widehat Z_{N}  $ instead of $   \widehat Z_{N,\varepsilon} $, cf.~Proposition~\ref{prop:Gaussc}. 
 
 Let  $ \delta > 0 $ be arbitrary, and consider the event $ \mathcal{G}_{N,\varepsilon}(\delta) $ on which (i)~the complement of the event in~\eqref{eq:conc2} as well as its analogue for $   \widehat Z_{N}  $  hold, and (ii)~$ \langle \| Q_N - m_N \|_2 \rangle_N \leq \varepsilon /2 $. By~\eqref{eq:conc2} as well as three applications of Markov's inequality, we have
 $$
 \mathbb{P}\left( \mathcal{G}_{N,\varepsilon}(\delta) \right) \geq 1 - 2 C e^{-c  \delta^2 N}  - \frac{2}{\varepsilon} \ \mathbb{E}\left[ \langle \| Q_N - m_N \|_2 \rangle_N  \right] . 
 $$ 
Hence, using Proposition~\ref{prop:conc}, for any $   \varepsilon, \delta  > 0 $, the event is asymptotically of full probability as $ N \to \infty $. On $  \mathcal{G}_{N,\varepsilon}(\delta) $, we may bound
$$
\frac{1}{N} \left| \mathbb{E}\left[ \ln  \widehat Z_{N,\varepsilon} -  \ln \widehat Z_{N}\right] \right|  \leq 2 \delta \beta + \frac{1}{N} \ln \frac{ \widehat Z_{N}}{\widehat Z_{N,\varepsilon}} \leq  2 \delta \beta + \frac{\ln 2 }{N} . 
$$
The last step uses $$ 1 -  \frac{ \widehat Z_{N,\varepsilon}}{\widehat Z_{N}} =  \langle 1[  \| Q_N - m_N \|_2 > \varepsilon] \rangle \leq  \varepsilon^{-1}  \langle \| Q_N - m_N \|_2 \rangle_N  \leq \frac{1}{2} , 
$$
by another application of Markov's inequality. Since $ \delta > 0 $ is arbitrary, this concludes the proof. 
 \end{proof} 
 
 We continue with the proof of Theorem~\ref{thm:main}. 
 \begin{proof}[Proof of Theorem~\ref{thm:main}, Part~\ref{2} (cont.)]
 We conclude from~\eqref{eq:claim1} and Lemma~\ref{lem:restralso0} that
 \begin{align}
 0 = \lim_{N\to \infty} \frac{1}{N}  \mathbb{E}\left[ \ln  \widehat Z_{N,\varepsilon} \right] & \leq  \lim_{N\to \infty} \frac{1}{N}  \ln \mathbb{E}\left[  \widehat Z_{N,\varepsilon} \right] =  \lim_{N\to \infty} \frac{1}{N}  \ln \nu_N\left( \mathcal{Q}_{N,\varepsilon}(m_N) \right)  \notag \\
 & \leq  \lim_{N\to \infty} \frac{1}{N}  \ln \nu_N\left( \left\| Q_N - \mu \right\|_2  > \left( \| m_N - \mu \|_2  - \varepsilon \right)_+ \right) \notag \\
 & \leq -   \lim_{N\to \infty} \frac{3}{2} \frac{\left( \| m_N - \mu \|_2  - \varepsilon \right)_+^2}{3 \sigma^2 + 2 \left( \| m_N - \mu \|_2  - \varepsilon \right)_+}  \leq 0 .
 \end{align}
 The first step uses Jensen's inequality. The second follows from the triangle inequality for the norm,  $ \| Q_N - \mu \|_2 \geq \| m_N - \mu \|_2 - \| Q_N - m_N \|_2 $, and $ ( \cdot)_+ $ stands for the positive part. The last inequality uses Proposition~\ref{prop:freeconc}. Since $ \varepsilon > 0 $ is arbitrary, we hence conclude that 
 \begin{equation}\label{eq:selfmum}
  \lim_{N\to \infty}  \| m_N - \mu \|_2 = 0 .
 \end{equation}
  It thus remains to establish~\eqref{eq:noreplica}. From Theorem~\ref{thm:CRS} with $ \lambda_1 = \lambda_2 = 1 $, we conclude that 
  $$
 \limsup_{N\to \infty}  \mathbb{E}\left[ \langle \langle \left\| R_N \right\|_2^2 1[ \mathcal{Q}_{N,\varepsilon}(\mu) \times   \mathcal{Q}_{N,\varepsilon}(\mu) ]  \rangle \rangle_{N}  \right]  \leq \lim_{N\to \infty}  \mathbb{E}\left[ \langle \langle \left\| R_N \right\|_2^2  \rangle \rangle_{N,\varepsilon;\mu}^{(1,1)}  \right]  \to 0  \quad \mbox{as $ \varepsilon \downarrow 0 $.} 
  $$
  The left side still depends on the restriction of the duplicated Gibbs measure corresponding to~\eqref{eq:probmeas} to paths $ \bxi , \bxi' \in \mathcal{Q}_{N,\varepsilon}(\mu) $. Since $  \left\| R_N \right\|_2^2 \leq 1 $, the claim~\eqref{eq:noreplica} thus follows from the fact that for any $ \varepsilon > 0 $
  $$
   \lim_{N\to \infty}  \mathbb{E}\left[  \langle 1[ \| Q_N - \mu \|_2 > \varepsilon ] \rangle_N \right] \leq \varepsilon^{-1}  \lim_{N\to \infty}  \mathbb{E}\left[ \langle \| Q_N - \mu \|_2  \rangle_N \right] = 0 
   $$
   by \eqref{eq:selfmum} and Proposition~\ref{prop:freeconc}. 
 \end{proof}

\section{Quantum Toninelli argument and the order parameter in the spin-glass regime}\label{sec:proofpart2}
For a proof of the assertion on the free energy in part~\ref{1} of Theorem~\ref{thm:main}, we adapt the strategy of Toninelli \cite{Toninelli02} to the variational problem in~\eqref{QPsoc}.  
We consider the 
$1$-step replica symmetry breaking classical path $ \pi_c: [0,1] \to [0,\infty) $, 
$$\pi_c(n) \coloneqq \begin{cases} 0 , & n \in [0, m), \\
    q , & n \in [m,1] ,
\end{cases}
$$ 
with parameters $ m \in (0,1] ,  q \geq 0  $. Inserting this path into the quantum Parisi functional~\eqref{eq:Parisifctl}, we obtain
\begin{align}\label{def:PMq}
P(m,q) \coloneqq & \mathcal{P}(\pi_c,0)   = X_0(m,q)  + \frac{\beta^2}{4} (1-m) q^2 \\
 & \mbox{with}\quad X_0(m,q) \coloneqq\frac{1}{m} \ln \mathbb{E}\left[I(q)^m \right], \notag \\ & \mbox{and}\quad  I(q)  \coloneqq  
\int e^{S_q(\xi)} \nu_1(d\xi) , \quad
S_q(\xi) \coloneqq \sqrt{q} \beta w \ \langle 1 , \xi \rangle  - \frac{q \beta^2}{2} \langle 1 , \xi \rangle^2 \notag ,
\end{align}
where $w$ is a standard, centered Gaussian random variable. Recall from~\eqref{eq:scalarprod} that we abbreviate the real $ L^2$-scalar product of $ \xi $ and the constant function $ 1 $ by
$
\langle 1 , \xi \rangle $.

The annealed functional agrees with $ P(1,q) = 0 $ for any $ q \geq 0 $. As in~\cite{Toninelli02}, we now investigate the behavior of $ P(m,q) $ in the vicinity of $ m = 1$ by calculating the partial derivatives. 

\subsection{Computation of derivatives}
The first derivative is straightforward to calculate.
For any $ q \geq 0 $, we have
\begin{equation}\label{pm}
\frac{\partial }{\partial m} P(1,q) =
\mathbb{E}\left[I(q) \ln I(q) \right] - \frac{\beta^2}{4} q^2, 
\end{equation}
such that  $ \frac{\partial }{\partial m} P(1,0) = 0 $.

To determine $ q \geq 0 $ such that $ \frac{\partial }{\partial m} P(1,q) > 0 $, we investigate the partial derivatives with respect to $ q $. This amounts  to calculating
\begin{equation}\label{eq:partialq}
    \frac{\partial}{\partial q } \mathbb{E}\left[I(q) \ln I(q) \right] = \mathbb{E}\left[I'(q) \left( 1 +\ln I(q) \right) \right]
\end{equation}
where
\begin{equation}
    I'(q) = \int e^{S_q(\xi)} \left[\frac{\beta w}{2\sqrt{q}} \langle 1 , \xi \rangle - \frac{ \beta^2}{2} \langle 1 , \xi \rangle^2\right] \nu_1(d\xi) . 
\end{equation}
To process the expectation in~\eqref{eq:partialq}, we use Gaussian integration by parts which 
for any boundedly differentiable function $ F$ of the Gaussian variable $ w $ yields:
\begin{equation}\label{prop:GIP}
\mathbb{E}\left[F(w)\ e^{S_q(\xi)} \left[\frac{\beta w}{2\sqrt{q}} \langle 1 , \xi \rangle - \frac{ \beta^2}{2} \langle 1 , \xi \rangle^2\right] \right] = \frac{\beta}{2} \ \mathbb{E}\left[\frac{1}{\sqrt{q}} \ F'(w) \ e^{S_q(\xi)} \langle 1 , \xi \rangle  \right] . 
\end{equation}

We will apply this first to the case $ F(w) = 1+ \ln I(q) $ for which
$$
\frac{1}{\sqrt{q}} F'(w) = \beta  \frac{J(q)}{I(q) } , \qquad \mbox{with}\quad J(q) \coloneqq\int e^{S_q(\xi)} \langle 1 , \xi \rangle \nu_1(d\xi) . 
$$ 
Applying~\eqref{prop:GIP} to \eqref{eq:partialq} thus yields 
\begin{equation}\label{pmq}
\frac{\partial}{\partial q } \mathbb{E}\left[I(q) \ln I(q) \right] = \frac{\beta^2}{2} \mathbb{E}\left[\frac{J(q)^2}{I(q) }  \right] ,
\end{equation}
which implies
\begin{equation}
    \frac{\partial^2 P }{\partial q \partial m}(1,q) = \frac{\beta^2}{2} \ \mathbb{E}\left[\frac{J(q)^2}{I(q) } \right] - \frac{\beta^2}{2} q . 
\end{equation}
Since $ S_0(\xi) = 0 $ and $ \int \xi(t) \nu_1(d\xi) = 0 $, the above derivative at $ m = 1 $ and $ q = 0 $ is zero. To determine the behavior of $ \partial P(1,q) /\partial m $ near $ q = 0 $, we proceed by calculating the second derivative with respect to $ q $ again using~\eqref{prop:GIP}:
\begin{align}
\frac{\partial}{\partial q } \mathbb{E}\left[\frac{J(q)^2}{I(q) } \right]
 = & \ 2 \ \mathbb{E}\left[\frac{J(q)}{I(q) } \int e^{S_q(\xi) } \langle 1 , \xi \rangle   \left(  \frac{\beta w}{2\sqrt{q}}  \langle 1 , \xi \rangle - \frac{\beta^2}{2} \langle 1 , \xi \rangle^2   \right)\nu_1(d\xi)\right] - \mathbb{E}\left[\frac{I'(q)}{I(q) } J(q)^2 \right] \notag \\
 = & \beta \ \mathbb{E}\left[ \int \int e^{S_q(\xi)} \langle 1 , \xi\rangle^2 \langle 1 , \eta\rangle \frac{1 }{\sqrt{q}} \frac{\partial}{\partial w } \left( \frac{1}{I(q)} e^{S_q(\eta)} \right)   \ \nu_1(d\xi) \nu_1(d\eta) \right] \notag \\
 & - \frac{\beta}{2} \ \mathbb{E}\left[ J(q) \frac{1 }{\sqrt{q}} \frac{\partial}{\partial w} \left(\frac{J(q)^2}{I(q)^2} \right)  \right] .
\end{align}
Abbreviating $ K(q) \coloneqq\int e^{S_q(\xi)} \langle 1 , \xi \rangle^2 \nu_1(d\xi) $, we thus conclude
\begin{align}\label{pmqq}
\frac{\partial}{\partial q } \mathbb{E}\left[\frac{J(q)^2}{I(q) } \right]
 = \beta^2 \mathbb{E}\left[ \frac{K(q)}{I(q)^2}\left( K(q) I(q) -  J(q)^2\right)  \right]
 - \beta^2  \mathbb{E}\left[ \frac{K(q) J(q)^2}{I(q)^2} - 
 \frac{J(q)^4}{I(q)^3} 
 \right] . 
 \end{align}
 At $ q = 0 $, we have
 \begin{align}\label{IJK}
 I(0) = 1 , \quad  J(0) = 0 \quad &\mbox{and}  \notag \\
K(0) = \int \langle 1 , \xi \rangle^2 \nu_1(d\xi) & = \langle 1 , \mu 1 \rangle = \frac{\tanh \beta b}{\beta b}  .  \notag 
 \end{align}
The above identity is by explicit calculation, which can be found in~\cite[Eq.~(2.14)]{Leschke:2021aa}.

\subsection{Proof of Theorem~\ref{thm:main} -- Part~\ref{1}}
We are now ready to summarize the above calculations in the 
\begin{proof}[Proof of Theorem~\ref{thm:main} Part~\ref{1}] 
The function $ P: (0,\infty)\times [0,\infty) \to \mathbb{R} $ given by \eqref{def:PMq} is arbitrarily continuously differentiable. 
From~\eqref{pm}, \eqref{pmq}, and \eqref{pmqq} together with \eqref{IJK} we infer
\begin{align}
	& \frac{\partial }{\partial m} P(1,0) =0 , \quad  \frac{\partial^2 }{\partial q \partial m} P(1,0) =0 \notag \\ & \frac{\partial^3 }{\partial q^2 \partial m} P(1,0) =  \frac{\beta^2}{2}  \left[ \beta^2  \left( \frac{\tanh \beta b}{\beta b}  \right)^2 -1 \right] .  
\end{align}
Consequently, if $b < \tanh(\beta b) $, then  $  \frac{\partial }{\partial m} P(1,q) > 0 $  for some $ q > 0 $, and hence  the infimum of $ P(m,q) $ for $ m \in (0,1) $ and $ q \geq 0 $ is strictly smaller than $ P(1,0) = 0 $. The annealed solution is unstable. 

To complete the proof of Part~\ref{1} of Theorem~\ref{thm:main}, it remains to show that the replica overlap is not self-averaging \eqref{eq:notselfav}. This is the content of Corollary~\ref{thm:replicaglass} below. 
\end{proof}

\subsection{Replica overlap in the spin-glass regime}

The assertion about the replica overlap relied on the following corollary. 
\begin{cor}\label{thm:replicaglass}
    Let $b < \tanh(\beta b)$. Then, there exists a Parisi measure $\pi_{*} \in \Pi(\varrho) $, i.e.\ a minimizer in \eqref{QPsoc}, with 
\begin{equation}\label{eq:roglass}
        \lim_{N\to \infty} \mathbb{E}\left[ \langle\langle \| R_N \|_2^2 \rangle\rangle_N \right] = \int_{0}^{1} \|\pi_{*}(m) \|^2_{2} \, dm > 0 .
    \end{equation}
\end{cor}
\begin{rem}
On purpose, we formulate Theorem~\ref{thm:replicaglass} without referring to the uniqueness of the Parisi measure, to emphasize that our argument does not depend on it and, in particular, can also be applied to related models where this uniqueness has not yet been proven.     
\end{rem}
\begin{proof}[Proof of Corollary~\ref{thm:replicaglass}] Let $ \varrho = \nabla \mathcal{F}(0) $ be the asymptotic self-overlap, and pick $\varepsilon_N \downarrow 0$ according to Proposition~\ref{prop:constrained} slow enough such that $$\widehat F(\lambda, 1;\varrho) =\lim_{N \to \infty} \frac1N \mathbb{E} \ln W_{N,\varepsilon_N}(\lambda,1;\varrho)$$ 
for $\lambda$ in some open neighborhood of $1$. By Theorem~\ref{prop:diffident} and the convexity and differentiability of $ \widehat F $ and its prelimit with respect to its first coordinate, it follows that 
     \begin{align*}
         \frac{\beta^2}{2} \left( \|\varrho \|_2^2 - \int_{0}^{1} \|\pi_{*}(s) \|^2_{2} \, ds \right) &= \frac{\partial \widehat F}{\partial \lambda_1}(1,1;\varrho) = \lim_{N \to \infty} \frac1N \frac{\partial}{\partial \lambda_1}  \mathbb{E} \ln W_{N,\varepsilon_N}(1,1;\varrho)\ \\
         &= \frac{\beta^2}{2} \lim_{N\to \infty} \left(  \mathbb{E}\left[ \langle\langle \| Q_N \|_2^2 \rangle\rangle_{N,\varepsilon_N;\varrho}^{(1,1)} \right] -   \mathbb{E}\left[ \langle\langle \| R_N \|_2^2 \rangle\rangle_{N,\varepsilon_N; \varrho}^{ (1,1)} \right]  \right) \\
         &= \frac{\beta^2}{2} \left( \|\varrho \|_2^2  -  \lim_{N \to \infty} \mathbb{E}\left[ \langle\langle \| R_N \|_2^2 \rangle\rangle_{N,\varepsilon_N; \varrho}^{(1, 1)} \right]  \right).
     \end{align*}
     The second line is by Gaussian integration by parts~\eqref{eq:GIP}, and the last line is immediate from constraining around $\varrho $. Hence, $\lim_{N \to \infty} \mathbb{E}\left[ \langle\langle \| R_N \|_2^2 \rangle\rangle_{N,\varepsilon_N;\varrho}^{ (1,1)} \right] = \int_{0}^{1} \|\pi_{*}(m) \|^2_{2} \, dm.$ The concentration~\eqref{eq:expconc} of the self-overlap in the self-overlap-corrected model implies for a possibly more slowly converging sequence $\varepsilon_N \downarrow 0$
     $$ \limsup_{N \to \infty} \left| \mathbb{E}\left[ \langle\langle \| R_N \|_2^2 \rangle\rangle_{N,\varepsilon_N;\varrho}^{(1,1)} \right] - \mathbb{E}\left[ \langle\langle \| R_N \|_2^2 \rangle\rangle_N \right] \right| = 0  , $$
     which completes the proof of the equality. 
     
     The fact that any minimizer $ \pi_* \in \Pi(\varrho) $ cannot have $ \int_0^1 \| \pi_*(m) \|_2^2 dm = 0 $ follows from the quantum Toninelli argument by contradiction. Suppose it is zero. Then monotonicity would imply $\pi_*(m)  = 0 $ for all $ m \in [0,1) $, and hence this path is among those considered in~\eqref{def:PMq}, for which we showed that the minimizer does not have $ m = 1 $. 
 \end{proof}

\appendix
\section{Compactness for non-negative Hilbert-Schmidt operators}

In the following, let $\mathcal{H}$ be a separable Hilbert space, and $ \mathcal{S}_2(\mathcal{H}) $ denote the Hilbert space of Hilbert\--Schmidt operators on $ \mathcal{H} $. 
The latter is equipped with a scalar product $ \langle A , B \rangle \coloneqq \tr A^* B $ and $ \| A \|_2 \coloneqq \sqrt{\langle A , A\rangle } $ denotes its Hilbert-Schmidt norm. The operator norm is $ \| A \| $, and $ \mathcal{S}^+_2(\mathcal{H})$ denotes the non-negative Hilbert-Schmidt operators. The following observation is our main tool to boost weak convergence results to strong convergence. 
\begin{prop}\label{lem:uniftail}
Let $ A_n \in \mathcal{S}_2^+(\mathcal{H}) $ be a sequence of Hilbert-Schmidt operators. Suppose:
\begin{enumerate}
\item there is an increasing family of finite-dimensional subspaces $  \mathcal{H}_D \subset \mathcal{H} $ whose orthogonal projections $ P_D $ converge weakly in $ \mathcal{H} $ to the identity, $ \lim_{D\to \infty} \langle \varphi , Q_D \psi \rangle = 0 $ for all $ \varphi, \psi \in \mathcal{H}$, where $ Q_D \coloneqq 1 - P_D $, and
\item 
for any $ \varepsilon > 0 $ there are $ D_\varepsilon , N_\varepsilon > 0 $ such that
\begin{equation}\label{eq:uniftail}
\sup_{n\geq N_\varepsilon} \left\| Q_{D_\varepsilon} A_n Q_{D_\varepsilon} \right\|_2 \leq \varepsilon .
\end{equation}
\end{enumerate}
Then:
\begin{enumerate}
    \item weak convergence of $ (A_n) $ in $ \mathcal{S}_2(\mathcal{H}) $ implies its strong convergence in $ \mathcal{S}_2(\mathcal{H}) $. 
    \item if $\sup_{n} \| A_n \|_2 < \infty $, then $ (A_n ) $ is relatively compact in $\mathcal{S}_2(\mathcal{H})  $.
\end{enumerate} 
\end{prop}
\begin{proof}
 1.  Let $ A \in \mathcal{S}_2(\mathcal{H}) $ stand for a weak limit of  $ (A_{n})$. For $ \varepsilon > 0 $ arbitrary, we first pick $ \widehat D_\varepsilon , N_\varepsilon\in \mathbb{N} $ such that~\eqref{eq:uniftail} applies and hence, since the subspaces are nested,
 $$
 \sup_{n\geq N_\varepsilon} \left\| Q_{D} A_n Q_{D} \right\|_2 \leq \varepsilon .
 $$
 holds for all $ D \geq \widehat D_\varepsilon$. We then choose $ D_\varepsilon \geq \widehat D_\varepsilon $ such that additionally $\| A - P_{ D_\varepsilon} A P_{D_\varepsilon} \|_2 \leq \varepsilon $. This is possible since $ A $ is a compact operator. The triangle inequality and \eqref{eq:HSblock} below then yield for all $ n \geq N_\varepsilon $:
\begin{align*}
 \left\| A - A_n \right\|_2 & \leq \left\|P_{D_\varepsilon} ( A - A_n ) P_{D_\varepsilon}\right\|_2 + \left\|A - P_{D_\varepsilon} A P_{D_\varepsilon}\right\|_2 + \sqrt{2 \left\|P_{D_\varepsilon} A_n P_{D_\varepsilon}\right\|_2 \left\|Q_{D_\varepsilon} A_n Q_{D_\varepsilon}\right\|_2} + \left\|Q_{D_\varepsilon} A_n Q_{D_\varepsilon}\right\|_2 \\
 &\leq  \left\|P_{D_\varepsilon} ( A - A_n ) P_{D_\varepsilon}\right\|_2 + 2 \varepsilon + \sqrt{2 \varepsilon \left\|P_{D_\varepsilon} A_n P_{D_\varepsilon}\right\|_2 } . 
\end{align*}
Since $ P_{D_\varepsilon} $ is finite dimensional, the first term on the right converges to zero as $ n \to \infty $. By the same reasoning, $ \lim_{n\to \infty} \left\|P_{D_\varepsilon} A_n P_{D_\varepsilon}\right\|_2  = \left\|P_{D_\varepsilon} A P_{D_\varepsilon}\right\|_2 \leq \| A \|_2 $. This completes the proof of the first assertion, since $\varepsilon > 0 $ was arbitrary. \\

\noindent
2.~The Banach-Alaoglu theorem implies that $ (A_n) $ is weakly compact in the Hilbert space $ \mathcal{S}_2(\mathcal{H}) $. By the first assertion, any weakly convergent subsequence converges strongly. 
\end{proof}

We provide an elementary block operator analysis result that was used in the previous proof.
\begin{lemma}\label{prop:HSblock}
    Let  $P$ be an orthogonal projection and $Q = I - P$ its complement in $ \mathcal{H}$.
    \begin{enumerate}
        \item For any non-negative bounded operator $0 \leq M  \in \mathcal{B}(\mathcal{H})$, there exists a bounded linear operator $K : \mathcal{H} \to \mathcal{H}$ with operator norm $\|K \| \leq 1$ and
        \begin{equation}\label{eq:offdiag}
            QMP = (QMQ)^{\frac12} K (PMP)^{\frac12}.
        \end{equation}
        \item For any non-negative Hilbert-Schmidt operator $0\leq B \in \mathcal{S}_2^+$, 
        \begin{equation}\label{eq:HSblock}
            \| B - PBP \|_2^2 \leq 2 \|PBP\|_2 \|QBQ\|_2 + \|QBQ\|_2^2 . 
        \end{equation}
    \end{enumerate}
\end{lemma}
 \begin{proof}
1.~ To simplify the notation, let us write the operator $M$ in block-form
$$ M = \begin{pmatrix}
        C & X^{*} \\
            X & D
        \end{pmatrix}$$
with $C :=  PMP \geq 0, \, D:=  QMQ \geq 0 $ and $X := QMP.$

\noindent Let us first assume that $C$ and $D$ are both invertible.
For arbitrary $u,v \in \mathcal{H}$, positivity gives
\[
0 \leq \left\langle \begin{pmatrix} u \\ v \end{pmatrix}, \begin{pmatrix} C & X^* \\ X & D \end{pmatrix} \begin{pmatrix} u \\ v \end{pmatrix} \right\rangle
= \langle u,Cu\rangle + 2\operatorname{Re}\langle v,Xu\rangle + \langle v,Dv\rangle . \]
Set  $a :=C^{1/2}u, \, b:=D^{1/2}v$ and $K:=D^{-1/2}XC^{-1/2}$.
Then the above inequality becomes
$$ \|a\|^2 + 2\operatorname{Re}\langle b,Ka\rangle + \|b\|^2 \geq 0. $$
Choosing $b=-Ka$ yields $0\le \|a\|^2-\|Ka\|^2$ for all $a\in H$. Hence, $\|K \| \leq 1$ and by definition
 $$ X=D^{1/2}KC^{1/2}. $$
 
\noindent In the general case $C,D\ge0$, we refer to a standard regularization trick.  For $\delta>0$, set
$$ C_\delta:=C+\delta I, \qquad D_\delta:=D+\delta I. $$
$C_\delta$ and $D_\delta$ are clearly invertible and, hence, 
there exists a contraction $K_\delta$ such that
$ X=D_\delta^{1/2}K_\delta C_\delta^{1/2}$  
and $\|K_\delta\|\leq 1 $.
By the Banach-Alaoglu theorem, the unit ball of $\mathcal{B}(\mathcal{H})$ is compact in the weak operator topology. Thus, there exists
$\delta_n\downarrow0$ and a contraction $K$ such that in the weak operator topology, $K_{\delta_n}\to K$. 
Moreover, the spectral theorem implies
$$ C_{\delta_n}^{1/2}u\to C^{1/2}u, \qquad D_{\delta_n}^{1/2}v\to D^{1/2}v $$
strongly for all $u,v\in \mathcal{H}$. Therefore, for arbitrary $u,v\in \mathcal{H}$,
$$ \langle Xu,v\rangle = \langle K_{\delta_n}C_{\delta_n}^{1/2}u, D_{\delta_n}^{1/2}v\rangle\longrightarrow \langle KC^{1/2}u,D^{1/2}v\rangle.
$$
Thus, $ X=D^{1/2}KC^{1/2}$.

\noindent 2.~By definition of the Hilbert-Schmidt norm
$$  \| B - PBP \|_2^2 = 2 \| QBP \|_2^2 + \| QBQ \|_2^2. $$
To bound the off-diagonal term, we employ the first assertion, which yields $QBP = (QBQ)^{\frac12} K (PBP)^{\frac12} $ with a contraction $K$. Note that
$ (QBP)^*(QBP) = (PBP)^{\frac12} K^{*}(QBQ) K (PBP)^{\frac12}  $, 
which, by the cyclicity of the trace, implies
\begin{align} \| QBP \|_2^2 & = \tr  (PBP)^{\frac12} K^* (QBQ)  K (PBP)^{\frac12} =  \tr K (PBP)K^* (QBQ) \notag \\
& \leq  \| K PBP K^* \|_2  \| QBQ \|_2 \leq \| K \|^2 \| PBP \|_2 \| QBQ \|_2 \, 
\end{align}
where the first two estimates are based on the Cauchy-Schwarz inequality for the Hilbert-Schmidt scalar product and the operator norm bound. Since  $ \| K \| \leq 1 $ this yields the claim. 
    \end{proof}
 
\section{Concentration of the self-overlap}

This appendix collects three results on the concentration of the self-overlap when viewed as a Hilbert-Schmidt operator. We start with concentration results for the non-interacting path measure, yielding an explicit concentration bound. Combined with prior results from~\cite{MW25}, this allows us to show weak concentration of the self-overlap. Some arguments in the main text rely on a stronger exponential concentration, which concerns us in the final part.
 
 \subsection{Non-interacting case}\label{app:confree} 
 
 Under the a priori probability measure $ \nu_N $, the self-overlap is a weighted sum of independent and identically distributed rank-one operators,
 $$
  Q_N[\bxi]  = \frac{1}{N} \sum_{j=1}^N | \xi_j \rangle\langle \xi_j | .
 $$
Each rank-one term $ | \xi_j \rangle\langle \xi_j | $, which is defined via its kernel $(t,s) \mapsto   \xi_j(t) \xi_j(s) $, is an element of the real Hilbert space of Hilbert-Schmidt operators on $ L^2(0,1) $. 
The Hilbert-Schmidt norm is bounded
\begin{equation}\label{eq:bound}
 \left\|  | \xi_j \rangle\langle \xi_j | \right\|_2 = 1  
\end{equation}
for all $ j \in \{ 1, \dots , N \} $. 
Under $ \nu_N $, the mean of $ | \xi_j \rangle\langle \xi_j | $  is given by $ \mu $, see~\eqref{eq:self-overlap}. The expected Hilbert-Schmidt norm of the fluctuations is
 \begin{align}\label{eq:sigma}
 \sigma^2 \coloneqq & \int \left\|  | \xi_j \rangle\langle \xi_j | - \mu \right\|_2^2 \nu_N(d\bxi) =   \int_0^1 \int_0^1 \left( 1-  \mu(t,s)^2 \right) ds dt  \notag \\
= & \ \frac{1}{2} \left[ 1+ \left(\tanh(\beta b) \right)^2  - \frac{\tanh(\beta b)}{\beta b}   \right] . 
 \end{align}
 Here, the last step is by explicit calculation, which is found in~\cite{Leschke:2021aa}. 
 An application of Pinelis' Bernstein-type concentration inequality for sums of independent vectors in a real Hilbert space, thus yields the following standard concentration inequality. 
 \begin{prop}
 For any $ N \in \mathbbm{N} $ and   $ \varepsilon > 0 $:
 \begin{equation}\label{eq:BernsteinHS}
 \nu_N\left( \left\| Q_N - \mu \right\|_2 > \varepsilon \right) \leq 2 \exp\left(- \frac{3 N}{2} \frac{\varepsilon^2}{3 \sigma^2 + 2 \varepsilon}  \right) 
 \end{equation}
 with $ \sigma^2 $ from~\eqref{eq:sigma}. 
 \end{prop}
\begin{proof}
The proof is an application of Pinelis's concentration of random vectors in a real Hilbert space \cite{Pin94}. In our case, the Hilbert space is the space of real Hilbert-Schmidt operators on $ L^2(0,1) $. 
The operators $ X_j = | \xi_j \rangle\langle \xi_j | - \mu $, which are defined via their kernels $ X_j(t,s) = \xi_j(t) \xi_j(s) - \mu(t,s) $, are independent and identically distributed vectors in the real Hilbert space $ \mathcal{S}_2 $. 
By \cite[Thm. 3.2]{Pin94} their sums $ S_N \coloneqq \sum_{j=1}^N X_j = N ( Q_N - \mu ) $ satisfy for all $ \beta > 0 $:
\begin{equation}
\int \cosh\left(\beta \left\| S_N \right\|_2 \right) \nu_N(d\bxi) \leq \prod_{j=1}^N \left(1 +  \int \left( e^{\beta \| X_j \|_2} - 1 - \beta \| X_j \|_2 \right) \nu_1(d\xi_j) \right) . 
\end{equation}  
By~\eqref{eq:bound} we have $ \| X_j \|_2 \leq 2 $, which for  $0 <  2\beta < 3 $ implies
\begin{align*}
\int \left( e^{\beta \| X_j \|_2} - 1 - \beta \| X_j \|_2 \right) \nu_1(d\xi_j)  & \leq \frac{\beta^2 \sigma^2}{2} \sum_{n=0}^\infty \frac{2 \ (2\beta)^n }{(n+2)!}  \leq  \frac{\beta^2 \sigma^2}{2} \sum_{n=0}^\infty \left(\frac{2\beta}{3}\right)^n \\
& \leq  \frac{\beta^2 \sigma^2}{2}\frac{3}{ 3 - 2\beta } . 
\end{align*}
The claim~\eqref{eq:BernsteinHS} then follows from a Chernoff estimate
\begin{align*}
 \nu_N\left( \left\| S_N \right\|_2 >  N \varepsilon \right) \leq   2  e^{- N \beta \varepsilon} \int \cosh\left(\beta \left\| S_N \right\|_2 \right) \nu_N(d\bxi)  
 \leq  2  \exp\left( - N \left[ \beta  \varepsilon - \frac{\beta^2 \sigma^2}{2} \frac{3}{ 3 - 2\beta }\right] \right) 
\end{align*}
in which we pick $ \beta = \frac{3\varepsilon}{3 \sigma^2 + 2 \varepsilon} $ . 
\end{proof}

\subsection{Finite-dimensional concentration}
For a proof of some of our approximation results as well as for the concentration in the interacting case, we will also need the following concentration results for the path set $ \mathcal{S}_{N,\varepsilon}^D $ from~\eqref{eq:restSO} under the Gibbs measure  $ \langle (\cdot) \rangle_N $  from~\eqref{eq:probmeas} or $\langle (\cdot) \rangle_N^{D'} $ from~\eqref{eq:probmeasD}.
\begin{prop}\label{prop:approx}
Let $ \beta> 0 $. There is some $ c_{\beta} \in (0,\infty ) $ such that 
for any $ \varepsilon > 0 $ there is $ D_\varepsilon \in \mathbb{N} $ such that for all $ D' \geq D \geq D_\varepsilon $ and $ N \in \mathbb{N} $:
\begin{equation}\label{eq:concSO}
\mathbb{E}\left[ \langle 1_{ \mathcal{S}_{N,\varepsilon}^D} \rangle_N^{D'} \right] \geq 1- \frac{e^{-N c_{\beta}}}{1- e^{-N c_{\beta}}} .
\end{equation}
The estimate remains true in case $ D' = \infty $ for which $\langle \cdot \rangle_N^{D'} $ is replaced by $\langle \cdot \rangle_N$. 
\end{prop}
\begin{proof}
The quantity $\widehat Z_{N,\varepsilon}^{D,D'} \coloneqq \widehat Z_N^{D'} \ \langle 1_{ \mathcal{S}_{N,\varepsilon}^D} \rangle_N^{D'}   $ represents the partition function, in which we restrict to the path set $ \mathcal{S}_{N,\varepsilon}^D $.
Since
$$
0 \leq 1 - \frac{\widehat Z_{N,\varepsilon}^{D,D'}}{\widehat Z_N^{D'}} \leq \ln \widehat Z_N^{D'} - \ln \widehat Z_{N,\varepsilon}^{D,D'} ,
$$ 
the claim thus follows from a straightforward extension of \cite[Cor. 3.4.]{MW25}, which deals with the case $ D' = \infty $.
\end{proof}

\subsection{Concentration in the interacting case}\label{app:conc2}

Equipped with Proposition~\ref{prop:approx} we are ready to spell out the proof of the concentration of the self-overlap in the interacting case. 

 \begin{proof}[Proof of Proposition~\ref{prop:conc}]
1.~The convexity and differentiability of $ \mathcal{F} $ were established as \cite[Cor. 2.4]{MW25}, which also contains a proof of the weak convergence 
 \begin{equation}\label{eq:weakHS}
     \lim_{N\to \infty} \langle m_N - \varrho , y \rangle = 0 
 \end{equation}
 for all $ y \in \mathcal{S}_2 $. 
 To boost weak to strong convergence, we use Proposition~\ref{lem:uniftail} and show that the sequence of Hilbert-Schmidt operators has a uniformly bounded tail. To do so, we 
  restrict the path measure $ \langle (\cdot) \rangle_N $ to paths whose self-overlap is restricted to
the set  $
 \mathcal{S}_{N,\varepsilon}^D $ defined in~\eqref{eq:restSO}.
More precisely, for the given $ \varepsilon > 0 $ we pick $ D_\varepsilon > 0 $ according to Proposition~\ref{prop:approx}, and denote by 
$$
 \langle (\cdot) \rangle_{N,\varepsilon} \coloneqq\langle (\cdot) \ 1[\mathcal{S}_{N,\varepsilon}^{D_\varepsilon} ]  \rangle_{N} 
 $$
the (unnormalized) measure restricted to the event $  \mathcal{S}_{N,\varepsilon}^{D_\varepsilon} $. We then estimate
using  $ \| Q_N \|_2 \leq \tr Q_N \leq 1 $ to conclude
\begin{align*}
    \left\| (1-P_{D_\varepsilon}) \mathbb{E}\left[\langle Q_N \rangle_{N}\right]  (1-P_{D_\varepsilon})  \right\|_2 
    & \leq  \left\|  \mathbb{E}\left[ \langle (1-P_{D_\varepsilon}) Q_N (1-P_{D_\varepsilon})  \rangle_{N,\varepsilon}  \right]  \right\|_2  + \mathbb{E}\left[ 1 -  \langle 1[\mathcal{S}_{N,\varepsilon}^{D_\varepsilon}]  \rangle_{N} \right] \notag \\
    & \leq \varepsilon + \frac{e^{-N c_{\beta}}}{1- e^{-N c_{\beta}}} .
\end{align*}
The last estimate used Proposition~\ref{prop:approx} and the bound 
$$\| (1-P_{D_\varepsilon}) Q_N (1-P_{D_\varepsilon})  \|_2 \leq \tr (1-P_{D_\varepsilon}) Q_N (1-P_{D_\varepsilon}) = \frac{1}{N} \sum_{j=1}^N  \langle \xi_j , (1-P_{D_\varepsilon}) \xi_j \rangle \leq \varepsilon 
$$
by 
the choice of $ D_{\varepsilon}$. For $ N $ large enough, the second term is arbitrarily small, which allows the application of Proposition~\ref{lem:uniftail}. This concludes the proof of~\eqref{eq:strongHS}. 

\noindent
2.~The proof of~\eqref{eq:expconc} again relies on the uniform approximability, and the differentiability of $ \mathcal{F}(x) = \inf_{\pi \in \Pi} \mathcal{P}_1(\pi,x)$. Let $ \varepsilon > 0 $ be arbitrary and 
    fix  $D \geq D_{\varepsilon^2/64} $ such that 
    $$
    \left\| (1- P_D) Q_N \right\|_2 \leq \| Q_N \|  \sqrt{\tr (1-P_{D}) Q_N (1-P_{D})} \leq \varepsilon /8 ,
    $$
    as well as  $\| \varrho - P_{D} \varrho P_{D} \|_2 < \varepsilon/4 $.  Using the triangle inequality and Proposition~\ref{prop:approx}, we conclude that for some $ c , C \in (0,\infty) $ and all $ N $ large enough:
    \begin{align}
    \mathbb{E}\left[\big\langle 1[ \|Q_N - \varrho   \|_2 > \varepsilon ]\big\rangle_{N} \right] &\leq \mathbb{E}\left[\big\langle1[\|P_D Q_N P_D - P_D\varrho P_D  \|_2 > \varepsilon/2 ] \big\rangle_{N} \right] + Ce^{-cN}. 
    \end{align}
    Expanding the Hilbert-Schmidt norm in an orthonormal basis $(e_\alpha ) $ of $ P_{D}L^2(0,1) $, and 
    employing the union bound, we obtain
 \begin{equation}\label{eq:finitebound}
     \mathbb{E}\left[\big\langle 1[ \|P_DQ_NP_D - P_D\varrho P_D  \|_2 > \varepsilon/2 ] \big\rangle_{N} \right]  \leq \max_{1 \leq \alpha ,\beta \leq 2^D} \mathbb{E}\left[\big\langle 1[| \langle e_\alpha , (Q_N -\varrho) e_\beta | > 4^{-(D +1)}\varepsilon ]\big\rangle_{N} \right]. 
    \end{equation}
    We now fix $\alpha,\beta \in \{ 1, \dots , 2^D \} $. Abbreviating 
    the rank-one operator $ x_{\alpha,\beta} \coloneqq |e_\beta \rangle \langle e_\alpha | $ and $\delta \coloneqq 4^{-(D +1)}\varepsilon$, the Markov inequality implies for any $ h > 0 $:
    \begin{equation}\label{eq:Markov1}
    \big\langle 1[ \langle e_\alpha , (Q_N -\varrho) e_\beta \rangle  > \delta ]\big\rangle_{N}  \leq \frac{{Z_N}(h x_{\alpha,\beta})}{Z_N(0)} \exp(-h N(\delta  + \langle e_\alpha,\varrho e_\beta \rangle)),
    \end{equation}
    where we introduced the abbreviation 
    $$ Z_N(x) :=  \int \exp\left(N \left\langle x , Q_N[\bxi] \right\rangle -\int_0^1 \beta U\left(\bxi(t)\right) dt  - \frac{\beta^2}{2} \int_0^1 \int_0^1  \mathbb{E}\left[U\left(\bxi(s)\right)  U\left(\bxi(t)\right)  \right] ds dt \right) \nu_N(d\mathbf{\bxi}) $$
    for the enriched partition function defined for all $x \in \mathcal{S}_2 $. We recall from~\cite[Thm.~2.3]{MW25} that
    $$ \lim_{N \to \infty} \frac{1}{N} \ \mathbb{E}[\ln Z_N(x) ] = \inf_{\pi \in \Pi}\mathcal{P}_1(\pi,x) = \mathcal{F}(x) $$
    for all $ x \in \mathcal{S}_2 $.
    Since $ \mathcal{F} $ is differentiable, we find a $h_{\alpha,\beta} > 0$ sufficiently small such that 
    $$\mathcal{F}(h_{\alpha,\beta}  x_{\alpha,\beta}) - \left(\mathcal{F}(0) + h_{\alpha,\beta} \langle e_\alpha,\varrho e_\beta \rangle) \right) \leq \frac{\delta}{4} h_{\alpha,\beta}.
    $$
    We now make use of the standard Gaussian concentration estimate, which yields some $ c_{\alpha,\beta} > 0 $ such that for all $ N $ sufficiently large
    $$  \mathbb{P}\left(| \ln Z_N(h_{\alpha,\beta} x_{\alpha,\beta}  ) - N \mathcal{F}(h_{\alpha,\beta} x_{\alpha,\beta}  ) | > N h_{\alpha,\beta} \delta/8\right) \leq e^{-c_{\alpha,\beta} N} , $$
    and a similar concentration bound applies to $Z_N(0)$. Hence, except for an event of exponentially small probability (in $\mathbb{P} $), one has the bound 
    $$  \big\langle 1[\langle e_\alpha, (Q_N -\varrho) e_\beta \rangle  > \delta ] \big\rangle_{N}  \leq e^{-N h_{\alpha,\beta} \delta/ 2 }. $$
    Collecting all Gaussian concentration bounds, one arrives at an exponential tail estimate of the form
    $$ \max_{1\leq \alpha,\beta\leq 2^D}\mathbb{E} \left[ \big\langle 1[ \langle e_\alpha , (Q_N -\varrho) e_\beta \rangle  > \delta ]\big\rangle_{N}  \right] \leq e^{-c N} $$
    for $N$ large enough and some different  $c>0$, which is finite, since the maximum is over a finite set.  Applying the same strategy for the lower tail bound $ \langle e_\alpha , (Q_N -\varrho) e_\beta \rangle  < - \delta  $, yields the claimed exponential estimate~\eqref{eq:expconc}. 
\end{proof}

\section{Approximation results}

\subsection{Proof of Proposition~\ref{prop:Parisiend}} \label{app:ProofParisi}

\begin{proof}[Proof of Proposition~\ref{prop:Parisiend}]
We first establish the itemized claims.\\
\noindent
1. Since the functionals $ \mathcal{F}^D$ and their limit $ \mathcal{F} $ are all convex and Gateaux-differentiable, the weak convergence 
$$
\lim_{D\to \infty } \langle \nabla \mathcal{F}^D(0) , y \rangle = \langle \nabla \mathcal{F}(0) , y \rangle
$$
for all $ y \in \mathcal{S}_2 $, follows from~\eqref{eq:functionalconf}. To conclude the strong convergence, we use Proposition~\ref{lem:uniftail} and construct a uniform tail estimate on $ \varrho^D $ relating it to the empirical self-overlaps $Q_N^D \equiv Q_N[\bxi^D] $.  For $ \varepsilon > 0 $ we pick $ D_\varepsilon \in \mathbb{N} $ according to Proposition~\ref{prop:approx}.
On $ \mathcal{S}_{N,\varepsilon}^{D_\varepsilon} $, for any $ D \geq D_\varepsilon $ by the nestedness of subspaces we have the estimate:
\begin{equation}\label{eq:eventcondS}
\frac{1}{N} \sum_{j=1}^N \langle \xi_j^D , (1-P_{D_\varepsilon} ) \xi_j^D \rangle  = \frac{1}{N} \sum_{j=1}^N \langle \xi_j , (P_D-P_{D_\varepsilon} ) \xi_j \rangle  \leq \frac{1}{N} \sum_{j=1}^N \langle \xi_j , (1-P_{D_\varepsilon} ) \xi_j \rangle \leq \varepsilon .
\end{equation}
Consequently, abbreviating $ \langle (\cdot) \rangle_{N,\varepsilon}^{D} \coloneqq\langle (\cdot) \ 1[\mathcal{S}_{N,\varepsilon}^{D_\varepsilon} ]  \rangle_{N} 
 $
the (unnormalized) measure restricted to the event $  \mathcal{S}_{N,\varepsilon}^{D_\varepsilon} $, and
using  $ \| Q_N^D \|_2 \leq \tr Q_N^D \leq 1 $ and $ \| (1-P_{D_\varepsilon}) Q_N^D (1-P_{D_\varepsilon})  \|_2 \leq \tr (1-P_{D_\varepsilon}) Q_N^D (1-P_{D_\varepsilon}) $ we estimate
\begin{align*}
    \left\| (1-P_{D_\varepsilon}) \mathbb{E}\left[\langle Q_N^D \rangle_{N}^D\right]  (1-P_{D_\varepsilon})  \right\|_2 
    & \leq  \left\|  \mathbb{E}\left[ \langle (1-P_{D_\varepsilon}) Q_N^D (1-P_{D_\varepsilon})  \rangle_{N,\varepsilon}^D  \right]  \right\|_2  + \mathbb{E}\left[ 1 -  \langle 1[\mathcal{S}_{N,\varepsilon}^{D_\varepsilon}]  \rangle_{N}^D \right] \notag \\
    & \leq \varepsilon + \frac{e^{-N c_{\beta}}}{1- e^{-N c_{\beta}}} .
\end{align*}
The last line is from~\eqref{eq:eventcondS} and Proposition~\ref{prop:approx}. 
Since the second term is arbitrarily small in the limit $ N \to \infty $, and $ \lim_{N\to \infty } \| \mathbb{E}\left[\langle Q_N^D \rangle_{N}^D\right] - \varrho^D \|_2 = 0 $ by~\eqref{eq:strondD}, also 
\begin{equation}
   \sup_{D \geq D_{\varepsilon}} \left\| (1- P_{D_\varepsilon} )\varrho^D  (1- P_{D_\varepsilon} ) \right\|_2 \leq \varepsilon , 
\end{equation}
so that the claimed strong convergence is implied using Proposition~\ref{lem:uniftail}.\\
\noindent
2.~Any $ \pi^D \in \Pi^D $ is non-negative and 
$$
\sup_{m\in[0,1]} \| \pi^D(m) \|_2 \leq \| \varrho^D \|_2 .
$$
By the strong convergence established in the first item, the right side is uniformly bounded in $ D $. An application of Proposition~\ref{lem:uniftail} hence yields the relative compactness of the sequence $ \pi^D $ in the topology of $ L^2([0,1];\mathcal{S}_2^+) $, which concludes the proof of the second item.\\

Let $ \pi_*^D \in \Pi^D(\varrho^D) $ be a minimizer of $ \inf_{\pi \in \Pi^D(\varrho^D) } \mathcal{P}_{\lambda}(\pi, 0) $. Note that the latter exists by the lower boundedness and Lipschitz continuity of the Parisi functional. By the relative compactness established in the second item, there is a convergent subsequence and some $ \pi_* \in \Pi $ such that:
$$
\lim_{k \to \infty} \int_0^1 \| \pi_*^{D_k}(m) - \pi_*(m) \|_2 dm = 0 \quad
 \mbox{and} \quad 
\lim_{k \to \infty}  \| \varrho^{D_k} - \varrho \|_2 = 0 .
$$

In particular, $ \pi_*(1) \leq \varrho $. Since the value of the Parisi functional $ \mathcal{P}_\lambda(\pi,x) $ is unchanged if one adds a step in the path at $ m=1 $, we may consider the path $ \widetilde \pi_*(m) = \pi_*(m) 1[m \in [0,1) ] + \varrho 1[m = 1 ] $ and hence assume, without loss of generality that  $ \pi_* \in \Pi_{\lambda}(\varrho) $. By the Lipschitz continuity~\eqref{eq:Lipschitz}  of the Parisi functional 
$$ \lim_{k\to \infty} \mathcal{P}_{\lambda}(\pi_*^{D_k},0) = \mathcal{P}_{\lambda}(\pi_*,0) \geq \inf_{\pi \in \Pi(\varrho) }  \mathcal{P}_{\lambda}(\pi,0) .
$$
This shows that $ \liminf_{D\to \infty} \inf_{\pi \in \Pi^D(\varrho^D) } \mathcal{P}_{\lambda}(\pi, 0)  \geq \inf_{\pi \in \Pi(\varrho) }  \mathcal{P}_{\lambda}(\pi,0) $. The complementary upper bound is straightforward, since by 
$$
\limsup_{D\to \infty} \inf_{\pi \in \Pi^D(\varrho^D) } \mathcal{P}_1(\pi, 0)  \leq \limsup_{D\to \infty} \mathcal{F}^D(0) = \mathcal{F}(0) \leq \inf_{\pi \in \Pi(\varrho) }  \mathcal{P}(\pi,0) .
$$
This concludes the proof of~\eqref{eq:Parisiend}.

It remains to establish the claim concerning limit points $ \pi^{*} = \lim_{k \to \infty} \pi_{D_k}^{*}$ of minimizers. Since the Parisi functionals $\mathcal{P}_\lambda^{D}(\pi,0)$ and $\mathcal{P}_\lambda(\pi,0)$ coincide if $\pi = P_D \pi P_D$, we have
    $$ \lim_{k \to \infty} \mathcal{P}_\lambda^{D_k}(\pi_{D_k}^{*}, 0) = \lim_{k \to \infty} \mathcal{P}_\lambda(\pi_{D_k}^{*}, 0) =  \mathcal{P}_\lambda(\pi^{*}, 0) ,$$
    where the last equality is by the Lipschitz continuity~\eqref{eq:Lipschitz} of the Parisi functional. 
    Hence, $\pi^{*}$ is a Parisi measure.
\end{proof}

\subsection{Proof of Lemma~\ref{lem:FDtoF}}\label{S:FDtoF}
\begin{proof}[Proof of Lemma~\ref{lem:FDtoF}]
     An upper bound on $F(\lambda;\upsilon)$ is obtained as follows:
    \begin{align*} F(\lambda;\upsilon) &= \frac{\lambda^2\beta^2}{4} \|\upsilon \|_2^2 + \inf_{x \in \mathcal{S}_2} \inf_{\pi \in \Pi(\upsilon)} \left[   \mathcal{P}_\lambda\left(\pi, x \right) - \langle x, \upsilon\rangle  \right] \leq \frac{\lambda^2\beta^2}{4} \|\upsilon \|_2^2 + \inf_{x \in \mathcal{S}_2^D} \inf_{\pi \in \Pi^D(P_D \upsilon P_D)} \left[   \mathcal{P}_\lambda\left(\pi, x \right) - \langle x, \upsilon\rangle  \right] \\
    &=  \frac{\lambda^2\beta^2}{4}\left(  \|\upsilon \|_2^2 - \|P_D\upsilon P_D\|_2^2 \right) + F^D(\lambda; P_D \upsilon P_D).
    \end{align*}
    The first inequality makes use of the fact that any path $\pi$ ending at $P_D \upsilon P_D$ can be replaced by the modified (and still increasing) path $\pi'(s) := \pi(s) \mathbbm{1}_{[0,1)}(s) + \upsilon \mathbbm{1}_{\{1 \}}(s) $ without changing the value of the Parisi functional.  The second line uses $\mathcal{P}_\lambda^D(\pi,x) = \mathcal{P}_\lambda(\pi,x) $ and $\langle x , \upsilon \rangle = \langle x, P_D \upsilon P_D \rangle$ if $x = P_D x P_D$ and $\pi = P_D \pi P_D.$ It then follows that $F(\lambda;\upsilon) \leq \liminf_{D \to \infty} F^D(\lambda; P_D \upsilon P_D).$ For the converse lower bound, we fix $\varepsilon > 0$ and choose $x_{\varepsilon}, \pi_{\varepsilon}$ such that
    $$  \inf_{\pi \in \Pi(\varrho)} \left[   \mathcal{P}_\lambda\left(\pi, x \right) - \langle x, \upsilon\rangle  \right] \geq  \mathcal{P}_\lambda\left(\pi_{\varepsilon}, x_{\varepsilon} \right) - \langle x_{\varepsilon}, \upsilon\rangle  -\varepsilon  . $$
    By the Lipschitz continuity~\eqref{eq:Lipschitz} of the Parisi functional, we obtain $\lim_{D \to \infty} \mathcal{P}_\lambda(P_D \pi P_D, P_D x P_D) = \mathcal{P}_\lambda(\pi,x)$ for any fixed $x,\pi$. In particular, for any large enough $D$:
    \begin{align*} \inf_{\pi \in \Pi(\upsilon)} \left[   \mathcal{P}_\lambda\left(\pi, x \right) - \langle x, \upsilon\rangle  \right] & \geq \mathcal{P}_\lambda\left(P_D\pi_{\varepsilon}P_D, P_Dx_{\varepsilon} P_D\right) - \langle P_D x_{\varepsilon} P_D, \upsilon\rangle - 2 \varepsilon \\
    &\geq \inf_{x \in \mathcal{S}_2^D} \inf_{\pi \in \Pi^D(P_D \upsilon P_D)} \left[   \mathcal{P}_\lambda\left(\pi, x \right) - \langle x, \upsilon\rangle  \right] - 2 \varepsilon.   
    \end{align*}
    Hence, $F(\lambda;\upsilon) \geq F^D(\lambda;P_D \upsilon P_D) - 2 \varepsilon $ from which the lower bound follows since $\varepsilon > 0$ can be chosen arbitrarily small. 
\end{proof}
\section{On the relation to the order parameter}\label{app:Panch}

 \begin{proof}[Proof of Lemma~\ref{lem:replica}]
     We recapitulate Panchenko's perturbation argument \cite{Pan18}, which establishes the Parisi formula for vector-spin glasses in the first place. We consider the unconstrained self-overlap-corrected model with an effective vector spin Hamiltonian and let us introduce the effective $D$-dimensional energy functional $U_N[\bxi^D] :=  \int_0^1 \lambda_1 \beta U\left(\bxi^D(t)\right) dt$, where as before  $\bxi^D$ stands for the projection to the square-wave pulses $e_k$ with $k = 1, \ldots, 2^D$.  The perturbed model takes the form 
     $$U_N^{g}[\bxi^D] := U_N(\bxi^D)  + N^{3/8} \sum_{k = 1}^{\infty} 2^{-k} g_k V_N^{k}(\bxi^D), $$
     where the $V_N^k$ stand for independent (also from $H_N$) Gaussian processes with variance uniformly bounded by $1$. The self-overlap correction remains unperturbed. The processes $V_N^{k}$ are chosen such that the vector spin glass becomes generic - for $D = 1$ this is achieved by taking all $p$-spin models, and in the general vector-valued case one has also to take the multiple directions into account (see e.g. \cite[Sec. 3]{Chen23} for more details). Let $\mathbb{E}_g$ stand for the average over the uniform product measure with $g_k \in [1,2]$. The main consequence of this perturbation is that the replica-overlap array satisfies the strong Ghirlanda-Guerra identities under the measure $\mathbb{E}_g \mathbb{E} \langle \cdot \rangle_N^{\infty}$ and one can extract a subsequence which converges in distribution to a Ruelle cascade corresponding to some path $\pi_D^{*}$. This cascade saturates the Aizenman-Sims-Starr scheme, and hence $\pi^D_{*}$ is a Parisi measure. By Gaussian integration by parts, one obtains at least along a subsequence
     $$ \lim_{k\to \infty} \mathbb{E}_g\mathbb{E}\left[ \langle\langle \| R_{N_k} \|_2^2 \rangle\rangle_{N_k}^{D,g} \right] = \int_{0}^{1} \|\pi_{*}^D(s) \|^2_{2} \, ds,$$
     where we stress that this holds for an average of perturbed self-overlap-corrected model without further constraining. 

     Since the perturbation is not extensive (it is almost surely uniformly bounded by $C N^{3/4}$ for some constant), the exponential concentration of the self-overlap guarantees that, upon choosing a sequence $\varepsilon_N$ converging slowly enough to $0$, 
     $$  \limsup_{N \to \infty} \left| \mathbb{E}_g\mathbb{E}\left[ \langle\langle \| R_N \|_2^2 \rangle\rangle_{N}^{D,g} \right] -  \mathbb{E}_g \mathbb{E}\left[ \langle\langle \| R_N \|_2^2 \rangle\rangle_{N,\varepsilon_N, {\varrho}^D}^{(D,g,1,1)} \right] \right| = 0.  $$
     Again by Gaussian integration by parts, we obtain along a subsequence
     \begin{align*}
     &\lim_{k \to \infty} \frac1{N_k} \frac{\partial}{\partial \lambda_1}  \mathbb{E}_g \mathbb{E} \ln W^{D,g}_{N_k,\varepsilon_{N_k}}(1,1;\varrho^D) \\
         &= \frac{\beta^2}{2} \lim_{k\to \infty} \left(  \mathbb{E}_g\mathbb{E}\left[ \langle\langle \| Q_{N_k} \|_2^2 \rangle\rangle_{N_k,\varepsilon_{N_k},{\varrho}^D}^{( D,g, 1, 1)} \right] -   \mathbb{E}_g\mathbb{E}\left[ \langle\langle \| R_{N_k} \|_2^2 \rangle\rangle_{N_k,\varepsilon_{N_k}, {\varrho}^D}^{( D, g,1, 1)} \right]  \right) \\
         &= \frac{\beta^2}{2} \left( \|{\varrho}^D \|_2^2  -  \int_{0}^{1} \|\pi_{*}^D(m) \|^2_{2} \, dm  \right).
     \end{align*}
     Recalling that the perturbation does not affect the limit of the free energy and that $\widehat F^D(\lambda_1,1;\varrho)$ is convex and differentiable in $\lambda_1$, we arrive at
     $$ \frac{\partial}{\partial \lambda_1} \widehat F^D(1, 1;\varrho^D) = \frac{\beta^2}{2} \left( \|\varrho^D \|_2^2 - \int_{0}^{1} \|\pi_{*}^D(m) \|^2_{2} \, dm  \right). $$
     Employing the convexity again, we relate $ \frac{\partial}{\partial \lambda_1} \widehat F^D(1, 1;\varrho^D) = \lim_{N\to \infty} \frac{1}{N} \frac{\partial}{\partial \lambda_1} \mathbb{E}\left[\ln W_{N,\varepsilon_N}^D(1,1;\varrho^D) \right] $. By Gaussian integration by parts, one may relate the last derivative to the empirical replica overlap as in~\eqref{eq:GIP}. Hence, we can upgrade the convergence of the replica overlap to the full $N \to \infty$ limit and remove the perturbation:
     $$   \lim_{N\to \infty} \mathbb{E}\left[ \langle\langle \| R_N \|_2^2 \rangle\rangle_{N,\varepsilon_N, {\varrho}^D}^{(D, 1,1)} \right] = \int_{0}^{1} \|\pi_{*}^D(s) \|^2_{2} \, ds, $$
     which completes the proof.
  \end{proof}

 \paragraph{Acknowledgements.} This work was funded by the Deutsche Forschungsgemeinschaft (DFG, German Research Foundation) - 558731723 (CM) and  EXC-2111 -- 390814868 (SW). 
 
 \paragraph{AI Declaration.} AI has only been used to create Figure~\ref{Figure} and for typographical and language checks. In particular, all mathematical arguments are due to the authors.

\bibliography{QParisi}
\bibliographystyle{abbrv}

\end{document}